\documentclass[reprint,twocolumn,notitlepage,superscriptaddress,nofootinbib,floatfix]{revtex4-2}
 
\usepackage{float}
\usepackage{verbatim}
\usepackage[utf8]{inputenc}
\usepackage[T1]{fontenc}
\usepackage[english]{babel}
\usepackage{dsfont}
\usepackage[normalem]{ulem}
\usepackage{amsmath}
\usepackage{amssymb}
\usepackage{mathrsfs}
\usepackage{amsthm}
\usepackage{mathtools}
\usepackage{wasysym}
\usepackage{breqn}
\usepackage{comment}
\usepackage{bbm}
\usepackage{nicefrac}
\usepackage{soul}
\usepackage{stackrel}
\usepackage{braket}

\usepackage{makecell}

\usepackage{graphicx}
\usepackage{tikz}

\usepackage{hyperref}
\usepackage{framed}

\def\diracspacing{0.7pt}

\newcommand{\ketbra}[2]{| \hspace{\diracspacing} #1 \rangle \langle #2 \hspace{\diracspacing} |} 

\addto\captionsenglish{}
\addto\captionsenglish{}

\newcommand{\norm}[1]{\left\|#1\right\|}
\newcommand{\abs}[1]{\left|#1\right|}
\DeclareMathOperator{\Tr}{Tr}
\newcommand{\id}{\mathrm{id}}

\newcommand{\ppt}[2]{\frac{\partial \, #1}{\partial\, #2 }}

\newcommand{\pppt}[2]{\frac{\partial^2 \, #1}{\partial\, #2^2}} 

\newcommand{\e}[1]{\mathrm{e}^{#1}}
\newcommand{\trz}[2]{\mathrm{Tr}_{#1}\left( #2\right)}

\newcommand{\de}[1]{\left( #1 \right)}

\newcommand{\DE}[1]{\left\{#1\right\}}

\newtheorem{Thm}{Theorem}
\newtheorem{Def}{Definition}

\newtheorem{Lmm}{Lemma}

\newtheorem*{Rmk}{Remark}
\newtheorem{corollary}{Corollary}
\DeclarePairedDelimiter{\ceil}{\lceil}{\rceil}
\DeclarePairedDelimiter\floor{\lfloor}{\rfloor}

\newcommand{\proj}[1]{\ketbra{#1}{#1}}
\usepackage{natbib}
\makeatletter
\let\cat@comma@active\@empty
\makeatother

\begin{document}

\title{Entropy bounds for multiparty device-independent cryptography}

\author{Federico Grasselli}
\email[corresponding author:]{ federico.grasselli@hhu.de}
\author{Gl\'{a}ucia Murta}
\email[corresponding author:]{ glaucia.murta@hhu.de}
\author{Hermann Kampermann}
\author{Dagmar Bru\ss}
\affiliation{Institut f\"ur Theoretische Physik III, Heinrich-Heine-Universit\"at D\"usseldorf, Universit\"atsstraße 1, D-40225 D\"usseldorf, Germany}

\begin{abstract}
\noindent Multiparty quantum cryptography based on distributed entanglement will find its natural application in the upcoming quantum networks. The security of many multipartite device-independent (DI) protocols, such as DI conference key agreement, relies on bounding the von Neumann entropy of the parties' outcomes conditioned on the eavesdropper's information, given the violation of a multipartite Bell inequality. We consider three parties testing the Mermin-Ardehali-Belinskii-Klyshko (MABK) inequality and certify the privacy of their outcomes by bounding the conditional entropy of a single party's outcome and the joint conditional entropy of two parties' outcomes. From the former bound, we show that genuine multipartite entanglement is necessary to certify the privacy of a party's outcome, while the latter significantly improve previous results. We obtain the entropy bounds thanks to two general results of independent interest. The first one drastically simplifies the quantum setup of an $N$-partite Bell scenario. The second one provides an upper bound on the violation of the MABK inequality by an arbitrary $N$-qubit state, as a function of the state's parameters.
\end{abstract}
\maketitle

\section{Introduction}  \label{sec:intro}
\noindent Stimulated by data security concerns and by commercial opportunities, several companies and governments are increasingly investing resources in quantum cryptography technologies \cite{QTEU,QKDmarkets}. Those include, most prominently, quantum key distribution (QKD) \cite{BB84,Bruss1998,E91,RennerThesis,Scarani2009,Lo2014,Diamanti2016,pir2019advances} and quantum random number generation \cite{QRNG1,QRNG2}. The former enables two parties to establish a secret key (shared random bitstring), while the latter is considered the only source of genuine randomness. In the context of emerging quantum networks \cite{Eppingnet1,Eppingnet2,Pirker2018,Hahn2019,lightmatter1,lightmatter2,satellite3,WEH18}, the task of QKD can be generalized to quantum conference key agreement (CKA) \cite{Epping,Grasselli_2018,WstateProtocol,FirstCVMDI,ZSG18,OLLP19,CKAreview}. Here, $N$ parties establish a common secret key to securely broadcast messages within their network, as proved by recent CKA experiments \cite{CKAexperiment,anonymousCKA}. However, it is challenging to ensure that the assumptions on the implementation of these cryptographic tasks are met in practice, hence jeopardizing their security.

This led to the development of device-independent (DI) cryptographic protocols, whose security holds independently of the actual functioning of the quantum devices and is based on the observation of a Bell inequality violation \cite{Bellineq}. Such protocols include DIQKD  \cite{YaoMayers,Acin2006,AcinBrunner2007,PironioAcin2009,Masanes2011,VidickDIQKD,DIQKDsupraquantum,Arnon-Friedman2018,HolzRepeaters} and DICKA \cite{SG01,SG_pra_01,JeremyParityCHSH,Holz2019DICKA} schemes, where a secret key is shared by two or more parties, respectively. Otherwise, with DI randomness generation (DIRG) protocols \cite{ColbeckThesis2006,Pironio2010,Colbeck2011,securityDIrandomness1,securityDIrandomness2,securityDIrandomness3,Woodhead2018} one can generate random bitstrings that are guaranteed to be private thanks to a Bell violation.

A crucial aspect of any DI protocol is the ability to carefully estimate, from the observed Bell violation, the minimum amount of uncertainty that a potential eavesdropper, Eve, could have about the protocol's outputs. Indeed, this quantity determines the length of the secret random bitstring that can be distilled from the protocol's outputs. Eve's uncertainty is quantified by appropriate conditional von Neumann entropies \cite{RennerThesis,AcinBrunner2007,PironioAcin2009,Arnon-Friedman2018} of the effective quantum state shared by the parties in a generic round of the protocol. The goal is to minimize the entropy over all the possible states yielding the observed Bell inequality violation.

This task can be carried out numerically, however the available techniques \cite{NPA1,NPA2,Nieto2014,Bancal2014} focus on minimizing a lower bound on the von Neumann entropy, namely the min-entropy \cite{RennerThesis}, thus producing sub-optimal results. Here we follow an analytical approach that reduces the degrees of freedom of the generic state shared by the parties without loss of generality, thereby allowing a direct minimization of the conditional von Neumann entropy. This can result in a tight bound of Eve's uncertainty, hence in longer secret bitstrings and higher noise tolerance for the DI protocol. To the best of our knowledge, such an analytical procedure has only been developed by Pironio et al. \cite{PironioAcin2009} for the case of two parties testing the Clauser-Horne-Shimony-Holt (CHSH) inequality \cite{CHSH}.

In this work we develop an analytical procedure applicable to multipartite DI scenarios. Specifically, we consider $N$ parties, each equipped with two measurement settings with binary outcomes, testing a generic full-correlator Bell inequality --i.e. an inequality where each correlator involves every party \cite{WernerWolf}. 
Remarkably and without loss of generality, we reduce the generic state shared by the $N$ parties (in one protocol round) to a mere $N$-qubit state almost diagonal in the GHZ basis. Notably, we recover the result of Pironio et al. when $N=2$.

We then focus on the Mermin-Ardehali-Belinskii-Klyshko (MABK) inequality \cite{Mermin,Ardehali,BK93} and derive an analytical bound on the maximal violation of the MABK inequality yielded by rank-one projective measurements on an given $N$-qubit state. This is a result of independent interest, which generalizes the bound for the bipartite case of \cite{HHH95} and constitutes, to our knowledge, the first of its kind valid for an arbitrary $N$-qubit state.

\begin{figure}[t]
	\centering
	\includegraphics[width=0.9\linewidth,keepaspectratio]{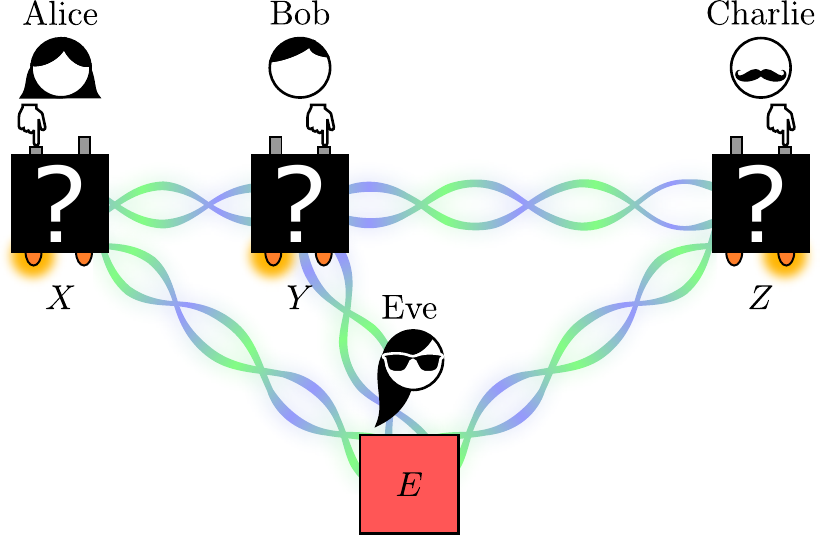}
	\caption{Alice, Bob and Charlie generate device-independent randomness from the input/output correlations of their unknown quantum devices by testing the MABK inequality. Each device is equipped with two inputs and two outputs. The eavesdropper Eve might hold a quantum memory $E$ entangled with the parties' devices and use it to guess their outcomes $X$, $Y$ or $Z$. We quantify Eve's uncertainty on Alice's outcome $X$ by computing the conditional von Neumann entropy $H(X|E)$. Additionally, we assume that Alice and Bob are co-located and collaborate to generate global randomness from their outcomes $X$ and $Y$. We quantify Eve's joint uncertainty on their outcomes by computing $H(XY|E)$.}
	\label{3DI-scenario}
\end{figure}
Our results on the state reduction in a multipartite DI scenario and on the MABK violation upper bound, allow us to quantify Eve's uncertainty about the parties' outcomes when three parties, Alice, Bob and Charlie, test the MABK inequality  (see Fig.~\ref{3DI-scenario}). Specifically, we obtain an analytical lower bound on the von Neumann entropy of Alice's outcome conditioned on Eve's information. We compare our bound with a numerical estimation of the corresponding tight entropy bound and with previous results. We additionally prove that genuine multipartite entanglement is a necessary resource to certify the privacy of a party's outcome in any $N$-party MABK scenario. The bound can find potential application in the security of DIRG based on multipartite nonlocality. We also provide a heuristic argument for which full-correlator Bell inequalities, such as the MABK inequality, are unlikely to be employed in any DICKA protocol.

In the same tripartite scenario of Fig.~\ref{3DI-scenario}, we derive a lower bound on the joint conditional von Neumann entropy of Alice and Bob's outcomes, which substantially improves the result derived in \cite{Woodhead2018}, where the authors bound the corresponding min-entropy. The derived bound can be employed in proving the security of DI global randomness generation schemes.

\section{Reduction of the $N$-party quantum state} \label{sec:reduction}
\noindent Let $\mathrm{Alice}_1$, $\mathrm{Alice}_2$,\dots, $\mathrm{Alice}_{N}$ be $N$ parties who want to generate private randomness (random bitstrings) from the outcomes of their uncharacterized quantum devices. In order to certify device-independently that the generated randomness is private, they test a generic full-correlator Bell inequality \cite{WernerWolf} where each party can choose among two measurement settings with binary outcomes on her respective device. We identify this as the $(N,2,2)$ DI scenario.

An eavesdropper, Eve, wants to learn the randomness generated by some of the parties. We consider the most adversarial scenario where Eve herself may distribute arbitrary quantum states to the parties' devices, which could  be forged by Eve. The device of each party measures the binary observable $A^{(i)}_x$ ($i=1,\dots,N$) on the received quantum state, according to $\mathrm{Alice}_i$'s measurement input ($x=0,1$). Note that the observables measured by each device may be pre-established by Eve. 

The tested Bell inequality is a linear combination of full-correlators of the form:
\begin{equation}
    \left< A^{(1)}_{x_1}\cdots A^{(N)}_{x_N}\right>  \label{full-corr}.
\end{equation}
From the observed Bell violation, the parties can quantify Eve's uncertainty on their random bitstrings by computing an appropriate conditional von Neumann entropy. With this result, a party could enhance the privacy of her bitstring (with privacy amplification \cite{RennerThesis}) and use it for various cryptographic tasks (e.g. DICKA or DIRG).

Here we present a fundamental result that enables a direct computation of the conditional von Neumann entropy of interest. Indeed, our result drastically simplifies, without loss of generality, the general quantum setup described above. For instance, the generic quantum state distributed to the parties is reduced to an $N$-qubit state (almost) diagonal in the GHZ basis.

The GHZ basis \cite{Epping} for the Hilbert space of $N$ qubits is defined as follows.
\begin{Def}\label{def:GHZbasis} The GHZ basis for the set of $N$-qubit states is composed of the following $2^N$ states:
\begin{align}
    \ket{\psi_{\sigma,\vec{u}}}=\frac{1}{\sqrt{2}}\left(\ket{0}\ket{\vec{u}}+(-1)^{\sigma}\ket{1}\ket{\vec{\bar{u}}}\right), \label{GHZbasis}
\end{align}
where $\sigma \in \{0,1\}$ while $\vec{u}\in \{0,1\}^{N-1}$ and $\vec{\bar{u}}=\vec{1}\oplus\vec{u}$ are $(N-1)$-bit strings. In particular, for a three-qubit state, the GHZ basis reads:
\begin{equation}
    \ket{\psi_{i,j,k}}=\frac{1}{\sqrt{2}} \left(\ket{0,j,k}+(-1)^i\ket{1,\bar{j},\bar{k}}\right) \quad i,j,k\in 
    \{0,1\}, \label{GHZijk}
\end{equation}
where the bar over a bit indicates its negation.
\end{Def}
We now formally state the first major result of this work, the proof of which is reported in Sec.~\ref{appendix:theorem1}.
\begin{Thm} \label{thm:reduction-to-GHZ}
Let $N$ parties test an $(N,2,2)$ full-correlator Bell inequality in order to certify the privacy of their outcomes. It is not restrictive to assume that, in each round, Eve distributes a mixture $\sum_\alpha p_\alpha \rho_\alpha$ of $N$-qubit states $\rho_\alpha$, together with a flag $\ket{\alpha}$ (known to her) which determines the measurements performed on $\rho_\alpha$ given the parties' inputs. Without loss of generality, the measurements performed by each device on $\rho_\alpha$ are rank-one binary projective measurements in the $(x,y)$-plane of the Bloch sphere. Moreover, each state $\rho_\alpha$ is diagonal in the GHZ basis, except for some purely imaginary off-diagonal terms:
\begin{align}
    \rho_{\alpha}= &\sum_{\vec{u}\in\{0,1\}^{N-1}} \left[\lambda^{\alpha}_{0 \vec{u}} \ketbra{\psi_{0,\vec{u}}}{\psi_{0,\vec{u}}} + \lambda^{\alpha}_{1 \vec{u}} \ketbra{\psi_{1,\vec{u}}}{\psi_{1,\vec{u}}} \right. \nonumber\\
    &\left.+ \mathbbm{i}s^{\alpha}_{\vec{u}} \left(\ketbra{\psi_{0,\vec{u}}}{\psi_{1,\vec{u}}} - \ketbra{\psi_{1,\vec{u}}}{\psi_{0,\vec{u}}} \right) \right]  \label{almostGHZ},
\end{align}
Finally, $N$ arbitrary off-diagonal terms $s^{\alpha}_{\vec{u}}$ can be assumed to be zero and the corresponding diagonal elements $(\lambda^{\alpha}_{0 \vec{u}},\lambda^{\alpha}_{1 \vec{u}})$ can be arbitrarily ordered (e.g. $\lambda^{\alpha}_{0 \vec{u}} \geq \lambda^{\alpha}_{1 \vec{u}}$).
\end{Thm}

In the following we focus our analysis on a given state $\rho_\alpha$. Hence, for ease of notation we drop the symbol $\alpha$ in the parameters related to the state $\rho_\alpha$ (e.g. $\lambda_{0,\vec{u}}^{\alpha}$ and $s_{\vec{u}}^{\alpha}$) when there is no ambiguity.

Note that, for $N=2$, we recover the result of \cite{PironioAcin2009}. By applying Theorem~\ref{thm:reduction-to-GHZ} to the case of $N=3$ parties, it is not restrictive to assume that they share a mixture of states $\rho_\alpha$, with the following matrix representation in the GHZ basis:
\begin{align}
    \rho_{\alpha}=  
  \begin{bmatrix}
    \lambda_{000} & 0 & 0 & 0 & 0 & 0 & 0 & 0 \\
    0 & \lambda_{100} & 0 & 0 & 0 & 0 & 0 & 0  \\
    0 & 0 & \lambda_{001} & 0 & 0 & 0 & 0 & 0  \\
    0 & 0 & 0 & \lambda_{101} & 0 & 0 & 0 & 0  \\
    0 & 0 & 0 & 0 & \lambda_{010} & 0 & 0 & 0 \\
    0 & 0 & 0 & 0 & 0 & \lambda_{110} & 0 & 0 \\
    0 & 0 & 0 & 0 & 0 & 0 & \lambda_{011} & \mathbbm{i}s \\
    0 & 0 & 0 & 0 & 0 & 0 & -\mathbbm{i}s & \lambda_{111}
  \end{bmatrix}  \label{almostGHZ3}.
\end{align}
The eigenvalues of \eqref{almostGHZ3} are given by:
\begin{align}
    \rho_{ijk} &= \lambda_{ijk}  \quad (j,k)\neq (1,1) \nonumber\\
    \rho_{i11} &= \frac{\lambda_{011}+\lambda_{111}+(-1)^i \sqrt{(\lambda_{011}-\lambda_{111})^2+4s^2}}{2}.  \label{eigenval}
\end{align}

In order to accurately quantify Eve's uncertainty on the parties' outcomes via conditional von Neumann entropies, one also needs an analytical expression for the maximal violation of the tested Bell inequality. In Sec.~\ref{sec:MABK-bound} we establish such a result for the MABK inequality.

\section{Upper bound on MABK violation} \label{sec:MABK-bound}

\noindent The MABK inequality \cite{Mermin,Ardehali,BK93} is one possible generalization of the CHSH inequality \cite{CHSH} and is derived on the following MABK operator.
\begin{Def} \label{def:NMABK}
The MABK operator $M_N$ is defined by recursion \cite{MABKrecursion,JeremyMABK}:
\begin{align}
    M_2 &= G_{\mathrm{CHSH}} (A_0^{(1)},A_1^{(1)},A_0^{(2)},A_1^{(2)}) \nonumber\\
    &\equiv  A_0^{(1)}\otimes A_0^{(2)} +A_0^{(1)}\otimes A_1^{(2)} +A_1^{(1)}\otimes A_0^{(2)}  \nonumber\\
    &- A_1^{(1)}\otimes A_1^{(2)} \nonumber \\
    M_N &= \frac{1}{2} G_{\mathrm{CHSH}} (M_{N-1},\overline{M_{N-1}},A_0^{(N)},A_1^{(N)})  \label{MN},
\end{align}
where $A_{x_i}^{(i)}$ for $x_i=0,1$ are the binary observables of $\mathrm{Alice}_i$ $($each observable satisfies: $(A_{x_i}^{(i)})^\dag=A_{x_i}^{(i)}$ and $(A_{x_i}^{(i)})^2 \leq \id$, where ``$\id$'' is the identity operator$)$ and where $\overline{M_{l}}$ is the operator obtained from $M_l$ by replacing every observable $A_{x_i}^{(i)}$ with $A_{1-x_i}^{(i)}$. For $N=3$, the MABK operator reads:
\begin{align}
    M_3= &A_0 \otimes B_0 \otimes C_1 +  A_0 \otimes B_1 \otimes C_0 \nonumber\\
    &+  A_1 \otimes B_0 \otimes C_0 -  A_1 \otimes B_1 \otimes C_1 \label{3-MABK}
\end{align}
where $A_x$, $B_y$ and $C_z$ are Alice's, Bob's and Charlie's observables, respectively.
\end{Def}

The $N$-partite MABK inequality reads \cite{MABKrecursion,JeremyMABK}:
\begin{equation}
    \left<M_N\right> = \Tr[M_N \rho] \leq  
    \begin{cases}
      2, & \text{classical bound} \\
      2^{N/2}, & \text{GME threshold} \\
      2^{(N+1)/2} & \text{quantum bound}
    \end{cases}
      \label{N-MABK}
\end{equation}
where $M_N$ is the MABK operator and a violation of the GME threshold implies that the parties share a genuine multipartite entangled (GME) state.

The second major result of this work is an upper bound on the maximal MABK violation obtained when $N$ parties share an $N$-qubit state and perform rank-one projective measurements on the respective qubits. The bound is state-dependent and tight on certain classes of states (proof and tightness conditions in Sec.~\ref{appendix:theorem2}). This is, to the best of our knowledge, the first bound of such kind for an $N$-partite Bell inequality. Recently, the authors in \cite{SS19} derived a similar bound in the $N=3$ case. Our bound is tight on a larger set of states (discussion in Sec.~\ref{appendix:theorem2}) and is valid for general $N$.
\begin{Thm} \label{theorem-Nparties}
The maximum violation $\mathcal{M}_\rho$ of the \mbox{$N$-partite} MABK inequality \eqref{N-MABK}, attained with rank-one projective measurements on an $N$-qubit state $\rho$, satisfies 
\begin{align}
    \mathcal{M}_\rho\leq 2\sqrt{t_0+t_1}   \label{NMABKviolationbound}
\end{align}
where $t_0$ and $t_1$ are the largest and second-to-the-largest eigenvalues of the matrix $T_{\rho} T_{\rho}^T$, where $T_\rho$ is the correlation matrix of $\rho$.
\end{Thm}
We define the correlation matrix of an $N$-qubit state as follows.
\begin{Def}\label{def:corrmatrix}
The correlation matrix of an $N$-qubit state $\rho$, $T_{\rho}$, is a square or rectangular matrix defined by the elements $[T_{\rho}]_{ij}=\Tr[\rho \sigma_{\nu_1}\otimes \dots \otimes \sigma_{\nu_N}]$ such that:
\begin{align}
	i &= 1+ \sum_{k=1}^{\ceil{N/2}} \,3^{\ceil{N/2}-k}(\nu_k-1)  \nonumber\\
	j &= 1+ \sum_{k=\ceil{N/2}+1}^{N} \, 3^{N-k}(\nu_k-1)  
\end{align}
where $\nu_1,\dots,\nu_N \in\{1,2,3\}$, $\sigma_{\nu_i}$are the Pauli operators and $\ceil{x}$ returns the smallest integer greater or equal to $x$.
\end{Def}

\begin{Rmk} \label{rmk:identity-multipartiteBell}
We remark that the most general measurements to be considered in computing the maximal MABK violation are projective measurements defined by observables $(A^{(i)}_{x_i})^2=\id$ {\normalfont\cite{WernerWolf}}, since POVMs never provide higher violations {\normalfont\cite{POVM-not-optimal,OptimalProjective-probabilityBellineq}}. Such measurements on qubits reduce to either {\normalfont (i)} rank-one projective measurements given by $A^{(i)}_{x_i}=\vec{a}^{\,i}_{x_i} \cdot \vec{\sigma}$ with unit vectors $\vec{a}^{\,i}_{x_i}\in\mathbbm{R}^3$ and where $\vec{\sigma}=(X,Y,Z)^T$ is the vector of Pauli operators, or {\normalfont (ii)} rank-two projective measurements given by the identity $A^{(i)}_{x_i}=\pm\id$, i.e. measurements with a fixed outcome. While for $N=2$ parties the identity does not lead to any violation {\normalfont\cite{HHH95}} and the optimal measurements are described by case {\normalfont(i)}, in a multipartite scenario case {\normalfont(ii)} cannot be ignored.
\end{Rmk}

For instance, if $N=3$ parties share the state $\id/2 \otimes \ketbra{\psi_{00}}{\psi_{00}}$ (with $\ket{\psi_{00}}$ given in Definition~\ref{def:GHZbasis}), an MABK violation of $2\sqrt{2}$ is achieved if the first party measures $A^{(1)}_{0}=A^{(1)}_{1}=\id$, whereas no violation is obtained if her measurements are restricted to $A^{(i)}_{x_i}=\vec{a}^{\,i}_{x_i} \cdot \vec{\sigma}$.

We point out that previous works \cite{neglects-identity-Bellviolation1,neglects-identity-Bellviolation2,neglects-identity-Bellviolation3,SS19} addressing the violation of multipartite Bell inequalities achieved by a given multi-qubit state have neglected case (ii) and only considered case (i). By applying the above example, we stress that the results of \cite{neglects-identity-Bellviolation1,neglects-identity-Bellviolation2,neglects-identity-Bellviolation3,SS19} characterizing Bell violations yielded by multi-qubit states are, in fact, less general than claimed.

Nevertheless, for states whose maximal violation is above the GME threshold, the bound we provide in Theorem~\ref{theorem-Nparties} is general and holds independently of the parties' measurements. Indeed, measuring the identity cannot lead to violations above the GME threshold and thus case (i) is already the most general.

By applying Theorem~\ref{theorem-Nparties} to the state $\rho_\alpha$ in \eqref{almostGHZ3}, we obtain an upper bound on the maximal MABK violation $\mathcal{M}_\alpha$ achievable on $\rho_\alpha$ with rank-one projective measurements.

\begin{corollary}\label{cor:almostGHZMermin}
For a tripartite state $\rho_\alpha$ of the form given in \eqref{almostGHZ3}, the maximal violation $\mathcal{M}_{\alpha}$ of the MABK inequality achieved with rank-one projective measurements satisfies:
\begin{align}
    \mathcal{M}_\alpha\leq \mathcal{M}^{\uparrow}_\alpha=4\sqrt{ \sum_{j,k=0}^1 (\rho_{0jk} -\rho_{1jk})^2},   \label{MABKupp-almostGHZ3}
\end{align}
where $\{\rho_{ijk}\}$ are the eigenvalues of the state $\rho_\alpha$, as specified in \eqref{eigenval}.
\end{corollary}
In Sec.~\ref{appendix:theorem2}, we provide the tightness conditions \eqref{tightnessconditions} for which the upper bound in \eqref{MABKupp-almostGHZ3} is achieved.

\section{One-outcome conditional entropy bound} \label{sec:H(X|E)}

\noindent Consider the $(3,2,2)$ DI scenario of Fig.~\ref{3DI-scenario}. Alice, Bob and Charlie test the tripartite MABK inequality in order to quantify Eve's uncertainty on the generic outcome $X$ of one of Alice's observables, by computing the conditional von Neumann entropy $H(X|E)$. We emphasize that, in a DIRG protocol, the entropy $H(X|E)$ determines the asymptotic rate of secret random bits extracted by applying privacy amplification \cite{RennerThesis} on Alice's $X$ outcomes \cite{EAT,DIRE-framework}. Similarly, in DICKA the secret key rate is determined by $H(X|E)$ decreased by the amount of classical information disclosed by the parties in the other steps of the protocol \cite{EAT,JeremyMABK,JeremyParityCHSH}.

We derive an analytical lower bound on $H(X|E)$ as a function of the observed MABK violation. Theorem~\ref{thm:reduction-to-GHZ} guarantees that we can restrict the computation of the conditional entropy $H(X|E_{\mathrm{tot}})$ over a mixture of states $\rho_\alpha$ of the form \eqref{almostGHZ3} and to rank-one projective measurements performed by the parties. We emphasize that the total information $E_{\mathrm{tot}}=E \Xi$ available to Eve includes the knowledge of the flag $\Xi$ which carries the value of $\alpha$ (see Sec.~\ref{appendix:theorem1}). The goal is to lower bound the conditional entropy $H(X|E_{\mathrm{tot}})$ with a function $F$ of the observed MABK violation $m$.

Thanks to Theorem~\ref{thm:reduction-to-GHZ}, we can express the conditional entropy $H(X|E_{\mathrm{tot}})$ as follows:
\begin{align}
    H(X|E_{\mathrm{tot}}) &=\sum_{\alpha}p_{\alpha}H(X|E \Xi=\alpha) \nonumber\\
    &=\sum_{\alpha}p_{\alpha}H(X|E)_{\rho_{\alpha}}, \label{Alice-cond-entropy}
\end{align}
as a matter of fact the state on which $H(X|E_{\mathrm{tot}})$ is computed is a classical-quantum state (see Eq.~\eqref{distributed-state}). At the same time, the observed violation $m$ can be expressed as:
\begin{align}
    m &= \sum_{\alpha}p_{\alpha}
    m_\alpha \label{Bellobs-smaller}.
\end{align}
In \eqref{Alice-cond-entropy}, the entropy $H(X|E)_{\rho_{\alpha}}$ is the conditional entropy of Alice's outcome given that Eve distributed the state $\rho_\alpha$, while $p_\alpha$ is the probability distribution of the mixture prepared by Eve. In \eqref{Bellobs-smaller}, $m_\alpha$ is the violation that the parties would observe had they shared the state $\rho_{\alpha}$ in every round of the protocol and performed the corresponding rank-one projective measurements.

We then aim at lower bounding $H(X|E)_{\rho_{\alpha}}$ with a convex function $F$ of the MABK violation $m_\alpha$:
\begin{equation}
    H(X|E)_{\rho_{\alpha}} \geq F(m_\alpha). \label{lower_bound-condH}
\end{equation}
Indeed, by combining \eqref{Alice-cond-entropy}, \eqref{Bellobs-smaller}, \eqref{lower_bound-condH} and the convexity of $F$, one can obtain the desired lower bound on $H(X|E_{\mathrm{tot}})$ as a function of the observed violation $m$:
\begin{equation}
    H(X|E_{\mathrm{tot}}) \geq F(m) \label{final-result}.
\end{equation}
The bound is tight if, for any given MABK violation $m$, there exist a quantum state and a set of measurements that achieve violation $m$ and whose outcome's conditional entropy is exactly given by $F(m)$. We now obtain the function $F$ by minimizing $H(X|E)_{\rho_\alpha}$ over all the states $\rho_\alpha$ yielding a violation $m_\alpha$.

The eigenvectors of the state $\rho_\alpha$, corresponding to the eigenvalues in \eqref{eigenval}, read:
\begin{align}
    \ket{\rho_{ijk}} &= \ket{\psi_{i,j,k}} \quad (j,k)\neq (1,1)  \nonumber\\
    \ket{\rho_{011}}&=\cos(t)\ket{\psi_{0,1,1}}-\mathbbm{i}\sin(t) \ket{\psi_{1,1,1}} \nonumber\\
    \ket{\rho_{111}}&=\cos(t)\ket{\psi_{1,1,1}}-\mathbbm{i}\sin(t) \ket{\psi_{0,1,1}}\label{eigenvec},
\end{align}
where the parameter $t$ is defined as:
\begin{equation}
    t = \arctan\frac{2 s }{\lambda_{011}-\lambda_{111}+ \sqrt{(\lambda_{011}-\lambda_{111})^2 + 4 s^2}} \label{q}.
\end{equation}
By combining the freedom in ordering the diagonal elements $\lambda_{ijk}$ of $\rho_\alpha$ (c.f. Theorem~\ref{thm:reduction-to-GHZ}) with the definition of the eigenvalues $\rho_{ijk}$ in \eqref{eigenval}, one can impose the following constraints on the eigenvalues:
\begin{equation}
    \rho_{0jk} \geq \rho_{1jk} \quad \forall\, j,k \label{ordered-rhoijk}.
\end{equation}
The entropy $H(X|E)_{\rho_\alpha}$ is computed on the classical-quantum state:
\begin{equation}
    \rho^\alpha_{XE} = (\mathcal{E}_X \otimes \mathrm{id}_E) \Tr_{BC}[\ketbra{\phi^\alpha_{ABCE}}{\phi^\alpha_{ABCE}}],  \label{rhoXE}
\end{equation}
where $\ket{\phi^\alpha_{ABCE}}$ is a purification of $\rho_\alpha$ (Eve holds the purifying system $E$), while $\mathcal{E}_X$ represents one of the two projective measurements of Alice, defined by the eigenvectors:
\begin{equation}
    \ket{a}_X=\frac{1}{\sqrt{2}} (\ket{0}+ (-1)^a e^{\mathbbm{i}\varphi_X} \ket{1}) \quad a\in \{0,1\},
\end{equation}
where $\varphi_X\in [0,2\pi)$ identifies the measurement direction in the $(x,y)$-plane of the Bloch sphere. For definiteness, we choose $\varphi_X$ to be the measurement direction of Alice's observable $A_0$: $\varphi_X=\varphi_{A_0}$. Hence, we are deriving a lower bound on  $H(X_{A_0}|E_{\mathrm{tot}})$, where $X_{A_0}$ is the outcome of Alice's observable $A_0$.

The entropy minimization can be simplified if, instead of minimizing over the matrix elements $\{\lambda_{ijk}\}$ and $s$ of $\rho_\alpha$, one minimizes over its eigenvalues $\{\rho_{ijk}\}$ and over $t$. This change of variables is legitimized by the bijective map linking the two sets of parameters, defined by the relations \eqref{eigenval} and \eqref{q}.

The solution of the following optimization problem yields a tight lower bound  on $H(X_{A_0}|E)_{\rho_\alpha}$:
\begin{align}
    &\min_{\{\rho_{ijk},t,\vec{\varphi}\}} H(X_{A_0}|E)_{\rho_\alpha} (\rho_{ijk},t,\varphi_{A_0}) \nonumber\\
    &\mbox{sub. to}\,\,\braket{M_3}_{\rho_\alpha}(\rho_{ijk},t,\vec{\varphi}) \geq m_\alpha \,;\, \rho_{0jk}\geq\rho_{1jk} \,;\, \nonumber\\
    &\quad\quad\quad{\textstyle\sum_{ijk}}\,\rho_{ijk}=1 \,;\,\rho_{ijk}\geq 0, \label{tight-optimization-Alice}
\end{align}
where $\vec{\varphi}:=(\varphi_{A_0},\varphi_{A_1},\varphi_{B_0},\varphi_{B_1},\varphi_{C_0},\varphi_{C_1})$  contains the measurement directions identified by the observables $A_0,A_1,B_0,B_1,C_0$ and $C_1$.
Notably, due to the symmetries of the MABK inequality, all the tight lower bounds on $H(X_{A_i}|E)_{\rho_\alpha}$, $H(Y_{B_i}|E)_{\rho_\alpha}$ and $H(Z_{C_i}|E)_{\rho_\alpha}$ (for $i=0,1$) coincide.
Thus, solving \eqref{tight-optimization-Alice} yields a tight lower bound on the conditional entropy $H(X|E)_{\rho_\alpha}$ of any single party's outcome $X$.

We drastically simplify the optimization problem in \eqref{tight-optimization-Alice}, by replacing the MABK expectation value $\braket{M_3}_{\rho_\alpha}$ with its upper bound $\mathcal{M}^\uparrow_{\alpha}$ derived in \eqref{MABKupp-almostGHZ3}. Indeed, this allows us to independently minimize $H(X|E)_{\rho_\alpha}$ over $t$ and $\varphi_{A_0}$ without affecting the MABK violation. The resulting conditional entropy is minimized by $t=\varphi_{A_0}=0$ and reads:
\begin{align}
    H(X|E)_{\rho_\alpha}&(\rho_{ijk},t=0,\varphi_{A_0}=0)= \nonumber\\
     &1 - H(\{\rho_{ijk}\}) +H(\{\rho_{ijk}+\rho_{i\bar{j}\bar{k}}\}) \label{cond-entropy-alpha},
\end{align}
where the Shannon entropy of a probability distribution $\{p_i\}_i$ is defined as $H(\{p_i\})=\sum_i -p_i \log_2 p_i$.

\begin{figure}[tb]
	\centering
	\includegraphics[width=0.9\linewidth,keepaspectratio]{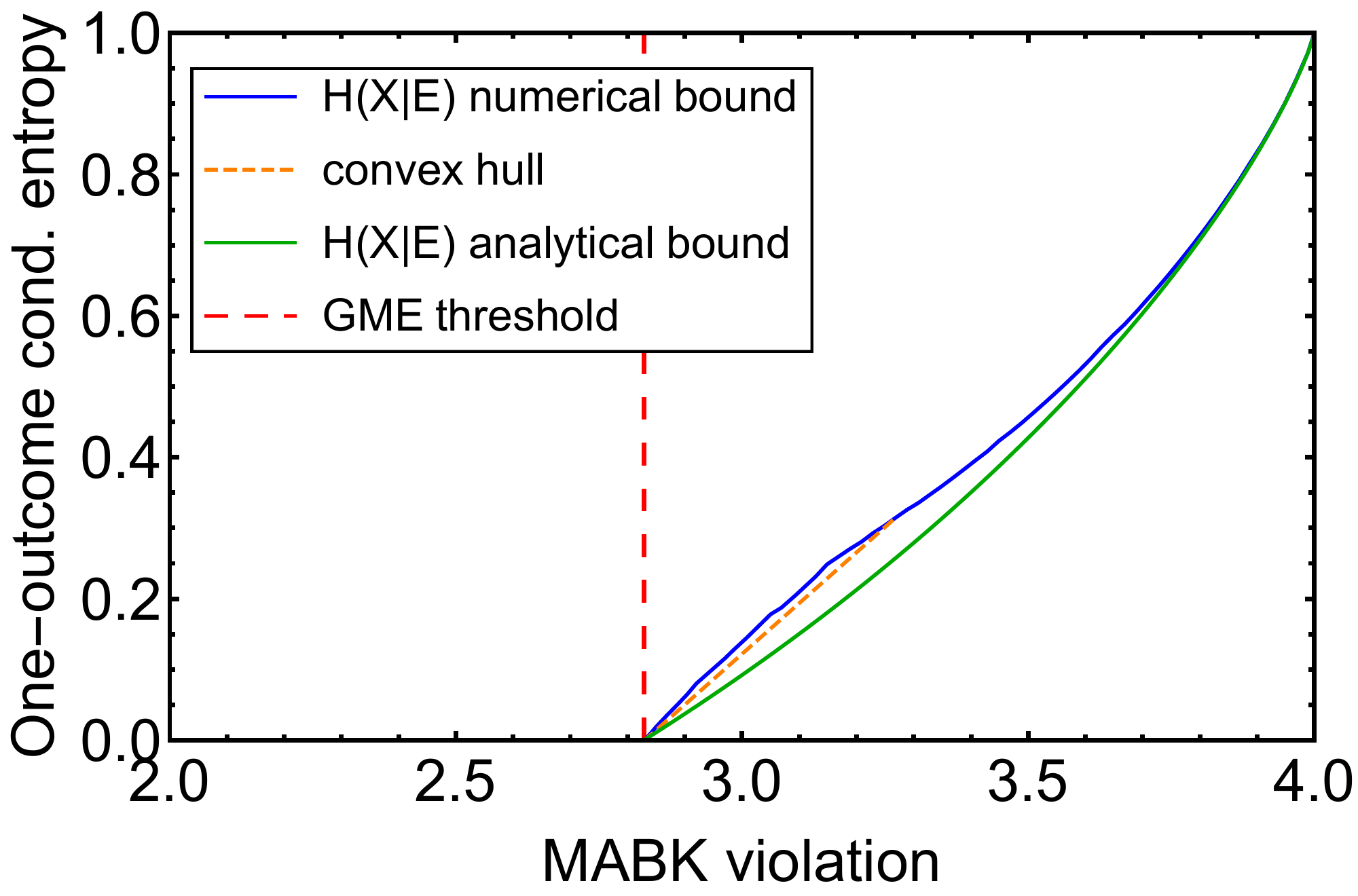}
	\caption{Analytical lower bound on the conditional von Neumann entropy $H(X|E_{\mathrm{tot}})$ as a function of the MABK inequality violation (green line, Eq.~\eqref{H(X|E)-final-bound}) observed by three parties. We compare it to the numerical optimization of \eqref{tight-optimization-Alice} (blue line), whose convex hull (dashed orange) yields an upper limit on the lowest value of $H(X|E_{\mathrm{tot}})$. We notice that Eve has no uncertainty on Alice's outcome $X$ for violations below the genuine multipartite entanglement (GME) threshold.}
	\label{plot-Hlower}
\end{figure}

We are thus left to solve the following optimization problem:
\begin{align}
    &\min_{\{\rho_{ijk}\}} 1 - H(\{\rho_{ijk}\}) +H(\{\rho_{ijk}+\rho_{i\bar{j}\bar{k}}\}) \nonumber\\
    &\mbox{sub. to}\,\, \mathcal{M}^\uparrow_{\alpha} (\rho_{ijk})\geq m_\alpha \,;\, \rho_{0jk}\geq\rho_{1jk} \,;\, \nonumber\\
    &\quad\quad\quad\,{\textstyle\sum_{ijk}}\,\rho_{ijk}=1 \,;\,\rho_{ijk}\geq 0, \label{optimization-Alice}
\end{align}
whose solution is a lower bound on the solution of the original optimization problem \eqref{tight-optimization-Alice}.
We analytically solve \eqref{optimization-Alice} and provide the complete proof in Appendix~\ref{sec:tau-proof}.

Importantly, the following family of states solves \eqref{optimization-Alice} for every value of the violation $m_\alpha$:
\begin{equation}
    \tau(\nu_m)= \nu_m \ketbra{\psi_{0,0,0}}{\psi_{0,0,0}}+(1-\nu_m)\ketbra{\psi_{0,1,1}}{\psi_{0,1,1}} , \label{tau-family}
\end{equation}
where the parameter $\nu_m$ is fixed by the violation $m_\alpha$ by:
\begin{equation}
    m_\alpha= \mathcal{M}^\uparrow_\tau (\nu_m) = 4 \sqrt{2 \nu_m^2 -2\nu_m +1} \label{tau-MABK},
\end{equation}
where we used \eqref{MABKupp-almostGHZ3} in the second equality.
The lower bound on the conditional entropy $H(X|E)_{\rho_\alpha}$ is thus given by the entropy of the states in \eqref{tau-family}:
\begin{equation}
    H(X|E)_{\rho_\alpha} \geq F(m_\alpha) := H(X|E)_{\tau (\nu_m)} \label{F-bound}.
\end{equation}
The entropy of the states in \eqref{tau-family} is easily computed from \eqref{cond-entropy-alpha} and can be expressed in terms of the violation $m_{\alpha}$ by reverting \eqref{tau-MABK}. We obtain:
\begin{equation}
    F (m_\alpha) = 1-h\left(\frac{1}{2} +\frac{1}{2}\sqrt{\frac{m_\alpha^2}{8} -1}\right) \label{tight-H},
\end{equation}
where $h(p)=-p\log_2 p - (1-p)\log_2 (1-p)$ is the binary entropy. Finally, the lower bound \eqref{tight-H} is a convex function, hence we can employ it in \eqref{final-result} and obtain the desired lower bound on $H(X|E_{\mathrm{tot}})$ as a function of the observed MABK violation:
\begin{equation}
    H(X|E_{\mathrm{tot}}) \geq 1-h\left(\frac{1}{2} +\frac{1}{2}\sqrt{\frac{m^2}{8} -1}\right) \label{H(X|E)-final-bound}.
\end{equation}

In Fig.~\ref{plot-Hlower} we plot the lower bound on the conditional entropy derived in \eqref{H(X|E)-final-bound}, as well as a numerical optimization of \eqref{tight-optimization-Alice}, which yields an upper bound on the minimal value of $H(X|E)_{\rho_\alpha}$. We can conclude that the tight lower bound on $H(X|E_{\mathrm{tot}})$ lies in the plot region delimited by the convex hull of the numerical curve (the bound in \eqref{lower_bound-condH} must be convex) and our analytical lower bound.

From Fig.~\ref{plot-Hlower}, we deduce that our analytical lower bound on $H(X|E_{\mathrm{tot}})$ leaves little room for improvement (compared to the ideal tight bound) and that it is actually tight up to the GME threshold of $m=2\sqrt{2}$. We prove this by showing that the state $\tau(1/2)$, which yields the analytical bound at $m=2\sqrt{2}$, is also an optimal solution of the original optimization problem \eqref{tight-optimization-Alice}. Indeed, when $m=2\sqrt{2}$, the tightness conditions of the MABK upper bound \eqref{tightnessconditions} applied to $\tau(1/2)$ are verified for $\varphi_{A_0}=\varphi_{A_1}=0$. In other words, there exist observables that Alice, Bob and Charlie can measure on $\tau(1/2)$ such that $\braket{M_3}_{\tau(1/2)}=2\sqrt{2}$. In particular, Alice's optimal observables are the Pauli $X$. Under these conditions, the entropy in \eqref{tight-optimization-Alice} reads: $H(X|E)_{\tau(1/2)}(\varphi_{A_0}=0)=0$, which must be the solution of \eqref{tight-optimization-Alice} since in general it holds $H(X|E)_{\rho_\alpha}\geq 0$. Thus the lower bound \eqref{H(X|E)-final-bound} is tight for $m=2\sqrt{2}$  and is equal to zero.

By combining this with the fact that the tight lower bound on $H(X|E_{\mathrm{tot}})$ is monotonically non-decreasing in $m$ by construction (see \eqref{tight-optimization-Alice}), we deduce that the conditional entropy of a party's outcome is zero for every violation below the GME threshold of $2\sqrt{2}$. Hence, GME states are a necessary resource to guarantee private randomness of a party's outcome in a tripartite MABK scenario.

Notably, the claim on the necessity of GME can be generalized to an $N$-party MABK scenario. Consider the following family of states that generalizes \eqref{tau-family} to $N$ parties:
\begin{equation}
    \tau(\nu)= \nu \ketbra{\psi_{0,\Vec{0}}}{\psi_{0,\vec{0}}}+(1-\nu)\ketbra{\psi_{0,\vec{1}}}{\psi_{0,\vec{1}}}. \label{N-tau-family}
\end{equation}
For $\nu=1/2$, we have that the $N$-party MABK violation upper bound \eqref{NMABKviolationbound} yields $2\sqrt{2}$ and its tightness conditions (see Appendix~\ref{sec:theorem-Nparties}) are satisfied when Alice measures $\pm X$ for both of her observables. With these settings, the $N$-party conditional entropy reads: $H(X|E)_{\tau(1/2)}(\varphi_{A_0}=0,\pi)=0$. By repeating the argument on the monotonicity of the entropy, we deduce that GME is necessary to certify the privacy of party's outcome in any MABK scenario.

Since private randomness of a party's outcome is a prerequisite of any DICKA protocol, it is an open question whether GME is a necessary ingredient for DICKA. Note, instead, that GME has been shown not to be necessary for device-dependent CKA \cite{Carrara2020}. Besides, in Sec.~\ref{sec:noDICKA} we argue on the apparent incompatibility of full-correlator Bell inequalities and DICKA protocols.

Finally, we mention that a lower bound on $H(X|E_{\mathrm{tot}})$ as a function of the MABK inequality violation is also derived in~\cite{JeremyMABK}, for the general $N$-party scenario. The conditional entropy bound obtained in~\cite{JeremyMABK} reads:
\begin{equation}
    H(X|E_{\mathrm{tot}}) \geq  1-h\left(\frac{1}{2} +\frac{1}{2}\sqrt{\frac{m^2}{2^N} -1}\right),  \label{Jeremy-bound}
\end{equation}
where $m$ is the observed violation of the $N$-partite MABK inequality \eqref{N-MABK}. Surprisingly, despite the fact that the bound in \cite{JeremyMABK} is derived with a completely different approach without aiming at optimality, the lower bound \eqref{Jeremy-bound} for $N=3$ coincides with the bound \eqref{H(X|E)-final-bound} obtained in this work.

\section{Full-correlator Bell inequalities and DICKA} \label{sec:noDICKA}

\noindent We provide an heuristic argument on why full-correlator Bell inequalities with two dichotomic observables per party, such as the MABK inequality, seem to be useless for DICKA protocols. We hope that this fundamental question can spark the interest of the community towards more conclusive results.

Any DICKA protocol is characterized by two essential ingredients: a violation of a multipartite Bell inequality to ensure secrecy of Alice's outcomes and correlated outcomes among all the parties yielding the conference key. Since a part of Alice's outcomes form the secret key, one of the measurements she uses to assess the violation of the inequality must be the same used for key generation \cite{PironioAcin2009,HolzComment,Holz2019DICKA}. Note that, unlike Alice, the other parties are equipped with an additional measurement option solely used for key generation.

It is known that every full-correlator Bell inequality with two dichotomic observables per party is maximally violated by the GHZ state \cite{WernerWolf}. Moreover, the only multiqubit state leading to perfectly correlated and random outcomes among all the parties is the GHZ state, when the parties measure in the $Z$ basis \cite{Epping}.

However, a GHZ state maximally violates a full-correlator Bell inequality when the measurements are chosen such that the resulting inequality (modulo rearrangements) is only composed of expectation values of GHZ stabilizers, which acquire the extremal value $1$. Moreover, the stabilizers appearing in the inequality do not act trivially on any qubit --i.e. do not contain the identity-- due to the full-correlator structure of the inequality. We call such stabilizers ``full-stabilizers'' for ease of comprehension.

The problem is that none of the $N$-partite GHZ state full-stabilizers, for $N$ odd, contains the $Z$ operator \cite{GHZstabilizers}. This implies that, in order to maximally violate the inequality, Alice's measurement directions are orthogonal to $Z$. Since one of these measurements is also used to generate her raw key, she would obtain totally uncorrelated outcomes with the rest of the parties (perfect correlations are only obtained with a GHZ state when measuring $Z$). This causes the unwanted situation of having maximal violation and perfect correlations among the parties' key bits as mutually exclusive conditions. Since both conditions are required in a DICKA protocol, the above argument constitutes an initial evidence that full-correlator Bell inequalities are not suited for DICKA protocols.

A similar argument holds when the number of parties $N$ is even ($N>2$). As a matter of fact, in this case there exists only one GHZ full-stabilizer which contains the $Z$ operator, namely: $Z^{\otimes N}$. If $\braket{Z^{\otimes N}}$ were to appear in the rearranged inequality expression, there should be at least another correlator containing at least one $Z$ operator. Indeed, if each observable in a correlator never appears again in any other term of the inequality, that correlator is useless since Eve could assign to it any value (Eve is supposed to know the inequality being tested). The lack of any other full-stabilizer containing the $Z$ operator prevents having a second correlator containing $Z$, thus excluding the term $\braket{Z^{\otimes N}}$ in the first place. Therefore, also in the $N$-even case Alice's measurements leading to maximal violation are orthogonal to $Z$, yielding uncorrelated raw key bits. We remark that the $N=2$ case is peculiar since the low number of parties allows $\braket{ZZ}$ (obtained from the term $\braket{A_1(B_0-B_1)}$ in the inequality) to appear just once in the CHSH inequality \cite{CHSH}.

It is worth mentioning that in Ref.~\cite{HolzComment} the apparent incompatibility of the MABK inequality with a DICKA protocol was already discussed. In particular, it is shown in the tripartite case that there exists no honest implementation such that the parties' outcomes are perfectly correlated and at the same time the MABK inequality is violated above the GME threshold, which is a necessary condition as we pointed out above.

Despite the concerns on the use of MABK inequalities in DICKA protocols, the results of this paper are still of fundamental interest for DIRG \cite{ColbeckThesis2006,Pironio2010,Colbeck2011,securityDIrandomness1,securityDIrandomness2,securityDIrandomness3,Woodhead2018} based on multiparty nonlocality. As a further application, in the following we improve the bound on Eve's uncertainty of Alice and Bob's outcomes derived in \cite{Woodhead2018}.

\section{Two-outcome conditional entropy bound} \label{sec:H(XY|E)}
\noindent Consider the same DI scenario of Fig.~\ref{3DI-scenario} and suppose that Eve wishes to jointly guess the measurement outcomes $X$ and $Y$ of Alice and Bob, respectively. This scenario may occur in DIRG protocols where the parties are assumed to be co-located and collaborate to generate global secret randomness \cite{EAT,Woodhead2018}. We estimate Eve's uncertainty by providing a lower bound on the conditional von Neumann entropy $H(XY|E)_{\rho_\alpha}$, as a function of the MABK violation $m_{\alpha}$. The entropy is computed on the following quantum state:
\begin{align}
    \rho^\alpha_{XYE} &= (\mathcal{E}_X \otimes \mathcal{E}_Y \otimes\mathrm{id}_E) \Tr_{C}[\ketbra{\phi^\alpha_{ABCE}}{\phi^\alpha_{ABCE}}],
\end{align}
where the maps $\mathcal{E}_X$ and $\mathcal{E}_Y$ represent Alice's and Bob's measurements, respectively, defined by the eigenvectors:
\begin{align}
    \ket{a}_X &=\frac{1}{\sqrt{2}} (\ket{0}+ (-1)^a e^{\mathbbm{i}\varphi_X} \ket{1}) \quad a\in \{0,1\}  \nonumber\\
    \ket{b}_Y &=\frac{1}{\sqrt{2}} (\ket{0}+ (-1)^b e^{\mathbbm{i}\varphi_Y} \ket{1}) \quad b\in \{0,1\} \label{eigenstates-observables}.
\end{align}
For definiteness, we select $\varphi_X=\varphi_{A_0}$ and $\varphi_Y=\varphi_{B_0}$ and define the optimization problem:
\begin{align}
    &\min_{\{\rho_{ijk},t,\vec{\varphi}\}} H(X_{A_0} Y_{B_0}|E)_{\rho_\alpha} (\rho_{ijk},t,\varphi_{A_0},\varphi_{B_0}) \nonumber\\
    &\mbox{sub. to}\,\,\braket{M_3}_{\rho_\alpha}(\rho_{ijk},t,\vec{\varphi}) \geq m_\alpha \,;\, \rho_{0jk}\geq\rho_{1jk} \,;\, \nonumber\\
    &\quad\quad\quad{\textstyle\sum_{ijk}}\,\rho_{ijk}=1 \,;\,\rho_{ijk}\geq 0, \label{tight-optimization-AliceandBob}
\end{align}
whose solution yields a tight lower bound on $H(X_{A_0} Y_{B_0}|E)_{\rho_\alpha}$.
Nonetheless, due to the MABK symmetries, the lower bounds on $H(X_{A_i} Y_{B_j}|E)_{\rho_\alpha}$, $H(X_{A_i} Z_{C_j}|E)_{\rho_\alpha}$ and $H(Y_{B_i} Z_{C_j}|E)_{\rho_\alpha}$ coincide (for $i,j\in\{0,1\}$). Thus, the solution of \eqref{tight-optimization-AliceandBob} actually provides the tight lower bound on the conditional entropy $H(XY|E)_{\rho_\alpha}$ of any pair of outcomes $X$ and $Y$ belonging to distinct parties.

Similarly to the case of $H(X|E)_{\rho_\alpha}$, we analytically solve the following simplified optimization problem (details in Appendix~\ref{sec:eta-proof}):
\begin{align}
    &\min_{\{\rho_{ijk},t,\varphi_X,\varphi_Y\}} H(X Y|E)_{\rho_\alpha} (\rho_{ijk},t,\varphi_X,\varphi_Y) \nonumber\\
    &\mbox{sub. to}\,\,\mathcal{M}^\uparrow_{\alpha}(\rho_{ijk}) \geq m_\alpha \,;\, \rho_{0jk}\geq\rho_{1jk} \,;\, \nonumber\\
    &\quad\quad\quad{\textstyle\sum_{ijk}}\,\rho_{ijk}=1 \,;\,\rho_{ijk}\geq 0, \label{optimization-AliceandBob}
\end{align}
which yields a lower bound on the solution of the original optimization problem \eqref{tight-optimization-AliceandBob}. The lower bound on $H(XY|E)_{\rho_\alpha}$ obtained by solving \eqref{optimization-AliceandBob} reads:
\begin{equation}
    H(XY|E)_{\rho_\alpha} \geq G(m_\alpha) \label{G-bound},
\end{equation}
where:
\begin{align}
    &G(m_\alpha) := 2-H\left(\left\lbrace 1-3 f(m_\alpha),f(m_\alpha),f(m_\alpha),f(m_\alpha)\right\rbrace\right), \label{cond-entropyXY-bound}
\end{align}
and where the function $f$ is defined as:
\begin{equation}
    f(m_\alpha) = \frac{1}{4}-\frac{\sqrt{3}}{24}\sqrt{m_\alpha^2-4} \label{functionf}.
\end{equation}
Similarly to the case of $H(X|E)_{\rho_\alpha}$, we can exploit the convexity of the  function in \eqref{cond-entropyXY-bound} to lower bound the conditional entropy of the global state prepared by Eve:
\begin{equation}
    H(XY|E_{\mathrm{tot}}) \geq G(m), \label{H(XY|E)-final-bound}
\end{equation}
where $m$ is the violation observed by Alice, Bob and Charlie and $G(m)$ is the function defined in \eqref{cond-entropyXY-bound}.
\begin{figure}[tb]
	\centering
	\includegraphics[width=0.9\linewidth,keepaspectratio]{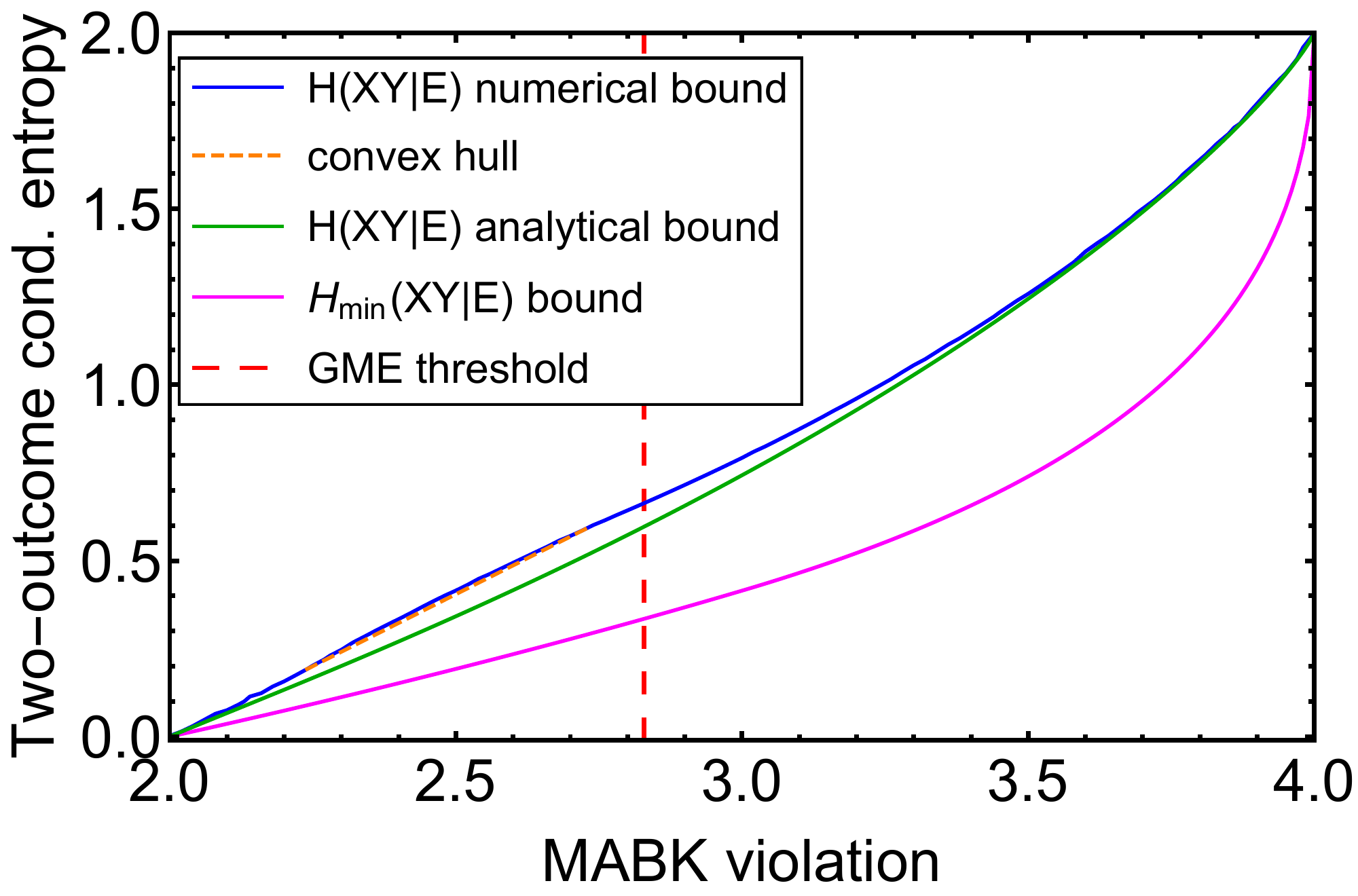}
	\caption{Analytical lower bound on the conditional von Neumann entropy $H(XY|E_{\mathrm{tot}})$ (green line, Eq.~\eqref{H(XY|E)-final-bound}) as a function of the MABK violation observed by three parties. We compare it to the lower bound on the conditional min-entropy $H_{\mathrm{min}}(XY|E_{\mathrm{tot}})$ derived in \cite{Woodhead2018} (magenta line) and to the numerical solution of  \eqref{tight-optimization-AliceandBob} (blue line), whose convex hull (dashed orange) yields an upper limit on the lowest value of $H(XY|E_{\mathrm{tot}})$. Our bound dramatically improves the one in \cite{Woodhead2018} since it directly bounds the von Neumann entropy. Unlike the case of $H(X|E_{\mathrm{tot}})$ in Fig.~\ref{plot-Hlower}, Eve's uncertainty on outcomes $X$ and $Y$ is nonzero even for violations below the GME threshold.}
	\label{plot-HXYlower}
\end{figure}

The bound in \eqref{H(XY|E)-final-bound} is plotted in Fig.~\ref{plot-HXYlower}, together with the tight lower bound on the correspondent min-entropy obtained in \cite{Woodhead2018} and a numerical optimization of \eqref{tight-optimization-AliceandBob}. As already mentioned in Sec.~\ref{sec:H(X|E)}, the tight bound on $H(XY|E_{\mathrm{tot}})$ must lie between the convex hull of the numerical curve and our analytical bound \eqref{H(XY|E)-final-bound}. Figure~\ref{plot-HXYlower} suggests that our analytical bound is close to the ideal tight bound.

We point out the dramatic improvement in certifying device-independently the privacy of two parties' outcomes with our lower bound on the conditional von Neumann entropy $H(XY|E_{\mathrm{tot}})$, as opposed to bounding the conditional min-entropy $H_{\mathrm{min}}(XY|E_{\mathrm{tot}})$ \cite{Woodhead2018}.

The min-entropy is often used to lower bound the von-Neumann entropy in DI protocols, since it can be directly estimated using the statistics of the measurement outcomes \cite{NPA1,NPA2}. In general it holds that $H\geq H_{\mathrm{min}}$~\cite{QuantumAEP}. However, bounding the von Neumann entropy with the min-entropy can be far from optimal, as in the case analyzed here (see Fig.~\ref{plot-HXYlower}).

From Fig.~\ref{plot-HXYlower} we also observe that the joint conditional entropy of two parties' outcomes $H(XY|E_{\mathrm{tot}})$ is nonzero for violations below the GME threshold, unlike the entropy of a single party's outcome $H(X|E_{\mathrm{tot}})$ (c.f Sec.~\ref{sec:H(X|E)}).

\section{Conclusion}  \label{sec:discussion}
\noindent The security of device-independent (DI) cryptographic protocols is based on the ability to bound the entropy of the protocols' outcomes, conditioned on the eavesdropper's knowledge, by a Bell inequality violation. To this aim, we considered a DI scenario where $N$ parties test a generic full-correlator Bell inequality, with two measurement settings and two outcomes per party. We proved, in this context, that it is not restrictive to reduce the most general quantum state tested by the parties to simple $N$-qubit states. Our result reduces to the only other one of this kind \cite{PironioAcin2009} when $N=2$.

In order to obtain the entropy bounds, we proved an analytical upper bound on the maximal violation of the MABK inequality achieved by a given $N$-qubit state, when the parties perform rank-one projective measurements. The bound is tight on certain classes of states and has general validity (i.e. independent of the parties' measurements) for states whose maximal violation is above the GME threshold. Our bound generalizes the known result~\cite{HHH95} valid for the CHSH inequality to an arbitrary number of parties. To the best of our knowledge, this is the first bound on the maximal violation of a $N$-partite Bell inequality achievable by a given state, expressed in terms of the state's parameters.

These results enabled us to derive an analytical lower bound on the conditional von Neumann entropy of a party's outcome, when Alice, Bob and Charlie test the tripartite MABK inequality. We also derived an analytical lower bound on the conditional von Neumann entropy of any pair of outcomes from distinct parties, which dramatically improves a similar estimation made in \cite{Woodhead2018} in terms of the corresponding min-entropy. The improvement gained by directly bounding the von Neumann entropy has direct implications for randomness generation protocols, inasmuch as it increases the fraction of random bits guaranteed to be private.

Moreover, both analytical bounds perform well when compared to the numerical estimation of the corresponding tight bounds, leaving little room for improvement.

By proving that our bound on the conditional entropy of a party's outcome is tight at the GME threshold, we deduced that genuine multipartite entanglement (GME) is necessary to guarantee the privacy of a party's random outcome in any device-independent scenario based on the MABK inequality. It is an open question whether GME is a fundamental requirement for DI conference key agreement (DICKA). In this regard, we heuristically argued that full-correlator Bell inequalities with two binary observables per party, such as the MABK inequality, are unlikely to be employed in any DICKA protocol. We envision further and more conclusive results in this direction from the scientific community interested in this topic.

The bounds on the conditional entropies derived in this work can find potential application in DI randomness generation based on multipartite nonlocality. Depending on the application, such protocols would generate local randomness for one party or global randomness for two or more parties. In all cases, the privacy of the generated random data would be ensured by entropy bounds like the ones we derived. 

Furthermore, the techniques developed in proving Theorem~\ref{thm:reduction-to-GHZ} can inspire analogous analytical reductions of the quantum state for other Bell inequalities. Indeed, of particular interest are the Bell inequalities employed in the existing DICKA protocols \cite{JeremyParityCHSH,Holz2019DICKA}, for which a result like Theorem~\ref{thm:reduction-to-GHZ} would be the first step towards a tight security analysis, which is still lacking.


\section{Methods} \label{sec:methods}

\noindent Here we present the proofs of Theorem~\ref{thm:reduction-to-GHZ} and Theorem~\ref{theorem-Nparties}.

\subsection{Proof of Theorem~\ref{thm:reduction-to-GHZ}} \label{appendix:theorem1}

\noindent The proof of Theorem~\ref{thm:reduction-to-GHZ} is based on three main ingredients:  (i) the fact that each party has only two inputs with two outputs allows to reduce the analysis to qubits and rank-one projective measurements; (ii) the symmetries of the MABK inequality allow us to set all the marginals to zero, without changing the MABK violation or the information available to the eavesdropper; (iii) the freedom in the definition of the local axes is used to further reduce the number of free parameters. Our proof is inspired by a similar proof given in \cite{PironioAcin2009}. However, our result is valid for an arbitrary number of parties $N$ in the generic $(N,2,2)$ DI scenario described in the main text. Notably, for $N=2$ we recover the result of \cite{PironioAcin2009}.

In order to prove Theorem~\ref{thm:reduction-to-GHZ}, we make use of the following Lemma~\ref{lem:dim2decomposition} which is a consequence of a result given in \cite{Masanes06} and whose proof is reported in Appendix~\ref{sec:rank-one-lemma-proof}.
\begin{Lmm}\label{lem:dim2decomposition}
Let $\{P_0,P_1\}$ and $\{Q_0,Q_1\}$ be two projective measurements acting on a Hilbert space $\mathcal{H}$, such that $P_0,P_1,Q_0$ and $Q_1$ are projectors and $P_0+P_1=\id$ and $Q_0+Q_1 = \id$. There exists an orthonormal basis in an enlarged Hilbert space $\mathcal{H^*}$ such that the four projectors are simultaneously block diagonal, in blocks of size $2\times 2$. Moreover, within a $2\times 2$ block, each projector has rank one.
\end{Lmm}
\begin{proof}[Proof of Theorem~\ref{thm:reduction-to-GHZ}]
The first step consists in reducing the state distributed by Eve to a convex combination of $N$-qubit states. To start with, every generalized measurement (positive-operator valued measure) can be viewed as a projective measurement in a larger Hilbert space. Since we did not fix the Hilbert space to which the shared quantum state belongs, we can assume without loss of generality that the parties' measurements are binary projective measurements on a given Hilbert space $\mathcal{H}$. In particular, the projectors $P_{0}^{(i)}$ and $P_{1}^{(i)}$ ($Q_{0}^{(i)}$ and $Q_{1}^{(i)}$) correspond to $\mathrm{Alice}_i$'s binary observable $A^{(i)}_{0}$ ($A^{(i)}_1$) relative to input $x_i=0$ ($x_i=1$).

Now we can apply Lemma~\ref{lem:dim2decomposition} to the projective measurements of Alice$_i$ for $i=1,\dots,N$ and state that, at every round of the protocol, the Hilbert space on which e.g. $\mathrm{Alice}_1$'s measurements are acting is decomposed as:
\begin{equation}
    \mathcal{H}^*=\oplus_\alpha \mathcal{H}^2_\alpha\,\,, \label{decomposed_Hspace}
\end{equation}
where every subspace $\mathcal{H}^2_\alpha$ is two-dimensional and both $\mathrm{Alice}_1$'s measurements act within $\mathcal{H}^2_\alpha$ as rank-one projective measurements.
From $\mathrm{Alice}_1$'s point of view, the measurement process consists of a projection in one of the two-dimensional subspaces followed by a projective measurement in that subspace (selected according to $\mathrm{Alice}_1$'s input). Therefore, Eve is effectively distributing to $\mathrm{Alice}_1$ a direct sum of qubits at every round. $\mathrm{Alice}_1$'s measurement then selects one of the qubit subspaces and performs a projective measurement within that subspace. Of course, since Eve fabricates the measurement device, the projective measurements occurring in every subspace can be predefined by Eve. Since this argument holds for every party, Eve is effectively distributing a direct sum of $N$-qubit states in each round.

Certainly, it cannot be worse for Eve to learn the flag $\alpha$ of the subspace selected in a particular round before sending the direct sum of $N$-qubit states to the parties. For this reason, we can reformulate the state preparation and measurement in a generic round of the protocol as Eve preparing a mixture
\begin{equation}
    \rho_{A_1\ldots A_N \Xi}= \textstyle\sum_{\alpha} p_{\alpha}\rho_{\alpha}\, \bigotimes_{i=1}^N \ketbra{\alpha}{\alpha}_{\xi_i} \label{distributed-state}
\end{equation}
of $N$-qubit states $\rho_{\alpha}$, together with a set of ancillae \mbox{$\Xi:=\{\xi_i\}_{i=1}^N$} (known to her) which fixes the rank-one projective measurements that each party can select on $\rho_{\alpha}$.

Let us now focus on one specific occurrence defined by a given $\alpha$, i.e. on one of the $N$-qubit states $\rho_\alpha$. For ease of notation, in the following we omit the symbol $\alpha$.

We define the plane induced by the two rank-one projective measurements of each party to be the $(x,y)$-plane of the Bloch sphere. Now, we assume without loss of generality that the statistics observed by the parties is such that every marginal is random:
\begin{align}
\left< \prod_{i\in P} A^{(i)}\right>=0, \label{statistics-assumption}
\end{align}
where $A^{(i)}$ is any dichotomic observable of $\mathrm{Alice}_i$ and $P$ is any non-empty strict subset of all the parties: $P\subsetneq \{1,\dots,N\}$. Indeed, if this is not the case, the parties can perform the following classical procedure on their outcomes which enforces the requirement in Eq.~\eqref{statistics-assumption}: ``$\mathrm{Alice}_1$ and $\mathrm{Alice}_i$ flip their outcome with probability \nicefrac{1}{2}'', repeated for every $i=2,\dots,N$. This procedure does not change the observed Bell violation since an even number of flips occurs at every time, thus leaving the correlators \eqref{full-corr} composing the Bell inequality unchanged. Moreover, it requires classical communication between the parties which we assume to be known by Eve.

Since the observed statistics always satisfies \eqref{statistics-assumption}, we can imagine that it is Eve herself who performs the classical flipping on the outputs in place of the parties. To this aim, Eve could apply the following map to the state $\rho$ she prepared, before distributing it:
\begin{equation}
    \rho \mapsto \bar{\rho}= \circ_{i=2}^{N}\mathcal{D}_i(\rho),  \label{rhomixed}
\end{equation}
where the composition operator in \eqref{rhomixed} represents the successive application of the following operations
\begin{equation}
    \mathcal{D}_{i}(\rho)=\frac{1}{2}\rho+\frac{1}{2}Z_1Z_i \rho Z_1^\dag Z_i^\dag \label{Di},
\end{equation}
with $Z_i$ representing the third Pauli operator applied on $\mathrm{Alice}_i$'s qubit. Note that the application of $Z$ prior to measurement flips the outcome of a measurement in the $(x,y)$-plane. Thus, by applying the map in \eqref{rhomixed}, Eve is distributing a state which automatically satisfies the condition \eqref{statistics-assumption}. We can safely assume that Eve implements the map in \eqref{rhomixed} since this is not disadvantageous to her. As a matter of fact, her uncertainty on the parties' outcomes, quantified by the conditional von Neumann entropy, does not increase when she sends the state $\bar{\rho}$ instead of $\rho$. We provide a detailed proof of this fact in Appendix~\ref{sec:proof-symmetrization}. Therefore, it is not restrictive to assume that the parties receive the state \eqref{rhomixed} from Eve, which can be recast as:
\begin{equation}
    \bar{\rho} = \frac{1}{2^{N-1}} \sum_{n=0}^{\floor{\frac{N}{2}}} \, \sum_{\mathbf{x}\in I(n)} Z^{\mathbf{x}} \rho Z^{\mathbf{x}}, \label{rhomixed2}
\end{equation}
with
\begin{align}
    I(n) &= \{\mathbf{x}\in\{0,1\}^N\, : \, \omega(\mathbf{x})=2n\}, \label{setI} \\
    Z^{\mathbf{x}}&= \bigotimes_{j=1}^N Z^{x_j}_j \label{Zi},
\end{align}
where the Hamming weight $\omega(\mathbf{x})$ of a bit string $\mathbf{x}$ returns the total number of bits that are equal to one and $\floor{y}$ returns the greatest integer smaller or equal to $y$.

By expressing the initial generic state $\rho$ in the GHZ basis:
\begin{align}
    \rho= \sum_{\vec{u},\vec{v}\in\{0,1\}^{N-1}}\sum_{\sigma,\tau=0}^1 \rho_{(\sigma \vec{u})(\tau \vec{v})} \ket{\psi_{\sigma,\vec{u}}}\bra{\psi_{\tau,\vec{v}}} ,
\end{align}
where $\rho_{(\sigma \vec{u})(\tau \vec{v})}\in\mathbbm{C}$ and by substituting it into \eqref{rhomixed2}, we notice that the state $\bar{\rho}$ is greatly simplified in the GHZ basis. In particular, all the coherences between states of the GHZ basis relative to different vectors $\vec{u}$ are null:
\begin{equation}
    \bar{\rho} = \sum_{\vec{u}\in\{0,1\}^{N-1}}\sum_{\sigma,\tau=0}^1 \rho_{(\sigma \vec{u})(\tau \vec{u})} \ket{\psi_{\sigma,\vec{u}}}\bra{\psi_{\tau,\vec{u}}}. 
\end{equation}
This means that the matrix representation of $\bar{\rho}$ is block-diagonal in the GHZ basis. By relabeling the non-zero matrix coefficients, we represent $\bar{\rho}$ as follows:
\begin{equation}
    \bar{\rho}= \bigoplus_{\vec{u}\in\{0,1\}^{N-1}} 
    \begin{bmatrix}
    \lambda_{0\vec{u}} & r_{\vec{u}} + \mathbbm{i}s_{\vec{u}} \\
    r_{\vec{u}} - \mathbbm{i}s_{\vec{u}} & \lambda_{1\vec{u}} \\
    \end{bmatrix} ,\label{rhomixed3}
\end{equation}
where $\lambda_{j\vec{u}},r_{\vec{u}}$ and $s_{\vec{u}}$ are real numbers.
The number of free parameters characterizing \eqref{rhomixed3} can be further reduced by exploiting the remaining degrees of freedom in the parties' local reference frames \cite{PironioAcin2009}. Indeed, although we identified the plane containing the measurement directions to be the $(x,y)$-plane for every party, they can still choose the orientation of the axes by applying rotations $R(\theta)$ along the $z$ direction. Consequently, the state distributed by Eve without loss of generality is given by:
\begin{equation}
    \bar{\rho}_+ = \bigotimes_{i=1}^N R_i(\theta_i) \,\bar{\rho} \, \bigotimes_{i=1}^N R_i^\dag(\theta_i) \label{rotatedrho},
\end{equation}
where the rotation $R_i(\theta_i)$ acts on the Hilbert space of party number $i$ and reads:
\begin{equation}
    R_i(\theta_i) = \cos\frac{\theta_i}{2} \mathrm{id} + \mathbbm{i}\sin\frac{\theta_i}{2} Z_i \label{rotation},
\end{equation}
where ``$\mathrm{id}$'' is the identity operator. Similarly to $\bar{\rho}$, even the global rotation operator is block-diagonal in the GHZ basis:
\begin{equation}
    \bigotimes_{i=1}^N R_i(\theta_i)= \bigoplus_{\vec{u}\in\{0,1\}^{N-1}} 
    \begin{bmatrix}
    \cos\frac{\beta(\vec{\theta},\vec{u})}{2} & \mathbbm{i}\sin\frac{\beta(\vec{\theta},\vec{u})}{2} \\[1ex]
    \mathbbm{i}\sin\frac{\beta(\vec{\theta},\vec{u})}{2} & \cos\frac{\beta(\vec{\theta},\vec{u})}{2} \\
    \end{bmatrix} ,\label{globalrotation}
\end{equation}
where $\vec{\theta}$ is the vector defined by the rotation angles $\{\theta_1,\dots,\theta_N\}$ and $\beta$ is a function of $\vec{\theta}$ and $\vec{u}$ defined as:
\begin{align}
    \beta(\vec{\theta},\vec{u}) &=\theta_1 + \textstyle\sum_{j=1}^{N-1} (-1)^{u_j} \theta_{j+1}. \label{beta}
\end{align}
This fact greatly simplifies the calculation in \eqref{rotatedrho}, as it allows to multiply the matrices \eqref{rhomixed3} and \eqref{globalrotation} block-by-block. The resulting block-diagonal matrix representing the state distributed by Eve reads:
\begin{equation}
    \bar{\rho}_+= \bigoplus_{\vec{u}\in\{0,1\}^{N-1}} 
    \begin{bmatrix}
    \lambda'_{0\vec{u}} & r_{\vec{u}} + \mathbbm{i}s'_{\vec{u}} \\
    r_{\vec{u}} - \mathbbm{i}s'_{\vec{u}} & \lambda'_{1\vec{u}} \\
    \end{bmatrix}\label{rhomixed4}
\end{equation}
where the new matrix coefficients are given by:
\begin{align}
    \lambda'_{0\vec{u}} = &\frac{1}{2}\left[\lambda_{0\vec{u}}+\lambda_{1\vec{u}} + (\lambda_{0\vec{u}}-\lambda_{1\vec{u}}) \cos\beta(\vec{\theta},\vec{u}) \right. \nonumber\\
    &\left. + 2s_{\vec{u}} \sin\beta(\vec{\theta},\vec{u}) \right] \label{lambda0prime}\\
    s'_{\vec{u}} = & s_{\vec{u}} \cos\beta(\vec{\theta},\vec{u}) -\frac{1}{2}(\lambda_{0\vec{u}}-\lambda_{1\vec{u}})\sin\beta(\vec{\theta},\vec{u})  \label{sprime}\\
    \lambda'_{1\vec{u}}= &\frac{1}{2}\left[\lambda_{0\vec{u}}+\lambda_{1\vec{u}} - (\lambda_{0\vec{u}}-\lambda_{1\vec{u}}) \cos\beta(\vec{\theta},\vec{u})\right. \nonumber\\
    &\left. - 2s_{\vec{u}} \sin\beta(\vec{\theta},\vec{u}) \right] \label{lambda1prime}.
\end{align}
From \eqref{sprime} we deduce that choosing the rotation angles $\theta_1,\dots,\theta_N$ such that the following linear constraint is verified:
\begin{equation}
    \theta_1 + \textstyle{\sum_{j=1}^{N-1} (-1)^{u_j} \theta_{j+1}} = \arctan \displaystyle\frac{2 s_{\vec{u}}}{\lambda_{0\vec{u}} - \lambda_{1\vec{u}}} \label{szero},
\end{equation}
sets the corresponding imaginary part in \eqref{rhomixed4} to zero: $s'_{\vec{u}}=0$. However, we can only impose $N$ constraints like \eqref{szero} on the $N$ rotation angles, thus we are able to arbitrarily set to zero $N$ terms like $s_{\vec{u}}$ in \eqref{rhomixed4}. Moreover, by applying further rotations (note that the composition of rotations is still a rotation) such that:
\begin{equation}
    \tilde{\theta}_1 + \textstyle{\sum_{j=1}^{N-1} (-1)^{u_j} \tilde{\theta}_{j+1}} = \pi \label{orderlambda},
\end{equation}
we can exchange the diagonal terms in \eqref{rhomixed4}: $\lambda'_{0\vec{u}}=\lambda_{1\vec{u}}$ and $\lambda'_{1\vec{u}}=\lambda_{0\vec{u}}$. This allows us to order up to $N$ pairs $(\lambda_{0\vec{u}},\lambda_{1\vec{u}})$, for the same argument as above. Note that the blocks with ordered pairs must be the same blocks with null imaginary parts. Indeed, if a block identified by $\vec{u}$ with null imaginary part undergoes a rotation such that $\beta(\vec{\theta},\vec{u})\neq \{0,\pm \pi\}$, it will acquire a non-zero imaginary part $s'_{\vec{u}}=-(1/2)(\lambda_{0\vec{u}}-\lambda_{1\vec{u}})\sin\beta(\vec{\theta},\vec{u})$, (see \eqref{sprime}).

Finally we construct the state $\bar{\rho}_-$ starting from $\bar{\rho}_+$ given in \eqref{rhomixed4} by replacing $r_{\vec{u}}$ with $-r_{\vec{u}}$:
\begin{equation}
    \bar{\rho}_-= \bigoplus_{\vec{u}\in\{0,1\}^{N-1}} 
    \begin{bmatrix}
    \lambda'_{0\vec{u}} & -r_{\vec{u}} + \mathbbm{i}s'_{\vec{u}} \\
    -r_{\vec{u}} - \mathbbm{i}s'_{\vec{u}} & \lambda'_{1\vec{u}} \\
    \end{bmatrix}\label{rhominus}.
\end{equation}
We observe that the two states $\bar{\rho}_\pm$ yield the same measurement statistics and provide Eve with the same information --i.e. their conditional entropies coincide. Additionally, it is not disadvantageous for Eve to prepare a balanced mixture of $\bar{\rho}_+$ and $\bar{\rho}_-$ given by $(\bar{\rho}_+ + \bar{\rho}_-)/2$, rather than preparing one of the two states with certainty. A detailed proof of these observations is given in Appendix~\ref{sec:proof-rhominus}.

We conclude that it is not restrictive to assume that Eve distributes to the parties a mixture of $N$-qubit states $\rho_\alpha$ together with an ancillary system fixing the parties' measurements. Each state $\rho_\alpha$ is represented by the following block diagonal matrix in the GHZ basis:
\begin{equation}
    \rho_\alpha = \frac{\bar{\rho}_+ + \bar{\rho}_-}{2} =  \bigoplus_{\vec{u}\in\{0,1\}^{N-1}} 
    \begin{bmatrix}
    \lambda_{0\vec{u}} & \mathbbm{i}s_{\vec{u}} \\
    -\mathbbm{i}s_{\vec{u}} & \lambda_{1\vec{u}} \\
    \end{bmatrix}\label{rhomixed5},
\end{equation}
where the diagonal elements of $N$ arbitrary blocks are ordered and the corresponding off-diagonal elements are zero. This concludes the proof.
\end{proof}

\subsection{Proof of Theorem~\ref{theorem-Nparties}} \label{appendix:theorem2}

\noindent We present the proof of Theorem~\ref{theorem-Nparties}, which generalizes the analogous result valid in the bipartite case for the CHSH inequality~\cite{HHH95}. This is, to the best of our knowledge, the only existing upper bound on the violation of the $N$-partite MABK inequality by rank-one projective measurements on an arbitrary $N$-qubit state, expressed as a function of the state's parameters. Note that an analogous upper bound on the violation of the tripartite MABK inequality was recently derived in~\cite{SS19}. However, here we show that our bound is tight on a broader class of states and valid for an arbitrary number of parties.
In order to prove Theorem~\ref{theorem-Nparties} we make use of the following Lemma~\ref{thm_two-largest-eigenvalues}, which generalizes an analogous result in \cite{HHH95} to rectangular matrices of arbitrary dimensions. The proof of Lemma~\ref{thm_two-largest-eigenvalues} is reported in Appendix~\ref{sec:theorem-Nparties}.

\begin{Lmm}\label{thm_two-largest-eigenvalues}
Let $Q$ be an $m\times n$ real matrix and let $\norm{\vec{v}}$ be the Euclidean norm of vectors  $\vec{v}\in\mathbb{R}^k$, for $k={m,n}$. Finally, let ``$\cdot$'' indicate both the scalar product and the matrix-vector multiplication. Then
\begin{equation}
    \max_{\stackrel[\norm{\vec{c}}=\norm{\vec{c}\,'}=1]{\vec{c}\perp\vec{c}\,' \,\mathrm{s.t.}}{}} \left[\norm{Q\cdot \vec{c}}^2+ \norm{Q\cdot \vec{c}\,'}^2\right] = u_1 + u_2 \,\,,  \label{two-largest-eigenvalues}
\end{equation}
where $u_1$ and $u_2$ are the largest and second-to-the-largest eigenvalues of $U \equiv Q^T Q$, respectively.
\end{Lmm}

For illustration purposes, here we report the proof of Theorem~\ref{theorem-Nparties} for the case of $N=3$ parties. The full proof is given in Appendix~\ref{sec:theorem-Nparties}.

\begin{proof}[Proof of Theorem 2 for $N=3$]
By assumption we restrict the description of the parties' observables to rank-one projective measurements on their respective qubit \cite{WernerWolf}. Hence they can be represented as follows:
\begin{align}\label{eq:observables}
    A_x=\vec{a}_x\cdot \vec{\sigma},\; B_y=\vec{b}_y\cdot \vec{\sigma},\;{\rm and}\; C_z=\vec{c}_z\cdot \vec{\sigma},
\end{align}
where $\vec{a}_x,\vec{b}_y,\vec{c}_z$ are unit vectors in $\mathbb{R}^3$ and where $\sigma_1=X,\sigma_2=Y$ and $\sigma_3=Z$. We can then express the tripartite MABK operator \eqref{3-MABK} as follows:
\begin{align}
    M_3 =\sum_{i,j,k=1}^3 M_{ijk}\sigma_i\otimes\sigma_j\otimes\sigma_k, \label{eq:MerminObs}
\end{align}
where we defined
\begin{align}
    M_{ijk}\equiv {a_0}_i{b_0}_j{c_1}_k+{a_0}_i{b_1}_j{c_0}_k+{a_1}_i{b_0}_j{c_0}_k-{a_1}_i{b_1}_j{c_1}_k.
\end{align}

A generic 3-qubit state can be expressed in the Pauli basis as follows
\begin{align}
    \rho=\frac{1}{8}\sum_{\mu,\nu,\gamma=0}^{3}\Lambda_{\mu\nu\gamma}\sigma_{\mu}\otimes \sigma_{\nu}\otimes \sigma_{\gamma}, \label{eq:state}
\end{align}
with $\Lambda_{\mu\nu\gamma}=\Tr[\rho \sigma_\mu \otimes \sigma_\nu \otimes \sigma_\gamma]$ and $\sigma_0=\mathrm{id}$. With the MABK operator in \eqref{eq:MerminObs}, the MABK expectation value on the generic 3-qubit state in \eqref{eq:state} is given by:
\begin{align}
   \braket{M_3}_\rho &= \Tr(M_3\rho) \nonumber\\
    &=\frac{1}{8}\sum_{i,j,k=1}^3 \sum_{\mu,\nu,\gamma=0}^3 M_{ijk}\Lambda_{\mu\nu\gamma}  \underbrace{\Tr\left(\sigma_i\sigma_{\mu}\otimes\sigma_j\sigma_{\nu}\otimes\sigma_k\sigma_{\gamma}\right)}_{8\delta_{i,\mu}\delta_{j,\nu}\delta_{k,\gamma}}\nonumber\\
    &=\sum_{i,j,k=1}^3 M_{ijk}\Lambda_{ijk} .  \label{expval}
\end{align}
By recalling the correlation matrix of a tripartite state (c.f.~Definition~\ref{def:corrmatrix}), the MABK expectation value in \eqref{expval} can be recast as follows:
\begin{align}
   \braket{M_3}_\rho=\,&(\vec{a}_0\otimes\vec{b}_1+\vec{a}_1\otimes\vec{b}_0)^T \cdot T_{\rho}\cdot\vec{c}_0 \nonumber\\
   &+ (\vec{a}_0\otimes\vec{b}_0 - \vec{a}_1\otimes\vec{b}_1)^T \cdot T_{\rho}\cdot\vec{c}_1 . \label{eq:MerminMatrixeq}
\end{align}
Finally, the maximum violation $\mathcal{M}_\rho$ of the MABK inequality achieved by an arbitrary 3-qubit state is obtained by optimizing \eqref{eq:MerminMatrixeq} over all possible observables that the parties can choose to measure:
\begin{align}
    \mathcal{M}_\rho = \max_{\substack{\vec{a}_i,\vec{b}_i,\vec{c}_i \,\,\mathrm{s.t.} \\ \norm{\vec{a}_i}=\norm{\vec{b}_i}=\norm{\vec{c}_i}=1}} &(\vec{a}_0\otimes\vec{b}_1+\vec{a}_1\otimes\vec{b}_0)^T \cdot T_{\rho}\cdot\vec{c}_0 \nonumber\\
    &+ (\vec{a}_0\otimes\vec{b}_0 - \vec{a}_1\otimes\vec{b}_1)^T \cdot T_{\rho}\cdot\vec{c}_1  \label{max-violation3}.
\end{align}
Let us now evaluate the norm of the composite vectors in \eqref{max-violation3}:
\begin{align}
    {\left\| \vec{a}_0\otimes \vec{b}_1 +   \vec{a}_1\otimes \vec{b}_0 \right\|}^2
    &=2+2\underbrace{\cos\theta_a\cos\theta_b}_{\equiv \cos\theta_{ab}} \nonumber\\
    &=4\cos^2\left(\frac{\theta_{ab}}{2}\right)\,\,,
\end{align}
where $\theta_{a}$ ($\theta_{b}$) is the angle between vectors $\vec{a}_0$ and $\vec{a}_1$ ($\vec{b}_0$ and $\vec{b}_1$). Similarly,
\begin{align}
    {\left\| \vec{a}_0\otimes \vec{b}_0-   \vec{a}_1\otimes \vec{b}_1 \right\|}^2 &=4\sin^2\left(\frac{\theta_{ab}}{2}\right).
\end{align}
We then define normalized vectors $\vec{v}_0$ and $\vec{v}_1$ such that
\begin{align}
    \vec{a}_0\otimes \vec{b}_1 +   \vec{a}_1\otimes \vec{b}_0 & = 2\cos\left(\frac{\theta_{ab}}{2}\right)\vec{v}_0,  \label{eq:vecv0}\\
    \vec{a}_0\otimes \vec{b}_0 -   \vec{a}_1\otimes \vec{b}_1 & = 2\sin\left(\frac{\theta_{ab}}{2}\right)\vec{v}_1.  \label{eq:vecv1}
\end{align}
It can be easily checked that the normalized vectors $\vec{v}_0$ and $\vec{v}_1$ are orthogonal. By substituting the definitions \eqref{eq:vecv0} and \eqref{eq:vecv1} into the maximal violation of the MABK inequality \eqref{max-violation3}, we can upper bound the latter as follows:
\begin{align}
    \mathcal{M}_\rho \leq\max_{\substack{\vec{c}_i,\vec{v}_i,\theta_{ab}\,\,\mathrm{s.t.}\\ \norm{\vec{c}_i}=\norm{\vec{v}_i}=1 \, \wedge \, \vec{v}_0 \perp \vec{v}_1}}   &2\cos\left(\frac{\theta_{ab}}{2}\right){\vec{v}_0}^T\cdot T_{\rho}\cdot \vec{c}_0 \nonumber\\
    &+ 2\sin\left(\frac{\theta_{ab}}{2}\right){\vec{v}_1}^T\cdot T_{\rho}\cdot \vec{c}_1 . \label{ineq:Mrho}
\end{align}
The inequality in \eqref{ineq:Mrho} is due to the fact that now the optimization is over arbitrary orthonormal vectors $\vec{v}_0,\vec{v}_1$ and angle $\theta_{ab}$, while originally the optimization was over variables satisfying the structure imposed by \eqref{eq:vecv0} and \eqref{eq:vecv1}. We now simplify the r.h.s. of \eqref{ineq:Mrho} to obtain the theorem claim. In particular, we optimize over the unit vectors $\vec{c}_0$ and $\vec{c}_1$ by choosing them in the directions of $T_\rho^T\cdot \vec{v}_0$ and $T_\rho^T\cdot \vec{v}_1$, respectively, and we also optimize over $\theta_{ab}$ by exploiting the fact that the general expression $A\,\cos\theta + B\,\sin\theta$ is maximized to $\sqrt{A^2 +B^2}$ for $\theta=\arctan{B/A}$:  
\begin{align}
    \mathcal{M}_\rho  &\leq \max_{\substack{\vec{v}_i,\theta_{ab}\,\,\mathrm{s.t.}\\ \norm{\vec{v}_i}=1 \, \wedge \, \vec{v}_0 \perp \vec{v}_1}} 2\left[ \cos\left(\frac{\theta_{ab}}{2}\right)\norm{T_\rho^T\cdot \vec{v}_0} \right.\nonumber\\
    &\left.\hspace{2.5cm}+ \sin\left(\frac{\theta_{ab}}{2}\right)\norm{T_\rho^T\cdot \vec{v}_1}\right] \nonumber\\
    &= \max_{\substack{\vec{v}_i\,\,\mathrm{s.t.}\\ \norm{\vec{v}_i}=1 \, \wedge \, \vec{v}_0 \perp \vec{v}_1}} 2\sqrt{\norm{T_\rho^T\cdot \vec{v}_0}^2 + \norm{T_\rho^T\cdot \vec{v}_1}^2} . \label{ineq:Mrho2}
\end{align}
Finally, by applying the result of Lemma~\ref{thm_two-largest-eigenvalues}, we know that the maximum in \eqref{ineq:Mrho2} is achieved when $\vec{v}_0$ and $\vec{v}_1$ are chosen in the direction of the eigenstates of $T_\rho T_\rho^T$ corresponding to the two largest eigenvalues. This concludes the proof for the $N=3$ case:
\begin{equation}
    \mathcal{M}_\rho  \leq 2\sqrt{t_0 + t_1} \,\,, \label{ineq:Mrho3}
\end{equation}
where $t_0$ and $t_1$ are the two largest eigenvalues of $T_\rho T_\rho^T$.
\end{proof}

\subsubsection{Tightness conditions}

\noindent The bound \eqref{NMABKviolationbound} is tight if the correlation matrix $T_\rho$ of the considered state satisfies certain conditions, i.e. for certain classes of states. Here we report the tightness conditions valid in the $N=3$ case, while the ones for general $N$ and their derivation are given in Appendix~\ref{sec:theorem-Nparties}.

The upper bound \eqref{NMABKviolationbound} on the maximal violation of the tripartite MABK inequality by a given state $\rho$ is tight, that is there exists a quantum implementation achieving the bound, if there exist unit vectors $\vec{a}_0,\vec{a}_1,\vec{b}_0$ and $\vec{b}_1$ in $\mathbb{R}^3$ such that the following identities are satisfied:
\begin{align}
   &\vec{a}_0\otimes \vec{b}_1 +\vec{a}_1\otimes \vec{b}_0 = 2\sqrt{\frac{t_0}{t_0 + t_1}}\, \vec{t}_0 \nonumber\\
   &\vec{a}_0\otimes \vec{b}_0 -\vec{a}_1\otimes \vec{b}_1 = 2\sqrt{\frac{t_1}{t_0 + t_1}}\, \vec{t}_1, \label{tightnessconditions}
\end{align}
where $\vec{t}_0$ and $\vec{t}_1$ are the normalized eigenvectors of $T_\rho T_\rho^T$ corresponding to the two largest eigenvalues $t_0$ and $t_1$. The tightness conditions in \eqref{tightnessconditions} are sufficient conditions such that the equality sign holds in \eqref{ineq:Mrho3} and can be directly deduced from the Theorem's proof.

We point out that by repeating the proof with different definitions of correlation matrix, one can potentially end up with alternative MABK violation upper bounds together with their own set of tightness conditions. This depends on the symmetries of the state $\rho$.

More concretely, the correlation matrix of a tripartite state $\rho$ used in the proof above is a $9\times3$ matrix expressed as follows (c.f. Definition~\ref{def:corrmatrix}):
\begin{equation}
    [T_{\rho}]_{ij}=\Tr[\rho (\sigma_{\ceil{\frac{i}{3}}}\otimes \sigma_{i-3(\ceil{\frac{i}{3}}-1)}\otimes \sigma_j)] , \label{correl-matrix}
\end{equation}
where $i\in\{1,\dots,9\}$ and $j\in\{1,2,3\}$. With the definition \eqref{correl-matrix}, we expressed the MABK expectation value as in \eqref{eq:MerminMatrixeq}. This led to the MABK violation upper bound \eqref{ineq:Mrho3} and to the tightness conditions \eqref{tightnessconditions}. However, nothing prevents us from defining the tripartite correlation matrix as:
\begin{equation}
    [T'_{\rho}]_{ij}=\Tr[\rho ( \sigma_{\ceil{\frac{i}{3}}}\otimes\sigma_j\otimes\sigma_{i-3(\ceil{\frac{i}{3}}-1)})] , \label{correl-matrix2}
\end{equation}
or as:
\begin{equation}
    [T''_{\rho}]_{ij}=\Tr[\rho (\sigma_j\otimes \sigma_{\ceil{\frac{i}{3}}}\otimes\sigma_{i-3(\ceil{\frac{i}{3}}-1)})] . \label{correl-matrix3}
\end{equation}
The alternative definitions of the correlation matrix lead to similar proofs of the MABK violation upper bound. In particular, we obtain an analogous MABK violation upper bound \eqref{ineq:Mrho3} and analogous tightness conditions \eqref{tightnessconditions}, except that the eigenvalues $t_0,t_1$ and eigenvectors $\vec{t_0},\vec{t_1}$ of $T_\rho T^T_\rho$ are replaced by the corresponding eigenvalues and eigenvectors of $T'_\rho {T'}^{T}_\rho$ or $T''_\rho {T''}^{T}_\rho$, depending on the chosen correlation matrix.

An example showing the importance of this remark is given by the family of states $\tau(\nu)$ defined in \eqref{tau-family}. Indeed, the MABK violation upper bound obtained for $\tau(\nu)$ by using the correlation matrices $T_\rho$, $T'_\rho$ and $T''_\rho$ reads the same and is given in \eqref{tau-MABK}. However, the tightness conditions related to $T_\rho$ and $T'_\rho$ are satisfied, while those related to $T''_\rho$ are not. Thus, the use of different correlation matrices in the above proof can lead to tighter MABK violation upper bounds or to a successful verification of their tightness.

It is interesting to compare the tightness of our bound with the bound derived in \cite{SS19}. The major difference is that our bound can be saturated even when the matrix $T_\rho T_\rho^T$ has no degenerate eigenvalues, opposed to \cite{SS19} which requires the degeneracy of the largest eigenvalue of $T_\rho T_\rho^T$. When the matrix $T_\rho T_\rho^T$ is degenerate in its largest eigenvalue (i.e. $t_0=t_1$), we recover the same tightness conditions of \cite{SS19}. For this reason, our bound is tight on a larger set of states compared to the bound in \cite{SS19}.

\section*{Acknowledgements}
\noindent We thank Timo Holz and Flavien Hirsch for helpful discussions, and Peter Brown for clarifying contributions regarding the tightness of the entropy bounds. This work was funded by the Deutsche Forschungsgemeinschaft (DFG, German Research Foundation) under Germany's Excellence Strategy - Cluster of Excellence Matter and Light for Quantum Computing (ML4Q) EXC 2004/1 - 390534769, by the European Union’s Horizon 2020 research and innovation programme under the Marie Sk{\l}odowska-Curie grant agreement No~675662, and by the Federal Ministry of Education and Research BMBF (Project Q.Link.X and HQS).

\appendix

\section{REDUCTION TO RANK-ONE PROJECTIVE MEASUREMENTS} \label{sec:rank-one-lemma-proof}

\noindent Here we provide a detailed proof of Lemma~\ref{lem:dim2decomposition}, by building on a result proved in Ref.~\cite{Masanes06}. We report the Lemma's statement for clarity.\\

\newcommand{\pa}{P_{0}}
\newcommand{\pb}{P_{1}}
\newcommand{\qa}{Q_{0}}
\newcommand{\qb}{Q_{1}}

\noindent \textbf{Lemma~1.}\textit{ Let $\DE{\pa, \pb}$ and $\DE{\qa, \qb}$ be two projective measurements acting on a Hilbert space $\mathcal{H}$, such that $\pa, \pb,\qa$ and $\qb$ are projectors, $\pa+\pb=\id$ and $\qa+\qb = \id$. There exists an orthonormal basis in an enlarged Hilbert space $\mathcal{H^*}$ such that the four projectors are simultaneously block diagonal, in blocks of size $2\times 2$. Moreover, within a $2\times 2$ block, each projector has rank one.}\\

\begin{proof}
Let us consider the following three positive operators $\pa$, $\pa\qa\pa$ and $\pa\qb\pa$. One can check that they commute and therefore can be simultaneously diagonalized. Let $\ket{v}$ be one of their simultaneous eigenvector. Since $\pb\cdot \pa =0$, then $\pb \ket{v}=0$. So $\ket{v}$ is also an eigenvector of $\pb$ with eigenvalue zero. 
Now, because $\qa+\qb=I$, we cannot have that $\qa \ket{v}=0$ and $\qb \ket{v}=0$. Therefore one of the following cases hold:
\begin{itemize}
    \item If $\qa \ket{v}=0$: then $\qb \ket{v}=\ket{v}$, and the span of $\ket{v}$ corresponds to a $1\times 1$ block in which $\pa, \pb, \qa, \qb$ have $\ket{v}$ as a common eigenvector with respective eigenvalues $1,0,0,1$.
    \item If $\qb \ket{v}=0$: then similarly we have a $1\times 1$ block in which $\pa, \pb, \qa, \qb$ have $\ket{v}$ as a common eigenvector with respective eigenvalues $1,0,1,0$.
    \item If $\qa \ket{v}\neq 0$ and $\qb \ket{v}\neq 0$: then we define the orthogonal vectors $\ket{u_0}=\qa \ket{v}$ and $\ket{u_1}=\qb \ket{v}$ and the 2-dimensional subspace $E_v=\DE{c_0\ket{u_0}+c_1\ket{u_1}: c_0,c_1\in \mathbb{C} }$. We have that $\ket{v}\in E_v$ since $\ket{v}=\ket{u_0}+\ket{u_1}$. Because $\ket{v}$ is also an eigenvector of $\pa\qa\pa$ and $\pa\qb\pa$, then $\pa\ket{u_0}=\pa\qa\ket{v}=\pa\qa\pa\ket{v}\propto \ket{v}$, similarly $\pa\ket{u_1}\propto \ket{v}$. Therefore, $\exists \ket{w}\in E_v$ such that $\pa \ket{w}=0$ and then $\pb \ket{w}=\ket{w}$. So the vectors $\ket{u_0},\ket{u_1}\in E_v$ are simultaneous eigenvectors of $\qa$ and $\qb$, and the vectors $\ket{v},\ket{w}\in E_v$ are simultaneous eigenvectors of $\pa$ and $\pb$. And the subspace $E_v$ corresponds to a $2 \times 2$ simultaneous diagonal block for the measurements operators $\pa, \pb, \qa, \qb$.
\end{itemize}
This procedure can be performed on all the simultaneous eigenvectors of $\pa$, $\pa\qa\pa$ and $\pa\qb\pa$, and similarly on the remaining simultaneous eigenvectors of  $\pb$, $\pb\qa\pb$ and $\pb\qb\pb$. 

Now, if we restrict to a $2 \times 2$ subspace $E_v$ with $\Pi_v$ being the projector on the subspace $E_v$, the projectors $\Pi_v\pa\Pi_v, \Pi_v\pb\Pi_v, \Pi_v\qa\Pi_v, \Pi_v\qb\Pi_v$ are given by
\begin{align}
\begin{split}
  \Pi_v\pa\Pi_v&=\frac{\ketbra{v}{v}}{\langle v | v \rangle}\\  
   \Pi_v\pb\Pi_v&=\frac{\ketbra{w}{w}}{\langle w | w \rangle}\\ 
    \Pi_v\qa\Pi_v&=\frac{\ketbra{u_0}{u_0}}{\langle u_0 | u_0 \rangle}\\ 
     \Pi_v\qb\Pi_v&=\frac{\ketbra{u_1}{u_1}}{\langle u_1 | u_1 \rangle}
   \end{split}
\end{align}
i.e., they are all rank-one projectors.

Within a $1\times 1$ block, the two measurements defined by $\DE{\pa, \pb}$ and $\DE{\qa, \qb}$ have fixed outputs. 
Let $\ket{\tilde{v}}$ be a normalized simultaneous eigenvector of $\pa$, $\pa\qa\pa$ and $\pa\qb\pa$ and consider the case $Q_0\ket{\tilde{v}}=0$, which leads to a block of size $1 \times 1$ formed by the span of the vector $\ket{\tilde{v}}$. We can now artificially enlarge the system dimension by embedding this block into a block of size $2 \times 2$. Let $\ketbra{\tilde{w}}{\tilde{w}}$ be a projector on the extra artificial dimension, with $\ket{\tilde{w}}$ a normalized vector. Then we can define the two-dimensional subspace $E_{\tilde{v}}=\DE{c_0\ket{\tilde{v}}+c_1\ket{\tilde{w}}:c_0,c_1\in \mathbb{C}}$, and we define the projectors within this subspace to be given by: $\Pi_{\tilde{v}}\pa\Pi_{\tilde{v}}=\ketbra{\tilde{v}}{\tilde{v}}$,  $\Pi_{\tilde{v}}\pb\Pi_{\tilde{v}}=\ketbra{\tilde{w}}{\tilde{w}}$,  $\Pi_{\tilde{v}}\qa\Pi_{\tilde{v}}=\ketbra{\tilde{w}}{\tilde{w}}$, and  $\Pi_v\qb\Pi_v=\ketbra{\tilde{v}}{\tilde{v}}$. One can perform a similar embedding for the other case that leads to a $1 \times 1$ block, that is $\qb \ket{\tilde{v}}=0$.
Note that the new projective measurements defined on $\mathcal{H^*}$, when applied to a quantum state $\rho$ on $\mathcal{H}$ that has no components in the artificial dimensions, have still fixed outcomes in the enlarged subspaces like $E_{\tilde{v}}$. 

With this artificial construction, the representation of the four projectors $\pa,\pb,\qa$ and $\qb$ in the artificially enlarged Hilbert space $\mathcal{H}^*$ is only composed of $2 \times 2$ diagonal blocks. Moreover, if we restrict to one of these blocks, the two measurements defined by $\DE{\pa, \pb}$ and $\DE{\qa, \qb}$ are rank-one projective measurements.
\end{proof}

\section{EVE'S UNCERTAINTY IS NON-INCREASING UNDER SYMMETRIZATION OF THE OUTCOMES} \label{sec:proof-symmetrization}
\setcounter{equation}{0}

\noindent In proving Theorem~\ref{thm:reduction-to-GHZ}, we argue that all the marginals are random without loss of generality. This can be enforced by assuming that Eve flips the classical outcomes of the measurements in specific combinations. Otherwise, Eve could also provide the parties with a state that inherently leads to the symmetrized marginals, which is the mixture $\bar{\rho}$ given in \eqref{rhomixed2}. However, Eve would provide such a state in place of the original (unknown) state $\rho$ only if her uncertainty on the parties' outcomes does not increase.

We quantify Eve's uncertainty via the von Neumann entropy of the classical outcomes conditioned on Eve's quantum side information $E$. The specific outcomes that we consider depend on the cryptographic application that is being addressed. For instance, in the main text we employ Theorem~\ref{thm:reduction-to-GHZ} to tightly estimate Eve's uncertainty on Alice's random outcome $X$ by computing $H(X|E)$, when Alice, Bob and Charlie test the MABK inequality. This result finds potential application in DICKA and DIRG protocols. Indeed, in a DICKA scheme Bob and Charlie would correct their raw key bits to match Alice's bits represented by $X$, while in a DIRG protocol the goal is to ensure that Alice's random outcome $X$ is unknown to Eve. Additionally, we employ Theorem~\ref{thm:reduction-to-GHZ} to estimate Eve's uncertainty on the outcomes of Alice ($X$) and Bob ($Y$) jointly, by computing $H(XY|E)$.

For illustration purposes, here we provide the full proof that Eve's uncertainty of Alice's outcome $X$ is non-increasing if she distributes the state $\bar{\rho}$ in place of $\rho$ to $N=3$ parties. However, we remark that an analogous proof would hold for any number of parties and any number of outcomes. Therefore, we must verify that the following condition is met:
\begin{equation}
    H(X|E)_{\rho} \geq H(X|E_{\mathrm{tot}})_{\bar{\rho}} , \label{uncertainty-nonincreasing}
\end{equation}
where Eve's quantum system $E_{\mathrm{tot}}=ETT'$ contains: the quantum side information $E$, the outcome of the random variable $T$ indicating to Eve which of the four states in the mixture $\bar{\rho}$ to distribute, and the purifying system $T'$. Indeed, Eve preparing $\bar{\rho}$ can be interpreted as she preparing one of the four states:
\begin{align}
    &\rho, \,\, (Z \otimes Z \otimes \mathrm{id})\, \rho \, (Z \otimes Z \otimes \mathrm{id}), \nonumber\\
    &(Z \otimes \mathrm{id} \otimes Z)\, \rho \, (Z \otimes \mathrm{id} \otimes Z), \nonumber\\
    &(\mathrm{id} \otimes Z \otimes Z)\, \rho \, (\mathrm{id} \otimes Z \otimes Z) \label{mixture-states}
\end{align}
depending on the outcome $t$ of a random variable stored in the register $T$. Since Eve holds the purification of every state in \eqref{mixture-states}: $\{\ket{\phi^t_{ABCE}}\}_{t=1}^4$, the global state prepared by Eve is:
\begin{align}
    \bar{\rho}_{ABCET}=\frac{1}{4} \sum_t \ket{\phi^t_{ABCE}}\bra{\phi^t_{ABCE}} \otimes \ketbra{t}{t}_T  \label{toss-prepared}
\end{align}
Finally, we assume that Eve holds the purifying system of the global state, thus the state she prepares is:
\begin{align}
    \ket{\bar{\phi}_{ABCETT'}}= \frac{1}{2} \sum_t \ket{\phi^t_{ABCE}}\otimes \ket{t}_T \otimes \ket{t}_{T'},
\end{align}
which is a purification of \eqref{toss-prepared}, where both registers $T$ and $T'$ are held by Eve and thus appear in $E_{\mathrm{tot}}$.\medskip\\

In order to prove \eqref{uncertainty-nonincreasing}, we start by using the strong subadditivity property:
\begin{align}
    H(X|E_{\mathrm{tot}})_{\bar{\rho}} \leq H(X|ET)_{\bar{\rho}}  \label{strong-sub}
\end{align}
where the r.h.s. entropy is computed on the following state:
\begin{align}
    &\bar{\rho}_{XET} = (\mathcal{E}_X\otimes \mathrm{id}_{ET} ) \Tr_{BC} \left[\bar{\rho}_{ABCET}\right] \nonumber\\
    &= \frac{1}{4} (\mathcal{E}_X\otimes \mathrm{id}_{ET} )\Tr_{BC} \left[ \sum_t \ket{\phi^t_{ABCE}}\bra{\phi^t_{ABCE}} \otimes \ketbra{t}{t}_T\right] \nonumber\\
    &\equiv \frac{1}{4} \sum_t \rho_{XE}^t \otimes \ketbra{t}{t}_T, \label{rhobarXET}
\end{align}
where the quantum map $$\mathcal{E}_X(\sigma)=\sum_{a=0}^1 \ketbra{a}{a} \bra{a} \sigma \ket{a}$$ represents the projective measurement performed by Alice.
Being the state in Eq.~\eqref{rhobarXET} a c.q. state, its entropy simplifies to:
\begin{align}
    H(X|ET)_{\bar{\rho}} = \frac{1}{4} \sum_t H(X|E)_{\rho^t} .\label{cqentropy}
\end{align}
The last part of the proof shows that $H(X|E)_{\rho^t}$ is actually independent of $t$ and equal to conditional entropy of the original state $H(X|E)_{\rho}$. This is clear if the state $\rho^t_{XE}$ is made explicit. From Eq.~\eqref{rhobarXET} we have that:
\begin{align}
    \rho^t_{XE} =(\mathcal{E}_X\otimes \mathrm{id}_{ET} )\Tr_{BC} \left[\ket{\phi^t_{ABCE}}\bra{\phi^t_{ABCE}} \right], \label{rhotXE}
\end{align}
where $\ket{\phi^t_{ABCE}}$ is the purification of one of the four states in \eqref{mixture-states} prepared by Eve according to the random variable $T$. For definiteness, let's fix that state to be $(Z \otimes Z \otimes \mathrm{id})\, \rho \, (Z \otimes Z \otimes \mathrm{id})$, although an analogous reasoning holds for any other state in Eq.~\eqref{mixture-states}. By writing $\rho$ in its spectral decomposition:
\begin{align}
    \rho= \sum_{\lambda} \lambda \ketbra{\lambda}{\lambda} ,
\end{align}
we can immediately explicit $\ket{\phi^t_{ABCE}}$ as follows:
\begin{align}
    \ket{\phi^t_{ABCE}} = \sum_\lambda \sqrt{\lambda} \ket{\lambda^t}_{ABC} \otimes \ket{e_\lambda}_E, \label{phiABCE}
\end{align}
where the eigenstates of the operator $(Z \otimes Z \otimes \mathrm{id})\, \rho \, (Z \otimes Z \otimes \mathrm{id})$ read: $\ket{\lambda^t}=(Z \otimes Z \otimes \mathrm{id}) \ket{\lambda}$. By 
substituting \eqref{phiABCE} into \eqref{rhotXE} and by expliciting the map $\mathcal{E}_X$ we obtain the following expression:
\begin{align}
    &\rho^t_{XE} = \nonumber\\
    &=\sum_{a=0}^1 \ketbra{a}{a} \otimes \sum_{\lambda,\sigma} \sqrt{\lambda \sigma} \Tr_{BC} \left[\bra{a}\ket{\lambda^t} \bra{\sigma^t}\ket{a} \right] \ketbra{e_\lambda}{e_\sigma} \nonumber\\
    &= \sum_{a=0}^1 \ketbra{a}{a} \otimes \sum_{\lambda,\sigma} \sqrt{\lambda \sigma} \Tr_{BC} \left[\bra{\bar{a}} \ket{\lambda} \bra{\sigma} \ket{\bar{a}} \right] \ketbra{e_\lambda}{e_\sigma} \nonumber\\
    &= \sum_{a=0}^1 \ketbra{\bar{a}}{\bar{a}} \otimes \sum_{\lambda,\sigma} \sqrt{\lambda \sigma} \Tr_{BC} \left[\bra{a} \ket{\lambda} \bra{\sigma} \ket{a} \right] \ketbra{e_\lambda}{e_\sigma} \nonumber\\
    &\equiv \sum_{a=0}^1 \ketbra{\bar{a}}{\bar{a}} \otimes \rho^a_E , \label{rhotXE2}
\end{align}
where in the second equality we used the fact that Alice's measurement lies in the $(x,y)$-plane hence the $Z$ operator flips its outcome ($a\rightarrow\bar{a}$) and the cyclic property of the trace. In the third equality we relabelled the classical outcomes: $a \leftrightarrow \bar{a}$. Finally, by comparing \eqref{rhotXE2} with the analogous state $\rho_{XE}$ obtained from the original state $\rho$ (i.e. in the case where Eve does not prepare the mixture of states in \eqref{mixture-states}):
\begin{align}
    \rho_{XE}=\sum_{a=0}^1 \ketbra{a}{a} \otimes \rho^a_E ,
\end{align}
we observe that $\rho^t_{XE}$ and $\rho_{XE}$ are the same state up to a permutation of the classical outcomes, thus their conditional entropies coincide:
\begin{align}
    H(X|E)_{\rho^t} = H(X|E)_{\rho} \quad \forall \,t  \label{sameconditionalentropies}.
\end{align}
In conclusion, by combining Eqs.~\eqref{sameconditionalentropies}, \eqref{cqentropy} and \eqref{strong-sub}, we obtain the claim given in Eq.~\eqref{uncertainty-nonincreasing}. This concludes the proof.

\section{EQUIVALENCE OF $\bar{\rho}_+$ AND $\bar{\rho}_-$} \label{sec:proof-rhominus}
\setcounter{equation}{0}

\noindent In the proof of Theorem~\ref{thm:reduction-to-GHZ} we claim that it is not restrictive to assume that Eve distributes the following mixture
\begin{equation}
    \rho_\alpha= \frac{\bar{\rho}_+ +\bar{\rho}_-}{2}, \label{mixture-rhoprime}
\end{equation}
in place of the state $\bar{\rho}_+$ given in Eq.~(52). For illustration purposes we prove the claim in the case where three parties, Alice, Bob and Charlie, test a $(3,2,2)$ full-correlator Bell inequality and are interested in bounding Eve's uncertainty about Alice's outcome $X$, quantified by the conditional von Neumann entropy $H(X|E)$. Nevertheless, an analogous proof would hold for any number of parties and joint entropies.

In the first part of the proof, we verify that the states $\bar{\rho}_+$ and $\bar{\rho}_-$ are equivalent from the viewpoint of the protocol. Precisely, the statistics generated by the two states coincides, as well as Eve's uncertainty about Alice's outcome, quantified by the conditional entropy $H(X|E)$. In the second part we show that Eve's uncertainty does not increase if she prepares a balanced mixture of the two states \eqref{mixture-rhoprime}, instead of preparing one of the two states singularly.

We start by computing the statistics generated by the states $\bar{\rho}_+$ in \eqref{rhomixed4} and $\bar{\rho}_-$, which read as follows for $N=3$:
\begin{align}
    \bar{\rho}_\pm &= \sum_{i,j,k=0}^1 \lambda_{ijk} \ket{\psi_{i,j,k}}\bra{\psi_{i,j,k}} \nonumber\\
    &\pm  \sum_{j,k=0}^1 r_{jk} \left(\ket{\psi_{0,j,k}}\bra{\psi_{1,j,k}} + \mathrm{h.c.}\right) \nonumber\\
    &+ \mathbbm{i} s \left(\ket{\psi_{0,1,1}}\bra{\psi_{1,1,1}} - \mathrm{h.c.} \right),
\end{align}
where $\mathrm{h.c.}$ indicates the Hermitian conjugate of the term appearing alongside it.
Note that we arbitrarily assumed three out of four off-diagonal elements to be purely real, according to the prescription characterizing $\bar{\rho}_+$ and $\bar{\rho}_-$.

Since we fixed the parties' measurements to be in the $(x,y)$-plane, their observables and the relative eigenstates can be written as follows:
\begin{align}
    &A=\cos(\varphi_A) X +\sin(\varphi_A) Y,  \nonumber\\
    &\ket{a}_A= \frac{1}{\sqrt{2}}(\ket{0} + (-1)^a e^{\mathbbm{i}\varphi_A} \ket{1}) \nonumber\\
    &B=\cos(\varphi_B) X +\sin(\varphi_B) Y, \nonumber\\
    &\ket{b}_B= \frac{1}{\sqrt{2}}(\ket{0} + (-1)^b e^{\mathbbm{i}\varphi_B} \ket{1}) \nonumber\\
    &C=\cos(\varphi_C) X +\sin(\varphi_C) Y,  \nonumber\\
    &\ket{c}_C= \frac{1}{\sqrt{2}}(\ket{0} + (-1)^c e^{\mathbbm{i}\varphi_C} \ket{1}) \label{observables},
\end{align}
where $X,Y$ and $Z$ are the Pauli operators, $A$, $B$ and $C$ are the observables of Alice, Bob and Charlie, respectively, and the measurement outcomes are defined to be $a,b,c \in\{0,1\}$ (where $a=0$ corresponds to eigenvalue +1 and $a=1$ to eigenvalue -1). 
Then, the statistics generated by the states $\bar{\rho}_+$ and $\bar{\rho}_-$ reads:\vspace{1ex}
\begin{align}
    &\Pr[A=a,B=b,C=c]_{\bar{\rho}_\pm} = \nonumber\\
    &\sum_{i,j,k=0}^1 \lambda_{ijk} \bra{\psi_{i,j,k}}\ketbra{a,b,c}{a,b,c}\ket{\psi_{i,j,k}} \nonumber\\
    &\pm 2 \sum_{j,k=0}^1 r_{jk} \mathrm{Re}[{\bra{\psi_{0,j,k}}\ketbra{a,b,c}{a,b,c}\ket{\psi_{1,j,k}}}]  \nonumber\\
    &- 2\, s\, \mathrm{Im}[{\bra{\psi_{1,1,1}}\ketbra{a,b,c}{a,b,c}\ket{\psi_{0,1,1}}}] \label{statistics}.
\end{align}
Therefore, the two statistics coincide if and only if the coefficients of the terms $r_{jk}$ are all identically null:
\begin{align}
    \mathrm{Re}[{\bra{\psi_{0,j,k}}\ketbra{a,b,c}{a,b,c}\ket{\psi_{1,j,k}}}]= 0 \quad \forall \, j,k,a,b,c \label{coeff-rjk}.
\end{align}
A straightforward calculation of the coefficients of $r_{jk}$, by using the expressions in Eqs.~\eqref{observables} and the GHZ-basis Definition, leads to the following result:
\begin{align}
    &{\bra{\psi_{0,j,k}}\ketbra{a,b,c}{a,b,c}\ket{\psi_{1,j,k}}} = \nonumber\\
    &\mathbbm{i} \, \frac{2 (-1)^{a+b+c} \mathrm{Im}[e^{\mathbbm{i}\varphi_A}e^{\mathbbm{i}\varphi_B (-1)^j}e^{\mathbbm{i}\varphi_C (-1)^k}]}{16},
\end{align}
which is indeed purely imaginary. This proves the condition \eqref{coeff-rjk} and thus that the statistics of $\bar{\rho}_+$ and $\bar{\rho}_-$ are identical.

The next step of the proof consists in showing that Eve's uncertainty about Alice's outcome is unchanged if she distributes $\bar{\rho}_+$ or $\bar{\rho}_-$, i.e. the following condition must be verified:
\begin{align}
    H(X|E)_{\bar{\rho}_+} = H(X|E)_{\bar{\rho}_-}. \label{same-uncertainty}
\end{align}
In order to show \eqref{same-uncertainty}, we compute the conditional entropy produced by each state as follows:
\begin{align}
    H(X|E) = H(E|X) + H(X) - H(E) \label{conditional-entropy-decomp}
\end{align}
and verify that each term in \eqref{conditional-entropy-decomp} is identical for the two states $\bar{\rho}_+$ and $\bar{\rho}_-$. To begin with, we know that the Shannon entropy $H(X)$ is given by:
\begin{equation}
    H(X)=h(\Pr[A=0]), \label{H(X)-appe}
\end{equation}
where $h(\cdot)$ is the binary entropy, defined as: $h(p)=-p\log_2 p - (1-p)\log_2 (1-p)$. Since we proved that the statistics generated by $\bar{\rho}_+$ and $\bar{\rho}_-$ are the same, it follows that:
\begin{equation}
    H(X)_{\bar{\rho}_+} = H(X)_{\bar{\rho}_-}  . \label{H(X)-proved}
\end{equation}
In order to compute the other two terms in \eqref{conditional-entropy-decomp}, we write $\bar{\rho}_+$ and $\bar{\rho}_-$ in their spectral decomposition:
\begin{align}
    \bar{\rho}_{\pm} \sum_{i,j,k=0}^1 \rho_{ijk} \ketbra{\rho_{ijk}^\pm}{\rho_{ijk}^\pm}, \label{spec-decomp}
\end{align}
where $\rho_{ijk}$ are the states' eigenvalues, which one can easily verify to be identical for the two states, while $\ket{\rho_{ijk}^\pm}$ are the normalized eigenvectors, expressed for simplicity in terms of the following non-normalized eigenvectors:
\begin{widetext}
\begin{align}
\ket{\tilde{\rho}_{ijk}^\pm} &= \frac{\lambda_{0jk}-\lambda_{1jk} - (-1)^i \sqrt{4r_{jk}^2+(\lambda_{0jk}-\lambda_{1jk})^2}}{\pm 2 r_{jk}} \ket{\psi_{0,j,k}} +\ket{\psi_{1,j,k}} \nonumber\\
&\equiv \pm f^i_{jk}  \ket{\psi_{0,j,k}} + \ket{\psi_{1,j,k}} \label{eigenstates1}\quad (j,k)\neq (1,1) \\
\ket{\tilde{\rho}_{i11}^\pm} &= (\pm r_{11} + \mathbbm{i}s) \frac{\lambda_{011}-\lambda_{111} - (-1)^i \sqrt{4r_{11}^2+4s^2+(\lambda_{011}-\lambda_{111})^2}}{ 2 (r_{11}^2+s^2)} \ket{\psi_{0,1,1}} +\ket{\psi_{1,1,1}} \nonumber\\
&\equiv (\pm g^i_{11} +\mathbbm{i}h^i_{11}) \ket{\psi_{0,1,1}} + \ket{\psi_{1,1,1}} .\label{eigenstates2}
\end{align}
\end{widetext}
Since $\bar{\rho}_+$ and $\bar{\rho}_-$ have the same eigenvalues, it holds that:
\begin{equation}
    H(ABC)_{\bar{\rho}_+} = H(ABC)_{\bar{\rho}_-}.
\end{equation}
Assuming that Eve holds the purification
\begin{equation}
    \ket{\bar{\phi}^\pm_{ABCE}} = \sum_{i,j,k=0}^1 \sqrt{\rho_{ijk}} \ket{\rho_{ijk}^\pm} \otimes \ket{e_{ijk}} \label{purification}
\end{equation}
of the parties' state, where $\{\ket{e_{ijk}}\}$ is an orthonormal basis in $E$, it follows that:
\begin{equation}
    H(E)_{\bar{\rho}_+} = H(E)_{\bar{\rho}_-}. \label{H(E)-proved}
\end{equation}
The remaining term in \eqref{conditional-entropy-decomp} is $H(E|X)$, which is computed on the c.q. state:
\begin{widetext}
\begin{align}
    \bar{\rho}^\pm_{XE}&= \sum_{a=0}^1 \ketbra{a}{a} \otimes \sum_{\substack{i,j,k=0 \\ l,m,n=0}}^1 \sqrt{\rho_{ijk}\rho_{lmn}} \Tr_{BC}\left[\bra{a}\ket{\rho_{ijk}^\pm} \bra{\rho_{lmn}^\pm}\ket{a} \right] \ketbra{e_{ijk}}{e_{lmn}} \nonumber\\
    &\equiv \sum_{a=0}^1 \Pr[A=a] \ketbra{a}{a} \otimes \rho^{a,\pm}_E \label{H(E|X)-appendix},
\end{align}
\end{widetext}
where $\rho^{a,\pm}_E$ is the conditional state of Eve, given that Alice obtained outcome $a$. By employing the expressions in Eqs.~\eqref{eigenstates1} and \eqref{eigenstates2}, one can verify that the operators $\rho^{a,\pm}_E$ are one the transpose of the other: $\rho^{a,+}_E=(\rho^{a,-}_E)^T$. Thus $\rho^{a,+}_E$ and $\rho^{a,-}_E$ have the same eigenvalues, which implies that:
\begin{equation}
    H(\rho^{a,+}_E) = H(\rho^{a,-}_E) 
\end{equation}
Finally, since the conditional entropy $H(E|X)$ is computed as follows on the classical quantum states in \eqref{H(E|X)-appendix}:
\begin{equation}
    H(E|X)_{\bar{\rho}_\pm} = \sum_{a=0}^1 \Pr[A=a] H(\rho^{a,\pm}_E),
\end{equation}
we conclude that:
\begin{equation}
    H(E|X)_{\bar{\rho}_+} = H(E|X)_{\bar{\rho}_-} \label{H(E|X)-proved}.
\end{equation}
By combining the results in Eqs.~\eqref{H(X)-proved}, \eqref{H(E)-proved} and \eqref{H(E|X)-proved} into \eqref{same-uncertainty}, we verified that the states $\bar{\rho}_+$ and $\bar{\rho}_-$ lead to the same conditional entropy.

The final part of the proof shows that Eve's uncertainty in preparing the mixture $\rho_\alpha$ \eqref{mixture-rhoprime} does not increase with respect to preparing one of the two states $\bar{\rho}_\pm$:
\begin{equation}
    H(X|E_{\mathrm{tot}})_{\rho_\alpha} \leq H(X|E)_{\bar{\rho}_+}. \label{non-incresing-uncertainty}
\end{equation}
In this way we can guarantee that it is not restrictive to assume that Eve prepares the mixture \eqref{mixture-rhoprime}.
In giving Eve maximum power, we assume that she prepares the following global pure state (similarly to Sec.~\ref{sec:proof-symmetrization}):
\begin{equation}
    \ket{\phi_{ABCEMM'}} \frac{1}{\sqrt{2}} \sum_{m=+,-}\ket{\bar{\phi}^m_{ABCE}} \otimes \ket{m}_M \otimes \ket{m}_{M'} \label{global-pure},
\end{equation}
where $\ket{\bar{\phi}^\pm_{ABCE}}$ are the purifications of the individual states $\bar{\rho}_\pm$ defined in \eqref{purification}, while $M$ is an ancillary system informing Eve on which of the two purified states she prepared and $M'$ is the purifying system of the global state. Therefore, Eve has maximum power and her quantum system comprises: $E_{\mathrm{tot}}=EMM'$. Naturally, it holds that:
\begin{equation}
    \rho_\alpha= \Tr_{E_{\mathrm{tot}}}\left[ \ketbra{\phi_{ABCEMM'}}{\phi_{ABCEMM'}}\right].
\end{equation}
For the strong subadditivity property, we have that:
\begin{align}
    H(X|E_{\mathrm{tot}})_{\rho_\alpha} &= H(X|EMM')_{\rho_\alpha} \leq H(X|EM)_{\rho_\alpha} \nonumber\\
    &= \frac{1}{2} \sum_{m=+,-} H(X|E)_{\bar{\rho}_m},  \label{strong-subadd}
\end{align}
where the last equality is due to the fact that the state $\Tr_{M'}[\ket{\phi_{ABCEMM'}}\bra{\phi_{ABCEMM'}}]$ is classical on $M$. Finally, by employing the result \eqref{same-uncertainty} into \eqref{strong-subadd}, we obtain the claim in \eqref{non-incresing-uncertainty}. This concludes the proof.
The same argument can be used to generalized the proof for the case of $N$ parties and for the conditional entropy of the joint outcome of more than one party.

\section{MAXIMAL MABK VIOLATION BY AN $N$-QUBIT STATE: PROOF} \label{sec:theorem-Nparties}
\setcounter{equation}{0}

\noindent Here we provide the full proof of Theorem~\ref{theorem-Nparties} and of Lemma~\ref{thm_two-largest-eigenvalues}, which combined provide an analytical upper bound on the maximal violation of the $N$-partite MABK inequality by an arbitrary $N$-qubit state, for rank-one projective measurements. This is, to our knowledge, the only  existing upper bound on the violation of an $N$-partite Bell inequality by an $N$-qubit state, expressed as a function of the state's parameters. In Ref.~\cite{SS19} the authors only conjectured a bound for the $N$-party case based on their result valid for three parties. Analogously to the three-party case (see Sec.~\ref{sec:methods}), our $N$-partite bound is tight on a broader class of states than the bound conjectured in Ref.~\cite{SS19}.

We start by proving Lemma~\ref{thm_two-largest-eigenvalues}, which plays an important role in the proof of Theorem~\ref{theorem-Nparties}. We report the Lemma's statement for clarity.\smallskip\\

\noindent\textbf{Lemma~2.} \textit{Let $Q$ be an $m\times n$ real matrix and let $\norm{\vec{v}}$ be the Euclidean norm of vectors  $\vec{v}\in\mathbb{R}^k$, for $k={m,n}$. Finally, let ``$\cdot$'' indicate both the scalar product and the matrix-vector multiplication. Then
\begin{equation}
    \max_{\stackrel[\norm{\vec{c}}=\norm{\vec{c}\,'}=1]{\vec{c}\perp\vec{c}\,' \,\mathrm{s.t.}}{}} \left[\norm{Q\cdot \vec{c}}^2+ \norm{Q\cdot \vec{c}\,'}^2\right] = u_1 + u_2 \,\,,  \label{two-largest-eigenvalues-app}
\end{equation}
where $u_1$ and $u_2$ are the largest and second-to-the-largest eigenvalues of $U \equiv Q^T Q$, respectively.}\smallskip\\

\begin{proof}
Note that $U$ is a symmetric $n\times n$ real matrix, thus it can be diagonalized. The eigenvalue equation for $U$ reads:
\begin{equation}
    U \cdot \vec{u}_i = u_i \, \vec{u}_i \quad i=1,\dots,n \,\,,
\end{equation}
where the set of eigenvectors forms an orthonormal basis of $\mathbb{R}^n$: $\vec{u}_i^T \cdot\vec{u}_j =\delta_{i,j}$ and without loss of generality we ordered the eigenvalues as: $u_1 \geq u_2 \geq \dots \geq u_n \geq 0$. Note that every eigenvalue is non-negative:
\begin{align}
    u_i &= \vec{u}_i^T \cdot U \cdot \vec{u}_i = \vec{u}_i^T \cdot Q^T Q \cdot \vec{u}_i  = \norm{Q \cdot \vec{u}_i}^2 \geq 0 \nonumber \,\,.
\end{align}
By considering that: $\norm{Q\cdot \vec{c}}^2=\vec{c}\,^T \cdot Q^T Q \cdot \vec{c} =\vec{c}\,^T \cdot U \cdot\vec{c}$ and by expressing the vectors $\vec{c}$ and $\vec{c}\,'$ in the eigenbasis of $U$:
\begin{align}
    \vec{c} &= \sum_{i=1}^{n} c_i \vec{u}_i \nonumber\\
    \vec{c}\,' &= \sum_{i=1}^{n} c_i' \vec{u}_i \nonumber \,\,,
\end{align}
we can recast the claim in \eqref{two-largest-eigenvalues-app} as follows:
\begin{equation}
    \max_{\stackrel[\norm{\vec{c}}=\norm{\vec{c}\,'}=1]{\vec{c}\perp\vec{c}\,' \,\mathrm{s.t.}}{}} \sum_{i=1}^{n} u_i (c^2_i+ {c_i'}^2)   =u_1 +u_2  \,\,. \label{newclaim}
\end{equation}
Let us consider the most general scenario in which some of the eigenvalues of $U$ are degenerate: $u_1 \geq u_2 = u_3 = \dots =u_d > u_{d+1} \geq \dots u_n \geq 0$, where $d=2,\dots,n$. Note that we also account for the possibility that $u_1=u_2$.\\
We are now going to prove \eqref{newclaim} by showing that for any couple of mutually-orthogonal unit vectors $\vec{c}$ and $\vec{c}\,'$ the left-hand-side of \eqref{newclaim} is upper bounded by $u_1+u_2$ and that the bound is tight.\\
We start by considering two unit vectors in $\mathbb{R}^n$:
\begin{equation}
    \left\{{\begin{array}{l}
		\vec{c}\,^T = (c_1,\dots,c_n) \quad \mbox{s.t.}\,\,\norm{\vec{c}}^2=1 \\
		{\vec{c}\,'}^T = (c_1',\dots,c_n') \quad \mbox{s.t.}\,\,\norm{\vec{c}\,'}^2=1 \,\,,\\
		\end{array}}
	\right. \label{unitvec}
\end{equation}
and we define two unit vectors $\vec{v},\vec{w}\in\mathbb{R}^{d-1}$ along the directions individuated by $(c_2,\dots,c_d)$ and $(c_2',\dots,c_d')$, i.e.:
\begin{align}
    c_v \vec{v}\,^T &\equiv (c_2,\dots,c_d) \nonumber\\
    c_w' \vec{w}\,^T &\equiv (c_2',\dots,c_d') \label{vandw} \,\,,
\end{align}
where $c_v$ and $c_w'$ are the norms of $(c_2,\dots,c_d)$ and $(c_2',\dots,c_d')$, respectively. For $d=2$ we simply have that $c_v \vec{v}\,^T = c_2$ and $c_w' \vec{w}\,^T=c_2'$.\\
With an abuse of notation, we can rewrite \eqref{unitvec} as:
\begin{equation}
    \left\{{\begin{array}{l}
		\vec{c}\,^T = (c_1,c_v \vec{v}\,^T,c_{d+1},\dots,c_n) \quad \mbox{s.t.}\,\,c_1^2+c_v^2+r=1 \\
		{\vec{c}\,'}^T = (c_1',c_w' \vec{w}\,^T,c_{d+1}',\dots,c_n') \, \mbox{s.t.}\,{c_1'}^2+ {c_w'}^2 + r' =1 ,\\
		\end{array}}
	\right. \label{unitvecnew}
\end{equation}
where $r \equiv \sum_{i=d+1}^n c_i^2$ and $r' \equiv \sum_{i=d+1}^n {c_i'}^2$ and for both holds that: $0 \leq r \leq 1$ and $0 \leq r' \leq 1$. From the orthogonality condition $\vec{c}\,^T \cdot \vec{c}\,'=0$ we get that:
\begin{equation}
    \abs{c_1 c_1'} = \abs{\sum_{i=2}^n c_i c_i'}  \label{orthogonality} \,\,,
\end{equation}
and from the Cauchy-Schwarz inequality we deduce that:
\begin{equation}
    \abs{\sum_{i=2}^n c_i c_i'} \leq \sqrt{(c_v^2 +r)({c_w'}^2 +r')} \label{CS}\,\,.
\end{equation}
By employing \eqref{orthogonality}, \eqref{CS} and the normalization conditions in \eqref{unitvecnew}, we show that $c_1^2 + {c_1'}^2 \leq 1$ holds:
\begin{align}
    \left(c_1^2+{c_1'}^2\right)^2 &= c_1^4 + {c_1'}^4 +2c_1^2 {c_1'}^2 \nonumber\\
    &\leq c_1^4 + {c_1'}^4 +2 \left(1-c_1^2\right)\left(1-{c_1'}^2\right) \nonumber\\
    &= 2 -2 \left(c_1^2+{c_1'}^2\right) + \left(c_1^2+{c_1'}^2\right)^2 \,\,.
\end{align}
By comparing the left-hand-side with the right-hand-side one gets the desired result:
\begin{equation}
    c_1^2 + {c_1'}^2 \leq 1 \label{c1smallerthan1} \,\,.
\end{equation}
We now prove the claim in \eqref{newclaim} through the following chain of equalities and inequalities:
\begin{widetext}
\begin{align}
     \sum_{i=1}^{n} u_i (c^2_i+ {c_i'}^2) &= u_1(c_1^2+{c_1'}^2) +u_2 (c_v^2 + {c_w'}^2) + \sum_{i=d+1}^{n} u_i (c^2_i+ {c_i'}^2) \nonumber\\
     &\leq u_1(c_1^2+{c_1'}^2) +u_2 (1-r-c_1^2 + 1-r'-{c_1'}^2) +  u_{d+1} (r+r') \nonumber\\
     &= u_2 + (u_1 -u_2)(c_1^2+{c_1'}^2) + u_2 -(r+r')( u_2 - u_{d+1}) \nonumber\\
     &\leq u_1 + u_2 \,\,, \label{claim}
\end{align}
\end{widetext}
where we used the normalization conditions and the fact that the eigenvalues are ordered in descending order for the first inequality, and we used \eqref{c1smallerthan1} together with the fact that $r,r' \geq 0$ for the second inequality.

We are left to show that \eqref{claim} is tight, that is there exist unit vectors $\vec{c}$ and $\vec{c}\,'$ for which the equality sign holds. If $u_1=u_2$, the upper bound is attained when $r=r'=0$. Thus the most general pair of vectors satisfying \eqref{newclaim} is given by:
\begin{equation}
    \left\{{\begin{array}{l}
		\vec{c}\,^T = (\vec{V}\,^T,0,\dots,0)  \\
		{\vec{c}\,'}^T = (\vec{W}\,^T,0,\dots,0) \\
		\end{array}} \quad,
		\right. \label{solution1}
\end{equation}
with $\vec{V},\vec{W} \in\mathbb{R}^d$ such that $\norm{\vec{V}}=\norm{\vec{W}}=1$ and  $\vec{V}\cdot\vec{W}=0$.

If instead $u_1>u_2$, the upper bound is attained when $r=r'=0$ and $c_1^2 + {c_1'}^2 = 1$. The second condition is verified when the equality holds in \eqref{CS}, which in turn happens when the unit vectors $\vec{v}$ and $\vec{w}$ are parallel. Thus the most general pair of vectors satisfying \eqref{newclaim} is given by:
\begin{equation}
    \left\{{\begin{array}{l}
		\vec{c}\,^T = (c_1,c_v\vec{v}\,^T,0,\dots,0)  \\
		{\vec{c}\,'}^T = (c_1',c_w' \vec{v}\,^T,0,\dots,0) \\
		\end{array}} \,\,,\,  \vec{v} \in\mathbb{R}^{d-1} \,\,\wedge\,\, c_1^2 + {c_1'}^2 = 1 ,
	\right. \label{solution2}
\end{equation}
and where the orthogonality and normalization conditions hold: $c_1 c_1' + c_v c_w' =0$, $c_1^2+c_v^2 =1$ and ${c_1'}^2 + {c_w'}^2 =1$. Such solutions can always be parametrized as follows:
\begin{equation}
    \left\{{\begin{array}{l}
		\vec{c}\,^T = (\cos\alpha ,\sin\alpha\,\vec{v}\,^T,0,\dots,0)  \\
		{\vec{c}\,'}^T = (-\sin\alpha,\cos\alpha\, \vec{v}\,^T,0,\dots,0) \\
		\end{array}} \,\,,\,\alpha\in \mathbb{R}\,\,.
	\right. \label{solution2param}
\end{equation}
This concludes the proof.
\end{proof}

We are now ready to prove Theorem~\ref{theorem-Nparties}.

\begin{proof}
Firstly, we present closed expressions for the $N$-partite MABK operator, defined recursively in Definition~2. In particular, in Ref.~\cite{HolzComment} the explicit expression of the $N$-partite MABK operator when $N$ is odd is given:
\begin{align}
    M_N^{\mathrm{odd}}= \frac{1}{2^{\frac{N-3}{2}}}\sum_{\mathbf{x}\in \mathcal{L}_N} (-1)^{\frac{1}{2}\left(\frac{N-1}{2}-\omega(\mathbf{x})\right)} \bigotimes_{i=1}^N A^{(i)}_{x_i},  \label{Modd}
\end{align}
where $A_{0}^{(i)}$ and $A_{1}^{(i)}$ are the two binary observables of $\mathrm{Alice}_i$, while $\mathbf{x}=(x_1,\dots,x_N)$ is a bit string with Hamming weight given by:
\begin{equation}
    \omega(\mathbf{x})= |\{1\leq i \leq N| x_i=1\}| , \label{hamming-weight}
\end{equation}
and the set $\mathcal{L}_N$ is defined as follows:
\begin{equation}
    \mathcal{L}_N = \left\lbrace \mathbf{x}\in \{0,1\}^N \Big|  \omega(\mathbf{x}) = \frac{N-1}{2} \bmod 2\right\rbrace.  \label{LN}
\end{equation}
By applying once the MABK recursive formula of Definition~2 on \eqref{Modd}, one obtains an explicit expression of the $N$-partite MABK operator for $N$ even. We distinguish the case $N/2$ even:
\begin{equation}
    M_N^{\overline{\mathrm{even}}}= \frac{1}{2^{\frac{N-2}{2}}}\sum_{\mathbf{x}\in \{0,1\}^N} (-1)^{\frac{N}{4}-\ceil[\big]{\frac{\omega(\mathbf{x})}{2}}} \bigotimes_{i=1}^N A^{(i)}_{x_i},  \label{Meven-top}
\end{equation}
and the case $N/2$ odd:
\begin{equation}
    M_N^{\underline{\mathrm{even}}}= \frac{1}{2^{\frac{N-2}{2}}}\sum_{\mathbf{x}\in \{0,1\}^N} (-1)^{\frac{N-2}{4}-\floor[\big]{\frac{\omega(\mathbf{x})}{2}}} \bigotimes_{i=1}^N A^{(i)}_{x_i}, \label{Meven-bottom}
\end{equation}
where $\ceil{a}$ and $\floor{a}$ are the ceiling and floor functions, respectively.

We now derive an explicit expression of the MABK expectation value for a generic $N$-qubit state. As shown above, the $N$-party MABK operator can be written in explicit form as follows:
\begin{equation}
    M_N = \frac{1}{\mathcal{N}_N} \sum_{\mathbf{x}\in \mathcal{S}_N} (-1)^{\xi_N(\mathbf{x})} \bigotimes_{i=1}^N A^{(i)}_{x_i}, \label{MN-1}
\end{equation}
where the normalization factor $\mathcal{N}_N$, the set of $N$-bit strings $\mathcal{S}_N$ and the exponent $\xi_N(\mathbf{x})$ depend on the parity of $N$ (e.g. $\mathcal{S}_N=\{0,1\}^N$ for $N$ even and $\mathcal{S}_N=\mathcal{L}_N$ for $N$ odd). By assumption we restrict to rank-one projective measurements, hence every observable $A^{(i)}_{x_i}$ can be individuated by a unit vector $\vec{a}^{\,i}_{x_i}\in \mathbb{R}^3$ such that: 
\begin{equation}
    A^{(i)}_{x_i} = \vec{a}^{\,i}_{x_i}\cdot\vec{\sigma} = \sum_{\nu_i=1}^3 a^i_{x_i,\nu_i} \sigma_{\nu_i},  \label{observablePi}
\end{equation}
where $\vec{\sigma}=(X,Y,Z)^T$. By substituting \eqref{observablePi} into \eqref{MN-1} are by rearranging the terms we get:
\begin{align}
     M_N &= \frac{1}{\mathcal{N}_N} \sum_{\nu_1,\dots,\nu_N=1}^3 \left[ \sum_{\mathbf{x}\in \mathcal{S}_N} (-1)^{\xi_N(\mathbf{x})}\prod_{i=1}^N a^i_{x_i,\nu_i} \right] \nonumber\\
     &\sigma_{\nu_1}\otimes \dots \otimes \sigma_{\nu_N} \nonumber\\
     &\equiv \frac{1}{\mathcal{N}_N} \sum_{\nu_1,\dots,\nu_N=1}^3 M_{\nu_1,\dots,\nu_N} \sigma_{\nu_1}\otimes \dots \otimes \sigma_{\nu_N}. \label{MN-2}
\end{align}
\begin{widetext}
We now employ \eqref{MN-2} and the expression for a generic $N$-qubit state:
\begin{equation}
    \rho= \frac{1}{2^N} \sum_{\mu_1\dots \mu_N=0}^3 \Lambda_{\mu_1\dots \mu_N} \sigma_{\mu_1}\otimes \dots \otimes \sigma_{\mu_N}, \label{N-qubit state}
\end{equation}
to derive an explicit expression for the MABK expectation value as follows:
\begin{align}
    \braket{M_N}_\rho &= \Tr[M_N \rho]=\frac{1}{\mathcal{N}_N}\sum_{\nu_1,\dots,\nu_N=1}^3 M_{\nu_1,\dots,\nu_N} \Lambda_{\nu_1\dots \nu_N}  \nonumber\\
    &= \frac{1}{\mathcal{N}_N}\sum_{\nu_1,\dots,\nu_N=1}^3 \left[\sum_{\mathbf{x}\in \mathcal{S}_N} (-1)^{\xi_N(\mathbf{x})}\,\, a^1_{x_1,\nu_1}\cdot \dotsc \cdot a^N_{x_N,\nu_N} \,\, \Lambda_{\nu_1\dots \nu_N}  \right]\label{MABK-value},
\end{align}
where we used the fact that $\Tr[\sigma_i \sigma_j]=2\delta_{i,j}$.

We now specify the expressions for $\mathcal{N}_N, \mathcal{S}_N$ and $\xi_N(\mathbf{x})$ when $N/2$ is even and prove the theorem's statement in this particular case. However, a similar procedure applies to the $N/2$ odd and $N$ odd cases and leads to the same final result.

We thus have the following expression for the MABK expectation value:
\begin{equation}
    \braket{M_N}_\rho= \frac{1}{2^{\frac{N-2}{2}}}\sum_{\nu_1,\dots,\nu_N=1}^3 \left[\sum_{\mathbf{x}\in \{0,1\}^N} (-1)^{\frac{N}{4}-\ceil[\big]{\frac{\omega(\mathbf{x})}{2}}}\,\, a^1_{x_1,\nu_1}\cdot \dotsc \cdot a^N_{x_N,\nu_N} \,\, \Lambda_{\nu_1\dots \nu_N}  \right] \label{MABK-value-1},
\end{equation}
and we rearrange it as follows:
\begin{align}
    \braket{M_N}_\rho &= \frac{1}{2^{\frac{N-2}{2}}}\sum_{\nu_1,\dots,\nu_N=1}^3 \left[\sum_{\mathbf{x}\in \{0,1\}^{N/2}}\sum_{\mathbf{y}\in \{0,1\}^{N/2}} (-1)^{\frac{N}{4}-\ceil[\big]{\frac{\omega(\mathbf{x})+\omega(\mathbf{y})}{2}}}  \right. \nonumber\\
    &\left. a^1_{x_1,\nu_1}\cdot \dotsc \cdot a^{N/2}_{x_{N/2},\nu_{N/2}} \,\, \Lambda_{\nu_1\dots \nu_N} \,\,  a^{N/2}_{y_1,\nu_{N/2}}\cdot \dotsc \cdot a^{N}_{y_{N/2},\nu_N} \right] \nonumber\\
    &=  \frac{1}{2^{\frac{N-2}{2}}}\sum_{\nu_1,\dots,\nu_N=1}^3 \left[\sum_{\mathbf{x}\in\mathcal{E}_{N/2}}\sum_{\mathbf{y}\in \{0,1\}^{N/2}} (-1)^{\frac{N}{4}-\ceil[\big]{\frac{\omega(\mathbf{x})+\omega(\mathbf{y})}{2}}}  \right. \nonumber\\
    &\left. a^1_{x_1,\nu_1}\cdot \dotsc \cdot a^{N/2}_{x_{N/2},\nu_{N/2}} \,\, \Lambda_{\nu_1\dots \nu_N} \,\,  a^{N/2}_{y_1,\nu_{N/2}}\cdot \dotsc \cdot a^{N}_{y_{N/2},\nu_N} \right. \nonumber\\
    &\left.+\sum_{\mathbf{x}\in\mathcal{O}_{N/2}}\sum_{\mathbf{y}\in \{0,1\}^{N/2}} (-1)^{\frac{N}{4}-\ceil[\big]{\frac{\omega(\mathbf{x})+\omega(\mathbf{y})}{2}}}  \right. \nonumber\\
    &\left. a^1_{x_1,\nu_1}\cdot \dotsc \cdot a^{N/2}_{x_{N/2},\nu_{N/2}} \,\, \Lambda_{\nu_1\dots \nu_N} \,\,  a^{N/2}_{y_1,\nu_{N/2}}\cdot \dotsc \cdot a^{N}_{y_{N/2},\nu_N} \right] \label{MABK-value-2},
\end{align}
where the sets $\mathcal{E}_{N/2}$ and $\mathcal{O}_{N/2}$ are defined as follows:
\begin{align}
    \mathcal{E}_{N/2} &= \left\lbrace \mathbf{x}\in\{0,1\}^{N/2} \big\vert \omega(\mathbf{x}) \bmod 2 =0\right\rbrace  \label{even-Hamming} \\
    \mathcal{O}_{N/2} &= \left\lbrace \mathbf{x}\in\{0,1\}^{N/2} \big\vert \omega(\mathbf{x}) \bmod 2 =1\right\rbrace  \label{odd-Hamming} .
\end{align}
We basically split the bit strings $\mathbf{x}$ into those with an even Hamming weight and those with an odd Hamming weight. Now considering that the following identity holds:
\begin{equation}
    \ceil[\Bigg]{\frac{\omega(\mathbf{x})+\omega(\mathbf{y})}{2}} = \left\{{\begin{array}{l}
		\omega(\mathbf{x}) \mbox{ even:}\quad \floor[\big]{\frac{\omega(\mathbf{x})}{2}} +\floor[\big]{\frac{\omega(\mathbf{y})}{2}} + (\omega(\mathbf{y}) \bmod 2)   \\[1ex]
		\omega(\mathbf{x}) \mbox{ odd:}\quad \floor[\big]{\frac{\omega(\mathbf{x})}{2}} +\floor[\big]{\frac{\omega(\mathbf{y})}{2}} + 1  \label{Hamming-relation},
		\end{array}}\right.
\end{equation}
we can recast the MABK expectation value in \eqref{MABK-value-2} as follows:
\begin{align}
    \braket{M_N}_\rho &= \frac{1}{2^{\frac{N-2}{2}}}\sum_{\nu_1,\dots,\nu_N=1}^3 \left[\left(\sum_{\mathbf{x}\in\mathcal{E}_{N/2}}(-1)^{\frac{N}{4}-\floor[\big]{\frac{\omega(\mathbf{x})}{2}}}\,\,  a^1_{x_1,\nu_1}\cdot \dotsc \cdot a^{N/2}_{x_{N/2},\nu_{N/2}} \right) \right. \nonumber\\
    &\left.  \Lambda_{\nu_1\dots \nu_N} \, \left(\sum_{\mathbf{y}\in \{0,1\}^{N/2}} (-1)^{\floor[\big]{\frac{\omega(\mathbf{y})}{2}} + (\omega(\mathbf{y}) \bmod 2)}\,\, a^{N/2}_{y_1,\nu_{N/2}}\cdot \dotsc \cdot a^{N}_{y_{N/2},\nu_N} \right) \right. \nonumber\\
    &\left.+\left(\sum_{\mathbf{x}\in\mathcal{O}_{N/2}} (-1)^{\frac{N}{4}-\floor[\big]{\frac{\omega(\mathbf{x})}{2}}}\,\, a^1_{x_1,\nu_1}\cdot \dotsc \cdot a^{N/2}_{x_{N/2},\nu_{N/2}}\right) \right. \nonumber\\
    &\left.  \Lambda_{\nu_1\dots \nu_N} \,\left(\sum_{\mathbf{y}\in \{0,1\}^{N/2}} (-1)^{\floor[\big]{\frac{\omega(\mathbf{y})}{2}} + 1}\,\, a^{N/2}_{y_1,\nu_{N/2}}\cdot \dotsc \cdot a^{N}_{y_{N/2},\nu_N}\right) \right] \nonumber\\
    &\equiv \frac{1}{2^{\frac{N-2}{2}}} \left[\vec{v}_0^T \cdot T_{\rho}\cdot \vec{u}_0 +\vec{v}_1^T \cdot T_{\rho}\cdot \vec{u}_1 \right]\label{MABK-value-3}.
\end{align}
\end{widetext}
In the expression \eqref{MABK-value-3} we defined the vectors:
\begin{align}
    \vec{v}_0 &= \sum_{\mathbf{x}\in\mathcal{E}_{N/2}}(-1)^{\frac{N}{4}-\floor[\big]{\frac{\omega(\mathbf{x})}{2}}}\,\, \bigotimes_{i=1}^{N/2} \vec{a}^{\,i}_{x_i},   \label{v0} \\
    \vec{v}_1 &= \sum_{\mathbf{x}\in\mathcal{O}_{N/2}}(-1)^{\frac{N}{4}-\floor[\big]{\frac{\omega(\mathbf{x})}{2}}}\,\, \bigotimes_{i=1}^{N/2} \vec{a}^{\,i}_{x_i},   \label{v1} \\
    \vec{u}_0 &= \sum_{\mathbf{y}\in \{0,1\}^{N/2}} (-1)^{\floor[\big]{\frac{\omega(\mathbf{y})}{2}} + (\omega(\mathbf{y}) \bmod 2)}\,\, \bigotimes_{i=1}^{N/2} \vec{a}^{\,N/2+i}_{y_i}, \label{u0} \\
    \vec{u}_1 &= \sum_{\mathbf{y}\in \{0,1\}^{N/2}} (-1)^{\floor[\big]{\frac{\omega(\mathbf{y})}{2}} + 1}\,\, \bigotimes_{i=1}^{N/2} \vec{a}^{\, N/2+i}_{y_i},  \label{u1}
\end{align}
and we used Definition~3 of the correlation matrix of an $N$-qubit state.
The $3^{N/2}$-dimensional vectors in \eqref{v0},\eqref{v1},\eqref{u0} and \eqref{u1} are heavily constrained by their tensor-product structure and satisfy the following properties:
\begin{align}
    &\mbox{Prop. 1:}\quad \norm{\vec{v}_0}^2 +\norm{\vec{v}_1}^2 = 2^{N/2} \label{prop1}\\
    &\mbox{Prop. 2:}\quad \norm{\vec{u}_0}^2 =\norm{\vec{u}_1}^2 = 2^{N/2} \label{prop2}\\
    &\mbox{Prop. 3:}\quad \vec{v}_0 \cdot \vec{v}_1 =0 \label{prop3}.
\end{align}
These properties play a fundamental role in deriving a meaningful upper bound on the MABK expectation value.

We prove the first property \eqref{prop1} by directly computing the l.h.s.: 
\begin{align}
    &\norm{\vec{v}_0}^2 +\norm{\vec{v}_1}^2 = \vec{v}_0 \cdot \vec{v}_0 + \vec{v}_1 \cdot \vec{v}_1 \nonumber\\
    &= \sum_{\mathbf{x},\mathbf{y}\in\mathcal{E}_{N/2}}(-1)^{\frac{N}{2}-\floor[\big]{\frac{\omega(\mathbf{x})}{2}} - \floor[\big]{\frac{\omega(\mathbf{y})}{2}}}\,\, \prod_{i=1}^{N/2} (\cos\theta_i)^{x_i \oplus y_i} \nonumber\\
    &+ \sum_{\mathbf{x},\mathbf{y}\in\mathcal{O}_{N/2}}(-1)^{\frac{N}{2}-\floor[\big]{\frac{\omega(\mathbf{x})}{2}} - \floor[\big]{\frac{\omega(\mathbf{y})}{2}}}\,\, \prod_{i=1}^{N/2} (\cos\theta_i)^{x_i \oplus y_i},  \label{v0pv1}
\end{align}
where we used the fact that $\vec{a}^{\,i}_{x_i}$ are unit vectors and we called $\theta_i$ the angle between the two measurement directions of party number $i$: $\cos\theta_i=\vec{a}^{\,i}_{0}\cdot \vec{a}^{\,i}_{1}$. Note that the symbol $\oplus$ is the binary operation XOR. We now define the bit string: $\mathbf{r}=\mathbf{x}\oplus \mathbf{y}$, whose Hamming weight can be computed as:
\begin{equation}
    \omega(\mathbf{r})= \omega(\mathbf{x}\oplus \mathbf{y})=\omega(\mathbf{x}) + \omega(\mathbf{y}) -2 \omega(\mathbf{x}\wedge \mathbf{y}), \label{H(r)}
\end{equation}
where $\wedge$ is the binary operation AND. From \eqref{H(r)} it follows immediately that the Hamming weight of the string $\mathbf{r}$ is always even, since the Hamming weights of $\mathbf{x}$ and $\mathbf{y}$ are either both even or both odd in \eqref{v0pv1}. With this information, we can recast \eqref{v0pv1} as follows:
\begin{widetext}
\begin{align}
    \norm{\vec{v}_0}^2 +\norm{\vec{v}_1}^2 &=  \sum_{\mathbf{r}\in\mathcal{E}_{N/2}}\left[\sum_{\mathbf{y}\in\mathcal{E}_{N/2}}(-1)^{\floor[\big]{\frac{\omega(\mathbf{r}\oplus\mathbf{y})}{2}} + \floor[\big]{\frac{\omega(\mathbf{y})}{2}}}+  \sum_{\mathbf{y}\in\mathcal{O}_{N/2}}(-1)^{\floor[\big]{\frac{\omega(\mathbf{r}\oplus\mathbf{y})}{2}} + \floor[\big]{\frac{\omega(\mathbf{y})}{2}}}\right] \prod_{i=1}^{N/2} (\cos\theta_i)^{r_i},  \label{v0pv1-1}
\end{align}
where we used the fact that $N/2$ is even and where the string $\mathbf{x}$ is completely fixed once $\mathbf{r}$ and $\mathbf{y}$ are fixed: $\mathbf{x}=\mathbf{r}\oplus\mathbf{y}$. Now we employ the relation \eqref{H(r)} in \eqref{v0pv1-1} and we make use of the information on the parity of the Hamming weights appearing in the two sums:
\begin{align}
    \norm{\vec{v}_0}^2 +\norm{\vec{v}_1}^2 &=  \sum_{\mathbf{r}\in\mathcal{E}_{N/2}}\left[\sum_{\mathbf{y}\in\mathcal{E}_{N/2}}(-1)^{\frac{\omega(\mathbf{r})}{2}+ \omega(\mathbf{y}) - \omega(\mathbf{r}\wedge \mathbf{y})} +  \sum_{\mathbf{y}\in\mathcal{O}_{N/2}}(-1)^{\frac{\omega(\mathbf{r})+\omega(\mathbf{y}) -2\omega(\mathbf{r}\wedge \mathbf{y})-1}{2} + \frac{\omega(\mathbf{y})-1}{2}}\right] \nonumber\\
    &\times\prod_{i=1}^{N/2} (\cos\theta_i)^{r_i}.  \label{v0pv1-2}
\end{align}
Note that $\floor[\big]{a/2}=(a-1)/2$ if $a$ is an odd number. The expression in \eqref{v0pv1-2} can be further simplified by considering that even addends in the exponents of $(-1)$ can be ignored:
\begin{align}
    \norm{\vec{v}_0}^2 +\norm{\vec{v}_1}^2 &=  \sum_{\mathbf{r}\in\mathcal{E}_{N/2}}\left[\sum_{\mathbf{y}\in\mathcal{E}_{N/2}}(-1)^{\frac{\omega(\mathbf{r})}{2}- \omega(\mathbf{r}\wedge \mathbf{y})} +  \sum_{\mathbf{y}\in\mathcal{O}_{N/2}}(-1)^{\frac{\omega(\mathbf{r})}{2} -\omega(\mathbf{r}\wedge \mathbf{y})}\right] \prod_{i=1}^{N/2} (\cos\theta_i)^{r_i} \nonumber\\
    &= \sum_{\mathbf{r}\in\mathcal{E}_{N/2}} (-1)^{\frac{\omega(\mathbf{r})}{2}}\left[\sum_{\mathbf{y}\in\{0,1\}^{N/2}}(-1)^{\omega(\mathbf{r}\wedge \mathbf{y})} \right] \prod_{i=1}^{N/2} (\cos\theta_i)^{r_i}  \nonumber\\
    &= 2^{N/2} + \sum_{\substack{\mathbf{r}\in\mathcal{E}_{N/2}\\\mathbf{r}\neq \mathbf{0}}} (-1)^{\frac{\omega(\mathbf{r})}{2}}\left[\sum_{\mathbf{y}\in\{0,1\}^{N/2}}(-1)^{\omega(\mathbf{r}\wedge \mathbf{y})} \right] \prod_{i=1}^{N/2} (\cos\theta_i)^{r_i} \label{v0pv1-3},
\end{align}
\end{widetext}
where we extracted the term $\mathbf{r}=\mathbf{0}$ from the sum in the last equality.

The last step to prove the first property \eqref{prop1} is to show that every term in the remaining sum in \eqref{v0pv1-3} is identically zero, i.e. we want to show that:
\begin{equation}
    \sum_{\mathbf{y}\in\{0,1\}^{N/2}}(-1)^{\omega(\mathbf{r}\wedge \mathbf{y})} = 0 \quad \forall \,\mathbf{r}\neq \mathbf{0}  \label{v0pv1-4}.
\end{equation}
In order for \eqref{v0pv1-4} to be verified, there must be as many $(-1)$ terms as $+1$ terms, and since there are in total $2^{N/2}$ terms, there must be exactly $2^{N/2-1}$ terms (half of the total) that are $(-1)$. We can count the number of $(-1)$ terms in \eqref{v0pv1-4} as follows:
\begin{equation}
    \sum_{\mathbf{y}\in\{0,1\}^{N/2}} (\omega(\mathbf{r}\wedge \mathbf{y}) \bmod 2) ,\label{count}
\end{equation}
and check whether it equals $2^{N/2-1}$, as claimed. Note that $\omega(\mathbf{r}\wedge \mathbf{y})$ represents the number of ones in $\mathbf{r}$ that are also in $\mathbf{y}$. The parity of this number is then summed over all the possible bit strings $\mathbf{y}$ of length $N/2$. We can thus recast the sum, as a sum over the number of ones that $\mathbf{r}$ and $\mathbf{y}$ have in common ($k$), times the number of bit strings $\mathbf{y}$ that share $k$ ones with $\mathbf{r}$:
\begin{align}
    &\sum_{\mathbf{y}\in\{0,1\}^{N/2}} (\omega(\mathbf{r}\wedge \mathbf{y}) \bmod 2) = \nonumber\\
    &\sum_{k=0}^{\omega(\mathbf{r})}(k \bmod 2) \binom{\omega(\mathbf{r})}{k} 2^{N/2-\omega(\mathbf{r})} \label{count2}.
\end{align}
Note that the number of bit strings $\mathbf{y}$ that have $k$ ones in common with a fixed string $\mathbf{r}$, is given by the number of possible combinations of $k$ ones from the total number of ones ($\omega(\mathbf{r})$) populating the string $\mathbf{r}$, times the number of possibilities ($2^{N/2-\omega(\mathbf{r})}$) that we have to fill the remaining bits of $\mathbf{y}$ that are not part of the $k$ ones in common with $\mathbf{r}$.

We can now adjust the r.h.s. of \eqref{count2} to the following computable form:
\begin{align}
    2^{N/2-\omega(\mathbf{r})} \sum_{\substack{k=0\\k \,\mathrm{odd}}}^{\omega(\mathbf{r})}\binom{\omega(\mathbf{r})}{k}  &= 2^{N/2-\omega(\mathbf{r})} \,\,2^{\omega(\mathbf{r})-1} \nonumber\\
    &= 2^{N/2-1}, \label{count3}
\end{align}
where the first equality is obtained by combining two known facts about the binomial coefficient, namely:
\begin{align}
    &\sum_{k=0}^n \binom{n}{k} = 2^n \label{binom1} \\
    &\sum_{k=0}^n (-1)^k \binom{n}{k} = 0 \label{binom2}.
\end{align}
Indeed, by subtracting \eqref{binom2} from \eqref{binom1} one gets that:
\begin{equation}
    \sum_{k \,\mathrm{odd}} \binom{n}{k} = 2^{n-1},  \label{binom3}
\end{equation}
which is used in the first equality in \eqref{count3}.

Combining \eqref{count2} and \eqref{count3} we conclude that \eqref{v0pv1-4} is verified. We have thus shown the validity of the first property \eqref{prop1}.

We move on to prove the second property \eqref{prop2}. We start from \eqref{u0} and use the fact that $\floor[\big]{\frac{\omega(\mathbf{x})}{2}}+(\omega(\mathbf{x})\bmod 2)=(\omega(\mathbf{x})+[\omega(\mathbf{x})\bmod 2)]/2)$. We obtain:
\begin{widetext}
\begin{align}
    \norm{\vec{u}_0}^2 &=  \sum_{\mathbf{x},\mathbf{y}\in\{0,1\}^{N/2}} (-1)^{\frac{\omega(\mathbf{x})+ (\omega(\mathbf{x}) \bmod 2)}{2} +\frac{\omega(\mathbf{y}) + (\omega(\mathbf{y}) \bmod 2)}{2}}\,\,\prod_{i=1}^{N/2} (\cos(\theta_{N/2+i}))^{x_i \oplus y_i} \nonumber\\
    &= \sum_{\mathbf{r}\in\{0,1\}^{N/2}} \left[\sum_{\mathbf{y}\in\{0,1\}^{N/2}}  (-1)^{\frac{\omega(\mathbf{r})+\omega(\mathbf{y}) -2\omega(\mathbf{r}\wedge \mathbf{y}) + (\omega(\mathbf{r})+\omega(\mathbf{y}) \bmod 2)}{2} +\frac{\omega(\mathbf{y}) + (\omega(\mathbf{y}) \bmod 2)}{2}} \right]  \nonumber\\
    &\times \prod_{i=1}^{N/2} (\cos(\theta_{N/2+i}))^{r_i}, \label{u0sq}
\end{align}
where we defined $\mathbf{r}=\mathbf{x}\oplus \mathbf{y}$ and we used the relation \eqref{H(r)}.
We proceed to simplify \eqref{u0sq} by splitting the sum over $\mathbf{y}$ over the strings with even and odd Hamming weight:
\begin{align}
    \norm{\vec{u}_0}^2 &= \sum_{\mathbf{r}\in\{0,1\}^{N/2}} \left[\sum_{\mathbf{y}\in\mathcal{E}_{N/2}}  (-1)^{\frac{\omega(\mathbf{r}) + (\omega(\mathbf{r}) \bmod 2)}{2} -\omega(\mathbf{r}\wedge \mathbf{y})}  \right.\nonumber\\
    &\left.+\sum_{\mathbf{y}\in\mathcal{O}_{N/2}}  (-1)^{\frac{\omega(\mathbf{r})+1 + (\omega(\mathbf{r})+1 \bmod 2)}{2} -\omega(\mathbf{r}\wedge \mathbf{y})+\omega(\mathbf{y})}  \right] \prod_{i=1}^{N/2} (\cos(\theta_{N/2+i}))^{r_i} . \label{u0sq-1}
\end{align}
By employing the following identities:
\begin{align}
    \frac{\omega(\mathbf{r}) + (\omega(\mathbf{r}) \bmod 2)}{2} &= \ceil[\Bigg]{\frac{\omega(\mathbf{r})}{2}} \\
    (-1)^a&=(-1)^1 \quad a\mbox{ odd} \\
    \frac{\omega(\mathbf{r})+1 + (\omega(\mathbf{r})+1 \bmod 2)}{2} &= \ceil[\Bigg]{\frac{\omega(\mathbf{r})+1}{2}} = \ceil[\Bigg]{\frac{\omega(\mathbf{r})}{2}} +1 - (\omega(\mathbf{r}) \bmod 2)
\end{align}
into \eqref{u0sq-1} we obtain:
\begin{align}
    \norm{\vec{u}_0}^2 &= \sum_{\mathbf{r}\in\{0,1\}^{N/2}} (-1)^{\ceil[\big]{\frac{\omega(\mathbf{r})}{2}}} \left[\sum_{\mathbf{y}\in\mathcal{E}_{N/2}}  (-1)^{\omega(\mathbf{r}\wedge \mathbf{y})}  \right.\nonumber\\
    &+\left.\sum_{\mathbf{y}\in\mathcal{O}_{N/2}}  (-1)^{\omega(\mathbf{r}\wedge \mathbf{y})+(\omega(\mathbf{r}) \bmod 2)}  \right] \prod_{i=1}^{N/2} (\cos(\theta_{N/2+i}))^{r_i} \nonumber\\
    &= 2^{N/2}+ \sum_{\substack{\mathbf{r}\in\{0,1\}^{N/2} \\ \mathbf{r}\neq \mathbf{0}}} (-1)^{\ceil[\big]{\frac{\omega(\mathbf{r})}{2}}} \left[\sum_{\mathbf{y}\in\mathcal{E}_{N/2}}  (-1)^{\omega(\mathbf{r}\wedge \mathbf{y})}  \right.\nonumber\\
    &+\left.\sum_{\mathbf{y}\in\mathcal{O}_{N/2}}  (-1)^{\omega(\mathbf{r}\wedge \mathbf{y})+(\omega(\mathbf{r}) \bmod 2)}  \right] \prod_{i=1}^{N/2} (\cos(\theta_{N/2+i}))^{r_i} \nonumber\\
    &= 2^{N/2}+ \sum_{\substack{\mathbf{r}\in\mathcal{E}_{N/2}\\ \mathbf{r}\neq \mathbf{0}}} (-1)^{\ceil[\big]{\frac{\omega(\mathbf{r})}{2}}} \left[\sum_{\mathbf{y}\in\{0,1\}^{N/2}}  (-1)^{\omega(\mathbf{r}\wedge \mathbf{y})}  \right] \prod_{i=1}^{N/2}(\cos(\theta_{N/2+i}))^{r_i} \nonumber\\
    &+ \sum_{\mathbf{r}\in\mathcal{O}_{N/2}} (-1)^{\ceil[\big]{\frac{\omega(\mathbf{r})}{2}}} \left[\sum_{\mathbf{y}\in\mathcal{E}_{N/2}}  (-1)^{\omega(\mathbf{r}\wedge \mathbf{y})} - \sum_{\mathbf{y}\in\mathcal{O}_{N/2}}  (-1)^{\omega(\mathbf{r}\wedge \mathbf{y})} \right]\prod_{i=1}^{N/2}(\cos(\theta_{N/2+i}))^{r_i}, \label{u0sq-2}
\end{align}
\end{widetext}
where we isolated the $\mathbf{r}=\mathbf{0}$ term in the second equality and we split the sum over $\mathbf{r}$ in two sums over the strings with even and odd Hamming weights in the third equality.

The first sum in \eqref{u0sq-2} is zero thanks to \eqref{v0pv1-4}. From \eqref{v0pv1-4} we also deduce that:
\begin{equation}
    \sum_{\mathbf{y}\in\mathcal{E}_{N/2}}  (-1)^{\omega(\mathbf{r}\wedge \mathbf{y})} + \sum_{\mathbf{y}\in\mathcal{O}_{N/2}}  (-1)^{\omega(\mathbf{r}\wedge \mathbf{y})}=0, 
\end{equation}
which means that the term in square brackets in the second sum can be reduced to:
\begin{equation}
    2 \sum_{\mathbf{y}\in\mathcal{E}_{N/2}}  (-1)^{\omega(\mathbf{r}\wedge \mathbf{y})} =0 \label{u0sq-3}.
\end{equation}
The proof that \eqref{u0sq-3} holds is analogous to that of \eqref{v0pv1-4}. In particular, \eqref{u0sq-3} is verified if the number of $(-1)$ terms is exactly half the total number of terms, that is $2^{N/2-2}$. We show that this is true by computing the number of $(-1)$ terms as follows:
\begin{align}
    &\sum_{\mathbf{y}\in\mathcal{E}_{N/2}} (\omega(\mathbf{r}\wedge \mathbf{y}) \bmod 2) = \nonumber\\
    &\sum_{k=0}^{\omega(\mathbf{r})}(k \bmod 2) \binom{\omega(\mathbf{r})}{k} 2^{N/2-\omega(\mathbf{r})-1} \nonumber\\
    &=2^{N/2-\omega(\mathbf{r})-1} \sum_{\substack{k=0\\ k\,\mathrm{odd}}}^{\omega(\mathbf{r})} \binom{\omega(\mathbf{r})}{k}. \label{countnew}
\end{align}
Note that this time, compared to \eqref{count2}, the number of possibilities ($2^{N/2-\omega(\mathbf{r})-1}$) to fill the non-fixed bits of $\mathbf{y}$ is halved. The reason is that in this case $\mathbf{y}$ is constrained to have an even number of ones, thus after fixing $N/2-1$ of its bits, no degree of freedom is left.

By employing again the result on binomial distributions \eqref{binom3} in \eqref{countnew}, we obtain:
\begin{align}
    \sum_{\mathbf{y}\in\mathcal{E}_{N/2}} (\omega(\mathbf{r}\wedge \mathbf{y}) \bmod 2) =2^{N/2-2} , \label{countnew2}
\end{align}
which proves \eqref{u0sq-3}.

We have thus shown that both the sums in \eqref{u0sq-2} are zero, thus proving the second property \eqref{prop2} for $\vec{u}_0$. The proof of \eqref{prop2} for $\vec{u}_1$ is analogous and we omit it.

Finally we show that the third property \eqref{prop3} is satisfied by direct computation:
\begin{align}
    &\vec{v}_0 \cdot \vec{v}_1 = \sum_{\substack{\mathbf{x}\in\mathcal{E}_{N/2} \\ \mathbf{y}\in\mathcal{O}_{N/2}}} (-1)^{\floor[\big]{\frac{\omega(\mathbf{x})}{2}} + \floor[\big]{\frac{\omega(\mathbf{y})}{2}}}\,\, \prod_{i=1}^{N/2} (\cos\theta_i)^{x_i \oplus y_i} \nonumber\\
    &= \sum_{\substack{\mathbf{x}\in\mathcal{E}_{N/2} \\ \mathbf{y}\in\mathcal{O}_{N/2}}} (-1)^{\frac{\omega(\mathbf{x})}{2} + \frac{\omega(\mathbf{y})-1}{2} }\,\, \prod_{i=1}^{N/2} (\cos\theta_i)^{x_i \oplus y_i} \nonumber\\
    &= \sum_{\mathbf{r}\in\mathcal{O}_{N/2}}\left[\sum_{\mathbf{x}\in\mathcal{E}_{N/2}} (-1)^{\frac{\omega(\mathbf{x})-1+\omega(\mathbf{x})+\omega(\mathbf{r})-2\omega(\mathbf{r}\wedge \mathbf{x})}{2} } \right]\nonumber\\
    &\prod_{i=1}^{N/2} (\cos\theta_i)^{r_i}
\end{align}
where we defined $\mathbf{r}=\mathbf{x}\oplus\mathbf{y}$ and used \eqref{H(r)}.
By simplifying the last expression we get:
\begin{align}
    &\vec{v}_0 \cdot \vec{v}_1 = \nonumber\\
    &\sum_{\mathbf{r}\in\mathcal{O}_{N/2}} (-1)^{\frac{\omega(\mathbf{r})-1}{2}} \left[\sum_{\mathbf{x}\in\mathcal{E}_{N/2}} (-1)^{\omega(\mathbf{r}\wedge \mathbf{x})} \right]\prod_{i=1}^{N/2} (\cos\theta_i)^{r_i} \nonumber\\
    &=0,
\end{align}
where we used \eqref{u0sq-3} to prove the final equality.

Thanks to the properties \eqref{prop1}, \eqref{prop2} and \eqref{prop3}, we can express the vectors $\vec{v}_k$ and $\vec{u}_k$ ($k=0,1$) as follows:
\begin{align}
    \vec{v}_0&=2^{N/4} \cos\theta \,\hat{v}_0 \label{newv0}\\
    \vec{v}_1&=2^{N/4} \sin\theta \,\hat{v}_1 \label{newv1}\\
    \vec{u}_k&=2^{N/4} \,\hat{u}_k \label{newui}
\end{align}
where $\hat{v}_k$ and $\hat{u}_k$ are unit vectors in the directions of $\vec{v}_k$ and $\vec{u}_k$, respectively, and where $\theta$ is a real number. With the expressions \eqref{newv0}, \eqref{newv1} and \eqref{newui} we recast the MABK expectation value \eqref{MABK-value-3} as follows:
\begin{align}
    \braket{M_N}_\rho &= \frac{2^{N/2}}{2^{\frac{N-2}{2}}} \left[\cos\theta\, \hat{v}_0^T \cdot T_{\rho}\cdot \hat{u}_0 +\sin\theta\, \hat{v}_1^T \cdot T_{\rho}\cdot \hat{u}_1 \right] \nonumber\\
    &= 2 \left[\cos\theta\, \hat{v}_0^T \cdot T_{\rho}\cdot \hat{u}_0 +\sin\theta\, \hat{v}_1^T \cdot T_{\rho}\cdot \hat{u}_1 \right] \label{MABK-value-3.5}.
\end{align}
The maximal violation $\mathcal{M}_\rho$ of the $N$-partite MABK inequality is then obtained by maximizing \eqref{MABK-value-3.5} over all the parties' measurements directions $\vec{a}^{\,i}_{0}$ and $\vec{a}^{\,i}_{1}$ (for $i=1,\dots,N$). A valid upper bound on the maximal violation $\mathcal{M}_\rho$ is thus given by:
\begin{equation}
    \mathcal{M}_\rho \leq \max_{\substack{\hat{v}_k,\hat{u}_k,\theta \\ \hat{v}_0 \perp \hat{v}_1}} 2  \left[\cos\theta\, \hat{v}_0^T \cdot T_{\rho}\cdot \hat{u}_0 +\sin\theta \,\hat{v}_1^T \cdot T_{\rho}\cdot \hat{u}_1 \right] ,  \label{MABK-value-4}
\end{equation}
where the inequality is due to the fact that we are now optimizing the expectation value over all the possible unit vectors $\hat{v}_k$ (such that $\hat{v}_0 \cdot \hat{v}_1 =0$) and $\hat{u}_k$, and freely over $\theta$, ignoring the more stringent structures \eqref{v0}-\eqref{u1} characterizing these vectors and their relation to $\theta$.
By choosing $\hat{u}_0$ and $\hat{u}_1$ in the direction of $\hat{v}_0^T \cdot T_{\rho}$ and $\hat{v}_1^T \cdot T_{\rho}$, respectively, and by fixing $\theta$ such that:
\begin{equation}
    \tan\theta= \frac{\norm{T^T_\rho \cdot \hat{v}_1}}{\norm{T^T_\rho \cdot \hat{v}_0}}, \label{thetafixed}
\end{equation}
we can simplify the maximization in \eqref{MABK-value-4} as follows:
\begin{align}
    \mathcal{M}_\rho  &\leq \max_{\substack{\hat{v}_k,\hat{u}_k,\theta \\ \hat{v}_0 \perp \hat{v}_1}} 2  \left[\cos\theta\, \hat{v}_0^T \cdot T_{\rho}\cdot \hat{u}_0 +\sin\theta \,\hat{v}_1^T \cdot T_{\rho}\cdot \hat{u}_1 \right] \nonumber\\
    &= \max_{\substack{\hat{v}_k,\theta \\ \hat{v}_0 \perp \hat{v}_1}} 2  \left[\cos\theta\, \norm{T^T_\rho \cdot \hat{v}_0} +\sin\theta \,\norm{T^T_\rho \cdot \hat{v}_1} \right] \nonumber\\
    &= \max_{\substack{\hat{v}_k \\ \hat{v}_0 \perp \hat{v}_1}} 2  \sqrt{\norm{T^T_\rho \cdot \hat{v}_0}^2 +\norm{T^T_\rho \cdot \hat{v}_1}^2}.  \label{MABK-value-5}
\end{align}
Finally, by employing the result of Lemma~\ref{thm_two-largest-eigenvalues} in \eqref{MABK-value-5}, we obtain the statement \eqref{NMABKviolationbound} of the theorem:
\begin{equation}
    \mathcal{M}_\rho  \leq \sqrt{t_0 + t_1},
\end{equation}
where $t_0$ and $t_1$ are the two largest eigenvalues of $T_\rho T_\rho^T$. This concludes the proof.
\end{proof}

\subsection{Tightness conditions}
\noindent Here we derive the conditions for which the upper bound on the MABK violation given in \eqref{NMABKviolationbound} is tight. That is, there exist observables for the $N$ parties such that the violation achieved on the state $\rho$ is exactly given by the r.h.s. of \eqref{NMABKviolationbound}. We first address the case $N/2$ even since it is the one explicitly derived in the proof, then we present the tightness conditions valid in the other cases.

The bound is tight when equality holds in \eqref{MABK-value-4}. Considering that we made specific choices for the unit vectors $\hat{v}_i$ and $\hat{u}_i$ and for $\theta$, the vectors in \eqref{v0}-\eqref{u1} should comply with these specific choices. In particular, consider the eigenvalue equation for $T_\rho T_\rho^T$ with normalized eigenvectors and where $t_0$ and $t_1$ are the two largest eigenvalues:
\begin{equation}
    T_\rho T_\rho^T \hat{t}_k = t_k \hat{t}_k.  \label{eigenval-eq}
\end{equation}
In order to use Lemma~\ref{thm_two-largest-eigenvalues} in \eqref{MABK-value-5}, it must hold that:
\begin{equation}
    \hat{v}_k=\frac{\vec{v}_k}{\norm{\vec{v}_k}} = \hat{t}_k \quad k=0,1, \label{tightcond1}
\end{equation}
where $\vec{v}_k$ ($k=0,1$) are defined in \eqref{v0} and \eqref{v1}. Employing \eqref{tightcond1} into the relation \eqref{thetafixed} that fixes $\theta$ we get:
\begin{equation}
    \frac{\norm{\vec{v}_1}}{\norm{\vec{v}_0}}= \tan\theta = \frac{\norm{T^T_\rho \cdot \hat{t}_1}}{\norm{T^T_\rho \cdot \hat{t}_0}}= \sqrt{\frac{t_1}{t_0}}, \label{tightcond2}
\end{equation}
where the last equality is due to \eqref{eigenval-eq}. Combining \eqref{tightcond2} with property \eqref{prop1} we completely fix the norms of vectors $\vec{v}_0$ and $\vec{v}_1$, while their direction is already fixed by \eqref{tightcond1}. In conclusion we get the following tightness conditions for $\vec{v}_0$ and $\vec{v}_1$, which we recall being specific combinations \eqref{v0} and \eqref{v1} of the parties' measurement directions:
\begin{equation}
    \vec{v}_k = 2^{N/4} \sqrt{\frac{t_k}{t_0 + t_1}} \hat{t}_k  \quad k=0,1  \label{tightcond-v}.
\end{equation}

In addition to this, we also fixed the directions $\hat{u}_0$ and $\hat{u}_1$ to those of $T^T_\rho \cdot \hat{v}_0$ and $T^T_\rho \cdot \hat{v}_1$, respectively. Due to \eqref{tightcond1} and recalling property \eqref{prop2}, we derive the following tightness conditions on $\vec{u}_0$ and $\vec{u}_1$:
\begin{equation}
    \vec{u}_k = \frac{2^{N/4}}{\sqrt{t_k}} T^T_\rho \hat{t}_k \quad k=0,1 .\label{tightcond-u}
\end{equation}
One can verify that upon substituting the tightness conditions \eqref{tightcond-v} and \eqref{tightcond-u} into the MABK expectation value \eqref{MABK-value-3}, the theorem claim is obtained.\bigskip\\

Here we recapitulate the tightness conditions of Theorem~\ref{theorem-Nparties} for the two cases $N$ even and $N$ odd. The bound in \eqref{NMABKviolationbound} is tight if there exist unit vectors $\vec{a}^{\,i}_{0},\vec{a}^{\,i}_{1}$ (with $i=1,\dots,N$) such that:
\begin{itemize}
    \item \emph{$N$ even}:
    \begin{equation}
        \vec{v}_k = 2^{N/4} \sqrt{\frac{t_k}{t_0 + t_1}} \hat{t}_k \,,\, \vec{u}_k = \frac{2^{N/4}}{\sqrt{t_k}} T^T_\rho \hat{t}_k \quad (k=0,1) \label{tightcondNeven}
    \end{equation}
    where vectors $\vec{v}_k$ and $\vec{u}_k$ are defined in \eqref{v0}-\eqref{u1} if $N/2$ is even, or as follows if $N/2$ is odd:
    \begin{align}
    \vec{v}_0 &= \sum_{\mathbf{x}\in\mathcal{E}_{N/2}}(-1)^{\frac{N-2}{4}-\floor[\big]{\frac{\omega(\mathbf{x})}{2}}}\,\, \bigotimes_{i=1}^{N/2} \vec{a}^{\,i}_{x_i},    \\
    \vec{v}_1 &= \sum_{\mathbf{x}\in\mathcal{O}_{N/2}}(-1)^{\frac{N-2}{4}-\floor[\big]{\frac{\omega(\mathbf{x})}{2}}}\,\, \bigotimes_{i=1}^{N/2} \vec{a}^{\,i}_{x_i},  \\
    \vec{u}_0 &= \sum_{\mathbf{y}\in \{0,1\}^{N/2}} (-1)^{\floor[\big]{\frac{\omega(\mathbf{y})}{2}}}\,\, \bigotimes_{i=1}^{N/2} \vec{a}^{\,N/2+i}_{y_i},  \\
    \vec{u}_1 &= \sum_{\mathbf{y}\in \{0,1\}^{N/2}} (-1)^{\ceil[\big]{\frac{\omega(\mathbf{y})}{2}}}\,\, \bigotimes_{i=1}^{N/2} \vec{a}^{\,N/2+i}_{y_i},
    \end{align}
    where the sets $\mathcal{E}_{N/2}$ and $\mathcal{O}_{N/2}$ are defined in \eqref{even-Hamming} and \eqref{odd-Hamming}, respectively.
    \item \emph{$N$ odd}:
    \begin{equation}
        \vec{v}_k = 2^{(N+1)/4} \sqrt{\frac{t_k}{t_0 + t_1}} \hat{t}_k \,,\, \vec{u}_k = \frac{2^{(N-3)/4}}{\sqrt{t_k}} T^T_\rho \hat{t}_k \, (k=0,1) \label{tightcondNodd}
    \end{equation}
    where vectors $\vec{v}_k$ and $\vec{u}_k$ are defined as follows:
    \begin{align}
    \vec{v}_0 &= \sum_{\mathbf{x}\in\mathcal{E}_{(N+1)/2}}(-1)^{\floor[\big]{\frac{N-1}{4}}-\floor[\big]{\frac{\omega(\mathbf{x})}{2}}}\,\, \bigotimes_{i=1}^{(N+1)/2} \vec{a}^{\,i}_{x_i},    \\
    \vec{v}_1 &= \sum_{\mathbf{x}\in\mathcal{O}_{(N+1)/2}}(-1)^{\floor[\big]{\frac{N-1}{4}}-\floor[\big]{\frac{\omega(\mathbf{x})}{2}}}\,\, \bigotimes_{i=1}^{(N+1)/2} \vec{a}^{\,i}_{x_i},  \\
    \vec{u}_0 &= \sum_{\mathbf{y}\in\mathcal{J}_{(N-1)/2}} (-1)^{\floor[\big]{\frac{\omega(\mathbf{y})}{2}}}\,\, \bigotimes_{i=1}^{(N-1)/2} \vec{a}^{\,(N+1)/2+i}_{y_i},  \\
    \vec{u}_1 &= \sum_{\mathbf{y}\in\overline{\mathcal{J}}_{(N-1)/2}} (-1)^{\ceil[\big]{\frac{\omega(\mathbf{y})}{2}}}\,\, \bigotimes_{i=1}^{(N-1)/2} \vec{a}^{\,(N+1)/2+i}_{y_i},
    \end{align}
    with the sets $\mathcal{J}_{(N-1)/2}$ and $\overline{\mathcal{J}}_{(N-1)/2}$ fixed as:
    \begin{align}
        &\mathcal{J}_{(N-1)/2} = \left\lbrace \mathbf{x}\in \{0,1\}^{(N-1)/2} \Big|  \omega(\mathbf{x}) = \frac{N-1}{2} \bmod 2\right\rbrace \\
        &\overline{\mathcal{J}}_{(N-1)/2} = \nonumber\\
        &\left\lbrace \mathbf{x}\in \{0,1\}^{(N-1)/2} \Big|  \omega(\mathbf{x})+1 = \frac{N-1}{2} \bmod 2\right\rbrace  .
    \end{align}
\end{itemize}

Note that, similarly to the $N=3$ case discussed in Sec.~\ref{sec:methods}, one can potentially obtain tighter MABK violation upper bounds accompanied by the corresponding tightness conditions if one employs variations (in terms of row and column definitions)  of the correlation matrix given in Definition~\ref{def:corrmatrix}.

\section{ANALYTICAL PROOF OF THE LOWER BOUND ON $H(X|E)_{\rho_\alpha}$} \label{sec:tau-proof}
\setcounter{equation}{0}

\noindent In this Appendix we derive the analytical solution of the optimization problem in \eqref{optimization-Alice}, which we report here for clarity:
\begin{align}
    &H(X|E)_{\rho_\alpha}^{\downarrow}(m) = \min_{\{\rho_{ijk}\}} 1 - H(\{\rho_{ijk}\}) +H(\{\rho_{ijk}+\rho_{i\bar{j}\bar{k}}\}) \nonumber\\
    &\mbox{sub. to}\quad\mathcal{M}^{\uparrow}_\alpha \geq m \,;\, \rho_{0jk}\geq\rho_{1jk} \,;\, \sum_{ijk} \rho_{ijk}=1\,;\,\rho_{ijk}\geq 0, \label{optimization-problem}
\end{align}
where the upper bound on the MABK violation is given in Corollary~\ref{cor:almostGHZMermin}, where the second constraint is given in \eqref{ordered-rhoijk} and where $m\geq 2\sqrt{2}$, otherwise the conditional entropy is null (see Fig.~\ref{plot-Hlower}). For ease of notation, we dropped the subscript $\alpha$ in the observed violation and we will indicate the objective function of the optimization problem as $H(X|E)_{\rho_\alpha}$.

Because of the symmetry of the problem, we can assume w.l.o.g. that the largest element in $\{\rho_{ijk}\}$ is $\rho_{000}$. Then, a necessary condition such that $\mathcal{M}^{\uparrow}_\alpha\geq  2\sqrt{2}$ is given by $\rho_{000}\geq 1/2$. Indeed, the following upper bound on $\mathcal{M}^{\uparrow}_\alpha$:
\begin{align}
     \mathcal{M}^{\uparrow}_\alpha &= 4\sqrt{ \sum_{j,k=0}^1 (\rho_{0jk} -\rho_{1jk})^2} \leq  4\sqrt{ \sum_{j,k=0}^1 \rho_{0jk}^2} \nonumber\\
     &\leq 4\sqrt{ \sum_{j,k=0}^1 \rho_{000}\cdot\rho_{0jk}} = 4\sqrt{ \rho_{000}\left(\sum_{j,k=0}^1 \rho_{0jk}\right)} \nonumber\\
     &\leq 4\sqrt{ \rho_{000}} \label{necessary-cond-MABKupp},
\end{align}
is greater than or equal to $2\sqrt{2}$ when $\rho_{000}\geq 1/2$.

Note that, by definition, the minimal entropy $H(X|E)_{\rho_\alpha}^{\downarrow}(m)$ in \eqref{optimization-problem} is monotonically non-decreasing in $m$.

The upper bound on the maximal MABK violation \eqref{MABKupp-almostGHZ3} is tight on the following class of states (the tightness conditions are verified):
\begin{align}
   \tau(\nu)=\nu \ketbra{\psi_{0,0,0}}{\psi_{0,0,0}}+(1-\nu)\ketbra{\psi_{0,1,1}}{\psi_{0,1,1}}, \label{tau}
\end{align}
and reads in this case
\begin{align}
  \mathcal{M}_{\tau}(\nu)=\mathcal{M}^\uparrow_{\tau}(\nu)=&4\sqrt{\nu^2+(1-\nu)^2}. \label{Mtau}
\end{align}
It is straightforward to verify that
\begin{equation}
    \mathcal{M}_{\tau}(\rho_{000}) \geq \mathcal{M}^{\uparrow}_\alpha \quad \forall\, \{\rho_{ijk}\}. \label{greaterM}
\end{equation}
Moreover, the objective function of the minimization, when evaluated on the states \eqref{tau}, reads:
\begin{equation}
    H(X|E)_{\tau (\nu)}= 1-h(\nu), \label{Htau}
\end{equation}
where we used the binary entropy ${h(p)=-p\log p-(1-p)\log(1-p)}$. Here and in the following, ``$\log$'' represents the logarithm in base 2.

By definition, the entropy minimized over all the states with $\mathcal{M}^\uparrow_\alpha \geq m$ \eqref{optimization-problem} is upper bounded by the entropy of any particular state with $\mathcal{M}^\uparrow_\alpha = m$:
\begin{equation}
    H(X|E)_{\rho_\alpha}^{\downarrow}(m) \leq H(X|E)_{\tau (\nu_m)} \label{ineq1}
\end{equation}
where $\nu_m$ is fixed such that the maximal violation of the state $\tau(\nu_m)$ is given by $m$:
\begin{equation}
    \mathcal{M}^\uparrow_{\tau}(\nu_m)=4\sqrt{\nu_m^2+(1-\nu_m)^2}=m. \label{nu-m}
\end{equation}

On the other hand, in the following we prove that:
\begin{equation}
    H(X|E)_{\rho_\alpha} \geq H(X|E)_{\tau (\rho_{000})} \quad \forall\, \{\rho_{ijk}\}, \label{proof}
\end{equation}
where $\rho_{000}\geq 1/2$ is the largest element in $\{\rho_{ijk}\}$. In particular, the last expression holds for the state $\rho_\alpha^*$ which is the solution of the minimization in \eqref{optimization-problem}:
\begin{align}
    H(X|E)_{\rho_\alpha}^{\downarrow}(m) =H(X|E)_{\rho_\alpha^*}
    &\geq H(X|E)_{\tau (\rho_{000}^*)} \nonumber\\
    &\geq H(X|E)_{\tau (\nu_m)}\label{ineq2}.
\end{align}
The last inequality in \eqref{ineq2} is due to a couple of observations. Firstly,
by applying \eqref{greaterM} to the state $\rho_\alpha^*$ we obtain $\mathcal{M}_\tau(\rho_{000}^*)\geq m$, which combined with \eqref{nu-m} implies that $\rho_{000}^*\geq \nu_m$ (in the interval of interest $\rho_{000}^*,\nu_m \geq 1/2$). Then, we observe that the entropy of the states $\tau$ in \eqref{Htau} is monotonically increasing in the interval $\nu\in[1/2,1]$. The two observations lead to the second inequality in \eqref{ineq2}.

By combining \eqref{ineq2} with \eqref{ineq1}, we obtain the desired lower bound:
\begin{equation}
    H(X|E)_{\rho_\alpha}^{\downarrow}(m) = H(X|E)_{\tau (\nu_m)}. \label{finalbound}
\end{equation}
Note that the family of states $\tau(\nu)$ in \eqref{tau} minimizes the entropy for every observed violation $m$. The bound in \eqref{finalbound} can be expressed in terms of the violation $m$ by reverting \eqref{nu-m} and by using it in \eqref{Htau}, thus obtaining \eqref{tight-H}.\vspace{1cm}

We are thus left to prove the inequality in \eqref{proof}, which can be recast as follows:
\begin{equation}
     h(\rho_{000})+H(\{\rho_{ijk}+\rho_{i\bar j\bar k}\})-H(\{\rho_{ijk}\})\geq 0 \label{toprove}.
\end{equation}
To start with, we simplify the difference of the following entropies:
\begin{align}
     h(\rho_{000})-H(\{\rho_{ijk}\}) &= -(1-\rho_{000})\log(1-\rho_{000})\nonumber\\
     &+\sum_{(i,j,k)\neq(0,0,0)}\rho_{ijk}\log\rho_{ijk}. \label{diff}
\end{align}
By substituting \eqref{diff} into the l.h.s. of \eqref{toprove}, we get:
\begin{align}
    &H(X|E)_{\rho_\alpha}- H(X|E)_{\tau (\rho_{000})} = H(\{\rho_{ijk}+\rho_{i\bar j\bar k}\}) \nonumber\\
    &+\sum_{(i,j,k)\neq(0,0,0)}\rho_{ijk}\log\rho_{ijk}-(1-\rho_{000})\log(1-\rho_{000}).\label{Eq-EntrDiff3}
\end{align}
We then apply Jensen's inequality
\begin{align}
  f(x+y)\geq \frac{f(2x)+f(2y)}{2},
\end{align}
where $f(x)=-x\log x$ is a concave function, to the last three terms of the first entropy in \eqref{Eq-EntrDiff3}:
\begin{align}
&H(\{\rho_{ijk}+\rho_{i\bar j\bar k}\})=-(\rho_{000}+\rho_{011})\log(\rho_{000}+\rho_{011}) \nonumber\\
&-(\rho_{001}+\rho_{010})\log(\rho_{001}+\rho_{010}) \nonumber\\
&-(\rho_{100}+\rho_{111})\log(\rho_{100}+\rho_{111}) \nonumber\\
&-(\rho_{101}+\rho_{110})\log(\rho_{101}+\rho_{110}),
\end{align}
such that we get
\begin{align}
&H(\{\rho_{ijk}+\rho_{i\bar j\bar k}\}) \geq-(\rho_{000}+\rho_{011})\log(\rho_{000}+\rho_{011})\nonumber\\
&+ \sum_{(i,j,k)\neq{(0,0,0)\atop (0,1,1)}}-\rho_{ijk}\log(2\rho_{ijk})\nonumber\\
=&-(\rho_{000}+\rho_{011})\log(\rho_{000}+\rho_{011})-(1-\rho_{000}-\rho_{011}) \nonumber\\
&+\sum_{(i,j,k)\neq{(0,0,0)\atop (0,1,1)}}-\rho_{ijk}\log\rho_{ijk}.
\end{align}
With this result, the difference of entropies in \eqref{Eq-EntrDiff3} can be estimated by
\begin{align}
 &H(X|E)_{\rho_\alpha}- H(X|E)_{\tau (\rho_{000})} \geq \nonumber\\
 &-(\rho_{000}+\rho_{011})\log(\rho_{000}+\rho_{011})\nonumber\\
 &-(1-\rho_{000})\log(2(1-\rho_{000})) \nonumber\\
 &+\rho_{011}\log(2\rho_{011})\nonumber\\
 &=: g(\rho_{000},\rho_{011}).\label{Eq-EstimateEntrDiff}
\end{align}
In the function $g$ the first two terms are positive and the last is negative. 
We further analyze and estimate the function $g(x,y)$ in the range of interest, i.e.\ $1/2\leq x\leq 1$, $0\leq y\leq 1-x$. In this range $g(x,y)$ is concave in $x$ because its second derivative is always negative:
\begin{align}
  \pppt{g(x,y)}{x}=& -\frac{1}{\ln(2)}\left(\frac{1}{(1-x)}+\frac{1}{(x+y)}\right)<0.
\end{align}
Consider the boundary $x+y=1$ of $g(x,y)$ for which we get $g(1-y,y)=0$. Due to the concavity it holds for $0\leq p\leq 1$ that: $$g\left(p\frac12+(1-p) (1-y),y\right)\geq p g\left(\frac12,y\right)+(1-p)g(1-y,y)$$ or equivalently that:
\begin{align}
g(x,y)\geq \left(\frac{1-x-y}{\frac12-y}\right) g\left(\frac12,y\right).\label{Eq-EstimateEntrDiff2}
\end{align}
Note that from the parameter regimes of $x$ and $y$ it follows that
\begin{equation}
    0\leq \left(\frac{1-x-y}{\frac12-y}\right)\leq 1. \label{prefactor}
\end{equation}
We finally analyze the properties of $g(\frac12,y)$, which is convex in $y$ as its second derivative is always positive:
 \begin{align}
     \pppt{g(\frac12,y)}{y}=&\frac{1}{y\ln(2)+y^2\ln(4)}>0.
 \end{align}
A convex function has a unique minimum if it exists in the parameter regime. In our case this is given by:
 \begin{align}
   \ppt{g(\frac12,y)}{y}=&\log(2y)-\log(\frac12+y) =0\quad \Leftrightarrow\quad
   y=\frac12
 \end{align}
for which $g(\frac12,\frac12)=0$ holds. Thus in general it holds:
\begin{align}
 g\left(\frac12,y\right)\geq 0.   \label{fpositive}
\end{align}
By combining these considerations we obtain the desired inequality \eqref{proof}: 
\begin{align}
H(X|E)_{\rho_\alpha}- &H(X|E)_{\tau (\rho_{000})}\stackrel{(\ref{Eq-EstimateEntrDiff})}{\geq}  g(\rho_{000},\rho_{011})\nonumber\\
&\stackrel{(\ref{Eq-EstimateEntrDiff2})}{\geq}  \left(\frac{1-\rho_{000}-\rho_{011}}{\frac12-\rho_{011}}\right) g\left(\frac12,\rho_{011}\right)\nonumber\\
&\hspace{0.21cm}\geq 0, 
\end{align}
where in the last inequality we used the fact that the pre-factor is positive \eqref{prefactor} and that $g(\frac12,\rho_{011})$ is lower bounded by zero \eqref{fpositive}.

\section{ANALYTICAL PROOF OF THE LOWER BOUND ON $H(XY|E)_{\rho_\alpha}$} \label{sec:eta-proof}
\setcounter{equation}{0}

\noindent In order to derive an analytical lower bound on $H(XY|E)_{\rho_\alpha}$, we solve the simplified optimization problem in \eqref{optimization-AliceandBob} where we can independently minimize the entropy over $t,\varphi_X$ and $\varphi_Y$ without affecting the MABK violation. We report \eqref{optimization-AliceandBob} for clarity:
\begin{align}
    &H(XY|E)_{\rho_\alpha}^{\downarrow}(m) = \min_{\{\rho_{ijk},t,\varphi_X,\varphi_Y\}} H(XY|E)_{\rho_\alpha} \nonumber\\
    &\mbox{sub. to}\quad\mathcal{M}^{\uparrow}_\alpha \geq m \,;\, \rho_{0jk}\geq\rho_{1jk} \,;\, \sum_{ijk} \rho_{ijk}=1\,;\,\rho_{ijk}\geq 0 , \label{optimization-problem-AliceBob}
\end{align}
where $\mathcal{M}^{\uparrow}_\alpha$ is the upper bound on the MABK violation derived in Corollary~\ref{cor:almostGHZMermin}, while $\varphi_X$ and $\varphi_Y$ are the measurement directions of the outcomes $X$ and $Y$ in the $(x,y)$-plane. We dropped the subscript in $m_\alpha$ for ease of notation.

Eve is assumed to hold the purifying system $E$ of the state $\rho_\alpha$ shared by Alice, Bob and Charlie. The purification of $\rho_\alpha$ can thus be written as follows:
\begin{align}\label{Eq-PurifiedState}
  \ket{\phi^\alpha_{ABCE}}=\sum_{ijk}\sqrt{\rho_{ijk}}\ket{\rho_{ijk}}\otimes \ket{e_{ijk}},
\end{align}
where $\ket{\rho_{ijk}}$ are the eigenstates of $\rho_\alpha$ defined in \eqref{eigenvec}, while $\{\ket{e_{ijk}}\}$ is an orthonormal basis of Eve's eight-dimensional Hilbert space $\mathcal{H}_E$.

We restrict our proof to states $\rho_\alpha$ with a non-negative off-diagonal term $s\geq 0$, which corresponds to $t \geq 0$ (see \eqref{q}). Since $t\in [-\pi/2,\pi/2]$ by definition, this means that we restrict ourselves to the region where $\sin t\geq 0$ and $\cos t\geq 0$. The complementary case corresponds to states $\rho_\alpha^*$ which would lead to the same result. For this, we employ a parametrization of the eigenstates slightly different from \eqref{eigenvec}, which reads as follows:
\begin{align}
\begin{split}
  \ket{\rho_{ijk}}=&\ket{\psi_{i,j,k}},\;\;\mbox{for}\;\; (j,k)\neq (1,1)\\
  \ket{\rho_{011}}=&\sqrt{(1-p)}\ket{\psi_{0,1,1}}-\mathbbm{i}\sqrt{p}\ket{\psi_{1,1,1}}\\
  \ket{\rho_{111}}=&\sqrt{p}\ket{\psi_{0,1,1}}+\mathbbm{i}\sqrt{(1-p)}\ket{\psi_{1,1,1}},
  \label{Eq-EigenstatesRhoLambda}
  \end{split}
\end{align}
where $\ket{\psi_{i,j,k}}$ are the GHZ basis states (Definition~1) and where $p$ is completely fixed by $t$ through the relation:
\begin{equation}
    p = \frac{(\tan t)^2}{1+(\tan t)^2} \label{p},
\end{equation}
from which we deduce that $0 \leq p \leq 1$ and that $p=0$ when $t=0$.

From now on, we omit the subscript $\rho_\alpha$ in the entropy symbol for ease of notation. We thus have that the conditional entropy $H(XY|E)$ can be expressed as:
\begin{align}
  H(XY|E)=& H(XY)+H(E|XY)-H(E)\nonumber\\
  =&2+H(E|XY)-H(\{\rho_{ijk}\}).\label{Eq-HXYgE}
\end{align}
where the last equation follows from the fact that all marginals have been symmetrized and from the fact that the state on $ABCE$ is pure \eqref{Eq-PurifiedState}, thus $H(E)=H(ABC)=H(\{\rho_{ijk}\})$.

The proof of the analytical lower bound on $H(XY|E)$ as a function of the MABK violation is subdivided in three parts:
(i) we first derive an analytical expression for $H(E|XY)$; (ii) we minimize $H(E|XY)$ with respect to $t,\varphi_X$ and $\varphi_Y$; (iii) we proceed minimizing the resulting expression of \eqref{Eq-HXYgE} given a fixed MABK violation. Note that we are allowed to minimize $H(E|XY)$ over $t,\varphi_X$ and $\varphi_Y$ independently of $H(E)$, since the latter is independent of the mentioned optimization variables.
\vspace{1em}

\noindent\textbf{Step 1 - Analytical expression for $H(E|XY)$:}

In order to derive the analytical expression for $H(E|XY)$, we will use the following Lemma.
\begin{Lmm} \label{lmm:conditional_entropy}
The following equality holds:
\begin{align}
   H(E|XY)=H(C|XY).
 \end{align}
\end{Lmm}

\begin{proof}
The proof follows from the fact that the state shared by Charlie and Eve conditioned on the outcomes $X=a$ of Alice and $Y=b$ of Bob, is a pure state. Indeed, if projective measurements are applied to a pure state, the resulting state, conditioned on a specific outcome, remains pure. Moreover, for a pure state, the entropies of its subsystems are equal, which implies
\begin{align}
    H(E|X\!=\!a,Y\!=\!b)=H(C|X\!=\!a,Y\!=\!b)
\end{align}
Therefore
\begin{align}
     H(E|XY)&=\sum_{a,b}\Pr(a,b)H(E|X=a,Y=b) \\
     &=\sum_{a,b}\Pr(a,b)H(C|X=a,Y=b) \\
     &=H(C|XY).
\end{align}
\end{proof}

Lemma~\ref{lmm:conditional_entropy} is of great use as Eve's system is described by an eight-dimensional Hilbert space, whereas Charlie is only in possession of a single qubit. So the computation of $H(C|XY)$ is significantly simpler.

We obtain Charlie's state, conditioned on the outcomes $X=a$ and $Y=b$, by partially tracing over Eve's degrees of freedom
\begin{align}
  \rho^\alpha_{C_{ab}}=\trz{E}{\proj{\phi^\alpha_{CE_{ab}}}},
\end{align}
where $\ket{\phi^\alpha_{CE_{ab}}}$ is the state of Charlie and Eve given that Alice and Bob obtain outcomes $X=a$ and $Y=b$ respectively, which is determined by
\begin{align}
&(P_{\ket{a}}\otimes P_{\ket{b}}\otimes\id_{CE})\proj{\phi^\alpha_{ABCE}}
(P_{\ket{a}}\otimes P_{\ket{b}}\otimes\id_{CE}) \nonumber\\
&=\frac{1}{4} P_{\ket{a}}\otimes P_{\ket{b}}\otimes\proj{\phi^\alpha_{CE_{ab}}}.
\end{align}
where $P_{\ket{a}}=\proj{a}$ and $P_{\ket{b}}=\proj{b}$. 
The projected state $\ket{\phi^\alpha_{CE_{ab}}}$ can be computed using the definition of the purification given in Eq.~\eqref{Eq-PurifiedState}, the definition of the eigenstates in Eq.~\eqref{Eq-EigenstatesRhoLambda},
and the fact that the measurements performed by Alice and Bob have been restricted to the $(x,y)$-plane. Indeed, the measurements are defined by the following projectors:
\begin{align}
    \ket{a}_X &=\frac{1}{\sqrt{2}} (\ket{0}+ (-1)^a e^{\mathbbm{i}\varphi_X} \ket{1}) \quad a\in \{0,1\}  \nonumber\\
    \ket{b}_Y &=\frac{1}{\sqrt{2}} (\ket{0}+ (-1)^b e^{\mathbbm{i}\varphi_Y} \ket{1}) \quad b\in \{0,1\} .
\end{align}
\newcommand{\xix}{\xi_a}
\newcommand{\xiy}{\xi_b}
In the following we abbreviate $\xix=(-1)^a\e{\mathbbm{i}\varphi_X}$ and $\xiy=(-1)^b\e{\mathbbm{i}\varphi_Y}$. We then have that
\begin{widetext}
\begin{align}
\begin{split}
\ket{\phi^\alpha_{CE_{ab}}}=&
  \sum_{{ljk\atop jk\neq 11}}\frac 1{\sqrt 2} \de{(\delta_{0j}+\delta_{1j}\xiy)\ket{k}+(\delta_{0\bar j}\xix+\delta_{1\bar j}\xix\xiy)(-1)^l\ket{\bar k}}\otimes\ket{e_{ljk}} \sqrt{\rho_{ljk}}\\
  &+\left(\left(\sqrt{(1-p)}-\mathbbm{i}\sqrt{p}\right)\xiy\ket{1}+\left(\sqrt{(1-p)}+\mathbbm{i}\sqrt{p}\right)\xix\ket{0}\right) \otimes\ket{e_{011}} \sqrt{\rho_{011}}\\
  &+\left(\left(\sqrt{p}+\mathbbm{i}\sqrt{(1-p)}\right)\xiy\ket{1}+\left(\sqrt{p}-\mathbbm{i}\sqrt{(1-p)}\right)\xix\ket{0}\right) \otimes\ket{e_{111}} \sqrt{\rho_{111}}.
  \end{split}
\end{align}
Finally, the partial trace over Eve results in
\begin{align}
   \rho^\alpha_{C_{ab}}=& \trz{E}{\proj{\phi^\alpha_{CE_{ab}}}}\\
   \begin{split}
   =&\de{ \sum_{{ljk\atop jk\neq 11}}\frac{1}{\sqrt{2}} \left((\delta_{0j}+\delta_{1j}\xiy)\ket{k}+(\delta_{0\bar j}\xix+\delta_{1\bar j}\xix\xiy)(-1)^l\ket{\bar k}\right)} \cdot\left(\mathrm{h.c.}\right)\rho_{ljk}\\
 &+\left(\left(\sqrt{(1-p)}-\mathbbm{i}\sqrt{p}\right)\xiy\ket{1}+\left(\sqrt{(1-p)}+\mathbbm{i}\sqrt{p}\right)\xix\ket{0}\right)\cdot\left(\mathrm{h.c.}\right)\rho_{011}\\
 &+\left(\left(\sqrt{p}+\mathbbm{i}\sqrt{(1-p)}\right)\xiy\ket{1}+\left(\sqrt{p}-\mathbbm{i}\sqrt{(1-p)}\right)\xix\ket{0}\right)\cdot\left(\mathrm{h.c.}\right)\rho_{111}.\label{Eq-RhoC}
 \end{split}
\end{align}

As $\rho^\alpha_{C_{ab}}$ is a qubit state, we can now analytically calculate its eigenvalues, which can be reduced to
\begin{align}
  \lambda_{1,2}(\rho^\alpha_{C_{ab}})=& \frac12\left(1\pm\abs{C}\right),\label{Eq-EigenvalueRhoC}
  \end{align}
 where:
  \begin{align}
  \begin{split}
  C=&\left(\rho_{000}-\rho_{100}\right)\xix^2+\left(\rho_{001}-\rho_{101}\right)\left(\xiy^2\right)^*+\left(\rho_{010}-\rho_{110}\right)\xix^2\left(\xiy^2\right)^*+\left(\rho_{011}-\rho_{111}\right)\left(1-2p-2\mathbbm{i}\sqrt{p(1-p)}\right)\\
  =&\left(\rho_{000}-\rho_{100}\right)\e{\mathbbm{i}2\varphi_{X}}+\left(\rho_{001}-\rho_{101}\right)\e{-\mathbbm{i}2\varphi_{Y}}+\left(\rho_{010}-\rho_{110}\right)\e{\mathbbm{i}2(\varphi_{X}-\varphi_{Y})}+\left(\rho_{011}-\rho_{111}\right)\e{\mathbbm{i}\varphi_3},
\end{split}
\end{align}
\end{widetext}
where $\varphi_3$ is a function of the parameter $p$. We see that the eigenvalues do not depend on the measurement outcomes $a$ and $b$ of Alice and Bob. The entropy is then given by
\begin{align}\label{eq:HE_XY}
  H(E|XY) = H(C|XY)= h\left(\frac12\left(1+\abs{C}\right)\right),
\end{align}
where $h(x)=-x\log x-(1-x)\log(1-x)$ is the binary entropy.
\vspace{1em}

\noindent\textbf{Step 2 - Minimization of $H(E|XY)$:}

Minimizing the binary entropy in \eqref{eq:HE_XY}, with respect to the measurement directions and the parameter $p$, is equivalent to maximizing the largest eigenvalue of $\rho^\alpha_{C_{ab}}$. The optimum can directly be deduced from Eq.~\eqref{Eq-EigenvalueRhoC}. Since it holds that $(\rho_{0jk}-\rho_{1jk})\geq 0\,\forall \,j,k$, the largest eigenvalue is maximized if
\begin{align}
  \e{\mathbbm{i}2\varphi_{X}}=\e{-\mathbbm{i}2\varphi_{Y}}=\e{\mathbbm{i}2(\varphi_{X}-\varphi_{Y})}=\e{\mathbbm{i}\varphi_3},
\end{align}
which holds e.g.\ for $\varphi_X=\varphi_Y=\varphi_3=0$. Since $\varphi_3=0$ implies $p=t=0$, we verified that in the minimization of the conditional entropy of two parties' outcomes, $H(XY|E)$, it would be optimal for Eve to distribute a GHZ-diagonal state which Alice and Bob measure in the $X$ basis. The largest eigenvalue of $\rho^\alpha_{C_{ab}}$ is then given by
\begin{align}
  \bar{\lambda}:=  \sum_{jk} \rho_{0jk}, \label{lambdabar}
\end{align}
where we used the normalization of the eigenvalues to eliminate the terms $\rho_{1jk}$. The lower bound on the conditional entropy $H(E|XY)$ is thus given by
\begin{align}
  H(E|XY)\geq & h(\bar{\lambda}). \label{step2result}
\end{align}
\vspace{1em}

\noindent\textbf{Step 3 - Minimization of $H(XY|E)$ with given MABK violation:}

By using \eqref{step2result} in \eqref{Eq-HXYgE}, we can concentrate on minimizing the following expression:
\begin{align}
  &H(E|XY)-H(E)
  \geq h(\bar{\lambda})- H(\{\rho_{ijk}\}) \nonumber\\
  &=-\bar{\lambda}\log \bar{\lambda} -(1-\bar{\lambda})\log(1-\bar{\lambda})\nonumber\\
  &+\sum_{ijk}\rho_{ijk}\log\rho_{ijk}\nonumber\\
 =&\bar{\lambda} \sum_{jk}\frac{\rho_{0jk}}{\bar{\lambda}}\log\frac{\rho_{0jk}}{\bar{\lambda}}+(1-\bar{\lambda})
 \sum_{jk}\frac{\rho_{1jk}}{(1-\bar{\lambda})}\log\frac{\rho_{1jk}}{(1-\bar{\lambda})}\nonumber\\
 =&-\bar{\lambda} H\left(\left\{\frac{\rho_{0jk}}{\bar{\lambda}}\right\}\right)-(1-\bar{\lambda}) H\left(\left\{\frac{\rho_{1jk}}{(1-\bar{\lambda})}\right\}\right), \label{concavity-step}
\end{align}
where we used the definition of $\bar{\lambda}$ in \eqref{lambdabar}.
We now use the concavity of the Shannon entropy over probability distributions $\vec{u}$ and $\vec{v}$, i.e.\
\begin{align}
  \bar{\lambda} H(\vec u)+(1-\bar{\lambda})H(\vec v)\leq H(\bar{\lambda} \vec u+(1-\bar{\lambda})\vec v),
\end{align}
in \eqref{concavity-step} and obtain
\begin{align}
   H(E|XY)-H(E)
 \geq &-H\left(\left\{\rho_{0jk}+\rho_{1jk} \right\}\right).\label{Eq-EstEntropyTwo}
\end{align}
With the lower bound obtained in \eqref{Eq-EstEntropyTwo}, the optimization problem we have to solve is now the following:
\begin{align}
\begin{split}
  \max_{\{\rho_{ijk}\}}\; &   H\left(\left\{\rho_{0jk}+\rho_{1jk} \right\}\right)\\
  \mbox{sub. to  }& \frac{m^2}{16}\leq \sum_{jk}(\rho_{0jk}-\rho_{1jk})^2\;;\; \sum_{ijk}\rho_{ijk}=1\;;\;\rho_{ijk}\geq 0
  \end{split}
\end{align}
where $m$ is the observed MABK violation. Now notice that for every solution $\{\rho_{0jk},\rho_{1jk}\}$ of the maximization problem, there exists another equivalent solution --i.e. that leads to the same value for $H\left(\left\{\rho_{0jk}+\rho_{1jk} \right\}\right)$-- of the form $\{\rho_{0jk}'=\rho_{0jk}+\rho_{1jk}, \rho_{1jk}'=0\}$. Therefore, we can restrict the optimization to the solutions of that form:
\begin{align}
\begin{split}
  \max_{\{\rho_{0jk}\}}\; &   H\left(\left\{\rho_{0jk}\right\}\right)\\
  \mbox{sub. to  }& \frac{m^2}{16}\leq \sum_{jk}\rho_{0jk}^2\\
  &\sum_{jk}\rho_{0jk}=1\\
  &\rho_{000}\geq\rho_{001}\geq\rho_{010}\geq\rho_{011}\geq 0,\label{Eq-OptimizationProblem}
  \end{split}
\end{align}
where we imposed the ordering of the four remaining eigenvalues $\{\rho_{0jk}\}$ without loss of generality, since the optimization problem is symmetric with respect to their permutations.

We have thus reduced the problem to the constrained maximization of $H\left(\left\{\rho_{0jk}\right\}\right)$, as described in \eqref{Eq-OptimizationProblem}. In the following calculations, we rescale the function $H\left(\left\{\rho_{0jk}\right\}\right)$ by $\ln 2$, so that it is expressed in terms of natural logarithms instead of the logarithm in base 2. This simplifies the notation when computing its derivatives but does not change the solution of the optimization problem.

We use the Karush-Kuhn-Tucker multipliers method \cite{Karush,KuhnTucker} to identify necessary conditions for extremal points of the optimization problem in \eqref{Eq-OptimizationProblem}.
The Lagrangian for our maximization problem is then given by:
\begin{align}
\begin{split}
  &\mathcal{L}(\rho_{000},\rho_{001},\rho_{010},\rho_{011}, u,v)=H\de{\rho_{000},\rho_{001},\rho_{010},\rho_{011}} \nonumber\\
  &+u\de{\rho_{000}^2+\rho_{001}^2+\rho_{010}^2+\rho_{011}^2-\frac{m^2}{16}} \nonumber\\
&+v\de{\rho_{000}+\rho_{001}+\rho_{010}+\rho_{011}-1}
\end{split}
\end{align}
The necessary conditions to have an extremal point are given by the solution of the following system:
\begin{align}
    \left\lbrace \begin{array}{l} \nabla_{\rho_{0jk}} \mathcal{L}=0\\[1ex]  \rho_{000}^2+\rho_{001}^2+\rho_{010}^2+\rho_{011}^2 \geq \frac{m^2}{16} \\[1ex]
     \rho_{000}+\rho_{001}+\rho_{010}+\rho_{011}=1\\[1ex]
     u \geq 0 \\[1ex]
     u\de{\rho_{000}^2+\rho_{001}^2+\rho_{010}^2+\rho_{011}^2-\frac{m^2}{16}} =0.
    \end{array}
    \right.  \label{KKTcond}
\end{align}
The last equation in \eqref{KKTcond} implies that either $u=0$ or the inequality constraint holds with the equal sign. Let us first consider the case that $u=0$ and compute the derivative of $\mathcal{L}$
with respect to $\rho_{0jk}$ in the first equation of \eqref{KKTcond}:
\begin{align}
   \frac{\partial \mathcal{L}}{\partial \rho_{0jk}}= -\ln{\rho_{0jk}} + v -1   = 0 \quad\,\, \forall\, \rho_{0jk}.
\end{align}
Since the logarithm is a monotonic function, the set of equations in the last expression imply one of the following cases:
\begin{itemize}
    \item[(a)] $\rho_{000}=\rho_{001}=\rho_{010}=\rho_{011}$,
    \item[(b)] $\rho_{000}=\rho_{001}=\rho_{010}$ and $\rho_{011}=0$,
    \item[(c)]$\rho_{000}=\rho_{001}$ and $\rho_{010}=\rho_{011}=0$,
    \item[(d)] $\rho_{001}=\rho_{010}=\rho_{011}=0$.
\end{itemize}
where we accounted for the border conditions, i.e. when one or more $\rho_{0jk}$ are equal to zero.

By combining the equality conditions with the constraint that $\rho_{0jk}$ sum to one, we can easily obtain the solution of the system \eqref{KKTcond} for each of the above cases. Note that the inequality constraint is still valid, therefore the derived solutions will only hold for certain values of $m$:
\begin{itemize}
    \item[(a)] $H (\{\rho_{0jk}\})=2$, valid for $m\leq 2$,
    \item[(b)] $H (\{\rho_{0jk}\})=\log 3$, valid for $m\leq 4/\sqrt{3}$
    \item[(c)]$H (\{\rho_{0jk}\})=1$, valid for $m\leq 2\sqrt{2}$
    \item[(d)] $H (\{\rho_{0jk}\})=0$, valid for $m\leq 4$.
\end{itemize}
The cases (a) and (d) are useless since the former is never valid in the range of interest for the observed violation (i.e. above the classical bound), while the latter leads to zero entropy, which is definitely not the solution of our maximization problem.

Let us consider now the case $u > 0$, which implies that the inequality constraint becomes an equality (the last equation in \eqref{KKTcond} must be satisfied). We compute the derivatives in the first equation of \eqref{KKTcond}:
\begin{align}
   \frac{\partial \mathcal{L}}{\partial \rho_{0jk}}=  2 \rho_{0jk} u- \ln{\rho_{0jk}} + v -1   = 0 \quad\,\, \forall\, \rho_{0jk}.
\end{align}
Notice that the function $g(x)= ax-\ln{x}+b$ can have at most two roots (zero points), because
\begin{align}
   g'(x)= a -\frac{1}{x},
\end{align}
has at most a single root (zero point), corresponding to one extremum for $g(x)$. It follows that there can be at most a single $y\neq x$ such that $g(x)=g(y)=0$. The potential critical points of the Lagrangian $\mathcal{L}$ are hence restricted to the following cases (remember we use the ordering $\rho_{000}\geq\rho_{001}\geq\rho_{010}\geq\rho_{011}\geq 0$)
\begin{itemize}
    \item[(i)] $\rho_{000}=\rho_{001}=\rho_{010}=\rho_{011}$,
    \item[(ii)] $\rho_{000}=\rho_{001}=\rho_{010}>\rho_{011}$,
    \item[(iii)] $\rho_{000}>\rho_{001}=\rho_{010}=\rho_{011}$,
    \item[(iv)] $\rho_{000}=\rho_{001}>\rho_{010}=\rho_{011}$.
\end{itemize}
We again account for the border conditions, and analog conditions directly follow in case some $\rho_{0jk}$ are zero:
\begin{itemize}
    \item[(v)] $\rho_{000}=\rho_{001}=\rho_{010}$ and $\rho_{011}=0$,
    \item[(vi)] $\rho_{000}=\rho_{001}>\rho_{010}$ and $\rho_{011}=0$,
    \item[(vii)] $\rho_{000}>\rho_{001}=\rho_{010}$ and $\rho_{011}=0$,
    \item[(viii)] $\rho_{000}>\rho_{001}$ and $\rho_{010}=\rho_{011}=0$,
    \item[(ix)]$\rho_{000}=\rho_{001}$ and $\rho_{010}=\rho_{011}=0$,
    \item[(x)] $\rho_{001}=\rho_{010}=\rho_{011}=0$.
\end{itemize}
Note that in all the listed cases there are a maximum of two distinct eigenvalues, which are thus completely fixed by the two equality constraints. Moreover, we observe that the cases (i), (v), (ix) and (x) correspond to the  already investigated cases (a), (b), (c) and (d), respectively.

Analysing the resulting entropy $H$ as a function of the MABK violation $m$ for each of the ten possible extremal points, we conclude that the maximum is achieved for the case (iii) for every value of $m$. In this case, the eigenvalues are fixed to:
\begin{align}
  \rho_{000}=& \frac18\left(2+\sqrt3\sqrt{m^2-4}\right)=: \nu_m\\
  \rho_{0jk}=&\frac{(1-\nu_m)}{3} \quad (j,k) \neq (0,0).
\end{align}
The solution of the optimization problem in \eqref{Eq-OptimizationProblem} then reads:
\begin{equation}
    H\left(\left\{\rho_{0jk}\right\}\right) = H\left(\left\lbrace \nu_m,\frac{1-\nu_m}{3},\frac{1-\nu_m}{3},\frac{1-\nu_m}{3} \right\rbrace\right)
\end{equation}
The lower bound on the entropy difference \eqref{Eq-EstEntropyTwo} is thus given by:
\begin{align}
&H(E|XY)-H(E)\geq \nonumber\\
&- H\left(\left\lbrace \nu_m,\frac{1-\nu_m}{3},\frac{1-\nu_m}{3},\frac{1-\nu_m}{3} \right\rbrace\right)
\end{align}
Finally we can lower bound the entropy of Alice and Bob's outcomes given Eve's quantum side information by
\begin{align}
  H(XY|E)&= 2 + H(E|XY) - H(E) \nonumber\\
  &\geq  2- H\left(\left\lbrace \nu_m,\frac{1-\nu_m}{3},\frac{1-\nu_m}{3},\frac{1-\nu_m}{3} \right\rbrace\right), \label{H(XY|E)-lowerbound}
  \end{align}
 with
 \begin{align}
  \nu_m=\frac14 +\frac{\sqrt3}{8}\sqrt{m^2-4}. \label{nu_m}
\end{align}
The r.h.s. of \eqref{H(XY|E)-lowerbound} is the bound reported in \eqref{G-bound}.

\bibliography{biblio-submission}

\begin{thebibliography}{10}

\bibitem{QTEU}
E.~Commission, ``The quantum flagship.'' \url{https://qt.eu}.
\newblock [Online].

\bibitem{QKDmarkets}
I.~Q. Technology, ``Quantum key distribution (qkd) markets: 2019-2028.''
  \url{https://www.insidequantumtechnology.com/product/quantum-key-distribution-qkd-markets-2019-2028}.
\newblock [Online].

\bibitem{BB84}
C.~H. Bennett and G.~Brassard, ``Quantum cryptography: Public key distribution
  and coin tossing,'' in {\em Proceedings of IEEE International Conference on
  Computers, Systems and Signal Processing}, pp.~175 -- 179, 1984.

\bibitem{Bruss1998}
D.~Bru\ss{}, ``Optimal eavesdropping in quantum cryptography with six states,''
  {\em Phys. Rev. Lett.}, vol.~81, pp.~3018--3021, Oct 1998.

\bibitem{E91}
A.~K. Ekert, ``Quantum cryptography based on {Bell's} theorem,'' {\em Phys.
  Rev. Lett.}, vol.~67, pp.~661--663, 1991.

\bibitem{RennerThesis}
R.~Renner, ``Security of quantum key distribution,'' {\em International Journal
  of Quantum Information}, vol.~06, no.~01, pp.~1--127, 2008.

\bibitem{Scarani2009}
V.~Scarani, H.~Bechmann-Pasquinucci, N.~J. Cerf, M.~Dusek, N.~L\"utkenhaus, and
  M.~Peev, ``The security of practical quantum key distribution,'' {\em Rev.
  Mod. Phys.}, vol.~81, pp.~1301--1350, Sep 2009.

\bibitem{Lo2014}
H.-K. Lo, M.~Curty, and K.~Tamaki, ``Secure quantum key distribution,'' {\em
  Nature Photonics}, vol.~8, no.~8, pp.~595--604, 2014.

\bibitem{Diamanti2016}
E.~Diamanti, H.-K. Lo, B.~Qi, and Z.~Yuan, ``Practical challenges in quantum
  key distribution,'' {\em npj Quantum Information}, vol.~2, no.~1, p.~16025,
  2016.

\bibitem{pir2019advances}
S.~Pirandola, U.~L. Andersen, L.~Banchi, M.~Berta, D.~Bunandar, R.~Colbeck,
  D.~Englund, T.~Gehring, C.~Lupo, C.~Ottaviani, {\em et~al.}, ``Advances in
  quantum cryptography,'' 2019.
\newblock arXiv:quant-ph/1906.01645.

\bibitem{QRNG1}
X.~Ma, X.~Yuan, Z.~Cao, B.~Qi, and Z.~Zhang, ``Quantum random number
  generation,'' {\em npj Quantum Information}, vol.~2, no.~1, p.~16021, 2016.

\bibitem{QRNG2}
M.~Herrero-Collantes and J.~C. Garcia-Escartin, ``Quantum random number
  generators,'' {\em Rev. Mod. Phys.}, vol.~89, p.~015004, Feb 2017.

\bibitem{Eppingnet1}
M.~Epping, H.~Kampermann, and D.~Bru{\ss}, ``Large-scale quantum networks based
  on graphs,'' {\em New Journal of Physics}, vol.~18, p.~053036, may 2016.

\bibitem{Eppingnet2}
M.~Epping, H.~Kampermann, and D.~Bru{\ss}, ``Robust entanglement distribution
  via quantum network coding,'' {\em New Journal of Physics}, vol.~18,
  p.~103052, oct 2016.

\bibitem{Pirker2018}
A.~Pirker, J.~Wallnöfer, and W.~Dür, ``Modular architectures for quantum
  networks,'' {\em New Journal of Physics}, vol.~20, p.~053054, may 2018.

\bibitem{Hahn2019}
F.~Hahn, A.~Pappa, and J.~Eisert, ``Quantum network routing and local
  complementation,'' {\em npj Quantum Information}, vol.~5, no.~1, p.~76, 2019.

\bibitem{lightmatter1}
V.~Krutyanskiy, M.~Meraner, J.~Schupp, V.~Krcmarsky, H.~Hainzer, and B.~P.
  Lanyon, ``Light-matter entanglement over 50 km of optical fibre,'' {\em npj
  Quantum Information}, vol.~5, no.~1, p.~72, 2019.

\bibitem{lightmatter2}
A.~Tchebotareva, S.~L.~N. Hermans, P.~C. Humphreys, D.~Voigt, P.~J. Harmsma,
  L.~K. Cheng, A.~L. Verlaan, N.~Dijkhuizen, W.~de~Jong, A.~Dr\'eau, {\em
  et~al.}, ``Entanglement between a diamond spin qubit and a photonic time-bin
  qubit at telecom wavelength,'' {\em Phys. Rev. Lett.}, vol.~123, p.~063601,
  Aug 2019.

\bibitem{satellite3}
S.-K. Liao, W.-Q. Cai, J.~Handsteiner, B.~Liu, J.~Yin, L.~Zhang, D.~Rauch,
  M.~Fink, J.-G. Ren, W.-Y. Liu, {\em et~al.}, ``Satellite-relayed
  intercontinental quantum network,'' {\em Phys. Rev. Lett.}, vol.~120,
  p.~030501, Jan 2018.

\bibitem{WEH18}
S.~Wehner, D.~Elkouss, and R.~Hanson, ``Quantum internet: A vision for the road
  ahead,'' {\em Science}, vol.~362, no.~6412, 2018.

\bibitem{Epping}
M.~Epping, H.~Kampermann, C.~Macchiavello, and D.~Bru{\ss}, ``Multi-partite
  entanglement can speed up quantum key distribution in networks,'' {\em New
  Journal of Physics}, vol.~19, p.~093012, sep 2017.

\bibitem{Grasselli_2018}
F.~Grasselli, H.~Kampermann, and D.~Bru{\ss}, ``Finite-key effects in
  multipartite quantum key distribution protocols,'' {\em New Journal of
  Physics}, vol.~20, p.~113014, nov 2018.

\bibitem{WstateProtocol}
F.~Grasselli, H.~Kampermann, and D.~Bru{\ss}, ``Conference key agreement with
  single-photon interference,'' {\em New Journal of Physics}, vol.~21,
  p.~123002, dec 2019.

\bibitem{FirstCVMDI}
Y.~Wu, J.~Zhou, X.~Gong, Y.~Guo, Z.-M. Zhang, and G.~He, ``Continuous-variable
  measurement-device-independent multipartite quantum communication,'' {\em
  Phys. Rev. A}, vol.~93, p.~022325, Feb 2016.

\bibitem{ZSG18}
Z.~Zhang, R.~Shi, and Y.~Guo, ``Multipartite continuous variable quantum
  conferencing network with entanglement in the middle,'' {\em Applied
  Sciences}, vol.~8, no.~8, 2018.

\bibitem{OLLP19}
R.~L. C.~{Ottaviani}, C.~{Lupo} and S.~{Pirandola}, ``Modular network for
  high-rate quantum conferencing,'' {\em Communications Physics}, vol.~2,
  no.~118, 2019.

\bibitem{CKAreview}
G.~Murta, F.~Grasselli, H.~Kampermann, and D.~Bruß, ``Quantum conference key
  agreement: A review,'' 2020.
\newblock arXiv:quant-ph/2003.10186.

\bibitem{CKAexperiment}
M.~Proietti, J.~Ho, F.~Grasselli, P.~Barrow, M.~Malik, and A.~Fedrizzi,
  ``Experimental quantum conference key agreement,'' 2020.
\newblock arXiv:quant-ph/2002.01491.

\bibitem{anonymousCKA}
F.~Hahn, J.~de~Jong, C.~Thalacker, B.~Demirel, S.~Barz, and A.~Pappa,
  ``Anonymous conference key agreement in quantum networks,'' 2020.
\newblock arXiv:quant-ph/2007.07995.

\bibitem{Bellineq}
J.~S. Bell, {\em Speakable and Unspeakable in Quantum Mechanics}.
\newblock Cambridge University Press, 2004.

\bibitem{YaoMayers}
A.~Yao and D.~Mayers, ``Quantum cryptography with imperfect apparatus,'' in
  {\em 2013 IEEE 54th Annual Symposium on Foundations of Computer Science},
  (Los Alamitos, CA, USA), p.~503, IEEE Computer Society, nov 1998.

\bibitem{Acin2006}
A.~Ac\'in, N.~Gisin, and L.~Masanes, ``From bell's theorem to secure quantum
  key distribution,'' {\em Phys. Rev. Lett.}, vol.~97, p.~120405, Sep 2006.

\bibitem{AcinBrunner2007}
A.~Ac\'in, N.~Brunner, N.~Gisin, S.~Massar, S.~Pironio, and V.~Scarani,
  ``Device-independent security of quantum cryptography against collective
  attacks,'' {\em Phys. Rev. Lett.}, vol.~98, p.~230501, Jun 2007.

\bibitem{PironioAcin2009}
S.~Pironio, A.~Ac\'in, N.~Brunner, N.~Gisin, S.~Massar, and V.~Scarani,
  ``Device-independent quantum key distribution secure against collective
  attacks,'' {\em New Journal of Physics}, vol.~11, p.~045021, apr 2009.

\bibitem{Masanes2011}
L.~Masanes, S.~Pironio, and A.~Ac{\'i}n, ``Secure device-independent quantum
  key distribution with causally independent measurement devices,'' {\em Nature
  Communications}, vol.~2, no.~1, p.~238, 2011.

\bibitem{VidickDIQKD}
U.~Vazirani and T.~Vidick, ``Fully device-independent quantum key
  distribution,'' {\em Phys. Rev. Lett.}, vol.~113, p.~140501, Sep 2014.

\bibitem{DIQKDsupraquantum}
S.~Pironio, L.~Masanes, A.~Leverrier, and A.~Ac\'in, ``Security of
  device-independent quantum key distribution in the bounded-quantum-storage
  model,'' {\em Phys. Rev. X}, vol.~3, p.~031007, Aug 2013.

\bibitem{Arnon-Friedman2018}
R.~Arnon-Friedman, F.~Dupuis, O.~Fawzi, R.~Renner, and T.~Vidick, ``Practical
  device-independent quantum cryptography via entropy accumulation,'' {\em
  Nature Communications}, vol.~9, no.~1, p.~459, 2018.

\bibitem{HolzRepeaters}
T.~Holz, H.~Kampermann, and D.~Bru\ss, ``Device-independent secret-key-rate
  analysis for quantum repeaters,'' {\em Phys. Rev. A}, vol.~97, p.~012337, Jan
  2018.

\bibitem{SG01}
V.~Scarani and N.~Gisin, ``Quantum communication between n partners and bell's
  inequalities,'' {\em Phys. Rev. Lett.}, vol.~87, p.~117901, Aug 2001.

\bibitem{SG_pra_01}
V.~Scarani and N.~Gisin, ``Quantum key distribution between n partners: Optimal
  eavesdropping and bell's inequalities,'' {\em Phys. Rev. A}, vol.~65,
  p.~012311, Dec 2001.

\bibitem{JeremyParityCHSH}
J.~Ribeiro, G.~Murta, and S.~Wehner, ``Reply to ``comment on `fully
  device-independent conference key agreement' '','' {\em Phys. Rev. A},
  vol.~100, p.~026302, Aug 2019.

\bibitem{Holz2019DICKA}
T.~Holz, H.~Kampermann, and D.~Bru\ss, ``A genuine multipartite bell inequality
  for device-independent conference key agreement,'' 2019.
\newblock arXiv:quant-ph/1910.11360.

\bibitem{ColbeckThesis2006}
R.~Colbeck, ``Quantum and relativistic protocols for secure multi-party
  computation,'' 2007.
\newblock PhD thesis, University of Cambridge. Also available at:
  arXiv:quant-ph/0911.3814.

\bibitem{Pironio2010}
S.~Pironio, A.~Ac\'in, S.~Massar, A.~B. de~la Giroday, D.~N. Matsukevich,
  P.~Maunz, S.~Olmschenk, D.~Hayes, L.~Luo, T.~A. Manning, {\em et~al.},
  ``Random numbers certified by bell's theorem,'' {\em Nature}, vol.~464,
  no.~7291, pp.~1021--1024, 2010.

\bibitem{Colbeck2011}
R.~Colbeck and A.~Kent, ``Private randomness expansion with untrusted
  devices,'' {\em Journal of Physics A: Mathematical and Theoretical}, vol.~44,
  p.~095305, feb 2011.

\bibitem{securityDIrandomness1}
C.~A. Miller and Y.~Shi, ``Robust protocols for securely expanding randomness
  and distributing keys using untrusted quantum devices,'' {\em J. ACM},
  vol.~63, Oct. 2016.

\bibitem{securityDIrandomness2}
S.~Pironio and S.~Massar, ``Security of practical private randomness
  generation,'' {\em Phys. Rev. A}, vol.~87, p.~012336, Jan 2013.

\bibitem{securityDIrandomness3}
S.~Fehr, R.~Gelles, and C.~Schaffner, ``Security and composability of
  randomness expansion from bell inequalities,'' {\em Phys. Rev. A}, vol.~87,
  p.~012335, Jan 2013.

\bibitem{Woodhead2018}
E.~Woodhead, B.~Bourdoncle, and A.~Ac\'in, ``Randomness versus nonlocality in
  the {M}ermin-{B}ell experiment with three parties,'' {\em {Quantum}}, vol.~2,
  p.~82, Aug. 2018.

\bibitem{NPA1}
M.~Navascu\'es, S.~Pironio, and A.~Ac\'{\i}n, ``Bounding the set of quantum
  correlations,'' {\em Phys. Rev. Lett.}, vol.~98, p.~010401, Jan 2007.

\bibitem{NPA2}
M.~Navascu{\'e}s, S.~Pironio, and A.~Ac\'in, ``A convergent hierarchy of
  semidefinite programs characterizing the set of quantum correlations,'' {\em
  New Journal of Physics}, vol.~10, p.~073013, jul 2008.

\bibitem{Nieto2014}
O.~Nieto-Silleras, S.~Pironio, and J.~Silman, ``Using complete measurement
  statistics for optimal device-independent randomness evaluation,'' {\em New
  Journal of Physics}, vol.~16, p.~013035, jan 2014.

\bibitem{Bancal2014}
J.-D. Bancal, L.~Sheridan, and V.~Scarani, ``More randomness from the same
  data,'' {\em New Journal of Physics}, vol.~16, p.~033011, mar 2014.

\bibitem{CHSH}
J.~F. Clauser, M.~A. Horne, A.~Shimony, and R.~A. Holt, ``Proposed experiment
  to test local hidden-variable theories,'' {\em Phys. Rev. Lett.}, vol.~23,
  pp.~880--884, Oct 1969.

\bibitem{WernerWolf}
R.~F. Werner and M.~M. Wolf, ``All-multipartite bell-correlation inequalities
  for two dichotomic observables per site,'' {\em Phys. Rev. A}, vol.~64,
  p.~032112, Aug 2001.

\bibitem{Mermin}
N.~D. Mermin, ``Extreme quantum entanglement in a superposition of
  macroscopically distinct states,'' {\em Phys. Rev. Lett.}, vol.~65,
  pp.~1838--1840, Oct 1990.

\bibitem{Ardehali}
M.~Ardehali, ``Bell inequalities with a magnitude of violation that grows
  exponentially with the number of particles,'' {\em Phys. Rev. A}, vol.~46,
  pp.~5375--5378, Nov 1992.

\bibitem{BK93}
A.~V. Belinskiĭ and D.~N. Klyshko, ``Interference of light and bell's
  theorem,'' {\em Phys. Rev. A}, vol.~36, pp.~653--693, 1993.

\bibitem{HHH95}
R.~Horodecki, P.~Horodecki, and M.~Horodecki, ``Violating bell inequality by
  mixed spin-12 states: necessary and sufficient condition,'' {\em Physics
  Letters A}, vol.~200, no.~5, pp.~340 -- 344, 1995.

\bibitem{MABKrecursion}
D.~Collins, N.~Gisin, S.~Popescu, D.~Roberts, and V.~Scarani, ``Bell-type
  inequalities to detect true $\mathit{n}$-body nonseparability,'' {\em Phys.
  Rev. Lett.}, vol.~88, p.~170405, Apr 2002.

\bibitem{JeremyMABK}
J.~Ribeiro, G.~Murta, and S.~Wehner, ``Fully device-independent conference key
  agreement,'' {\em Phys. Rev. A}, vol.~97, p.~022307, Feb 2018.

\bibitem{SS19}
M.~A. Siddiqui and S.~Sazim, ``Tight upper bound for the maximal expectation
  value of the mermin operators,'' {\em Quantum Information Processing},
  vol.~18, p.~131, Mar 2019.

\bibitem{POVM-not-optimal}
T.~V\'ertesi and E.~Bene, ``Two-qubit bell inequality for which positive
  operator-valued measurements are relevant,'' {\em Phys. Rev. A}, vol.~82,
  p.~062115, Dec 2010.

\bibitem{OptimalProjective-probabilityBellineq}
Y.-C. Liang and A.~C. Doherty, ``Bounds on quantum correlations in
  bell-inequality experiments,'' {\em Phys. Rev. A}, vol.~75, p.~042103, Apr
  2007.

\bibitem{neglects-identity-Bellviolation1}
M.~Zukowski and C.~Brukner, ``Bell's theorem for general n-qubit states,'' {\em
  Phys. Rev. Lett.}, vol.~88, p.~210401, May 2002.

\bibitem{neglects-identity-Bellviolation2}
W.~Laskowski, T.~Paterek, M.~\ifmmode~\dot{Z}\else \.{Z}\fi{}ukowski, and
  i.~c.~v. Brukner, ``Tight multipartite bell's inequalities involving many
  measurement settings,'' {\em Phys. Rev. Lett.}, vol.~93, p.~200401, Nov 2004.

\bibitem{neglects-identity-Bellviolation3}
M.~Li and S.-M. Fei, ``Bell inequalities for multipartite qubit quantum systems
  and their maximal violation,'' {\em Phys. Rev. A}, vol.~86, p.~052119, Nov
  2012.

\bibitem{EAT}
F.~{Dupuis} and O.~{Fawzi}, ``Entropy accumulation with improved second-order
  term,'' {\em IEEE Transactions on Information Theory}, vol.~65,
  pp.~7596--7612, Nov 2019.

\bibitem{DIRE-framework}
P.~J. {Brown}, S.~{Ragy}, and R.~{Colbeck}, ``A framework for quantum-secure
  device-independent randomness expansion,'' {\em IEEE Transactions on
  Information Theory}, vol.~66, no.~5, pp.~2964--2987, 2020.

\bibitem{Carrara2020}
G.~Carrara, H.~Kampermann, D.~Bruß, and G.~Murta, ``Genuine multipartite
  entanglement is not a precondition for secure conference key agreement,''
  2020.
\newblock arXiv:quant-ph/2007.11553.

\bibitem{HolzComment}
T.~Holz, D.~Miller, H.~Kampermann, and D.~Bru\ss, ``Comment on ``fully
  device-independent conference key agreement'','' {\em Phys. Rev. A},
  vol.~100, p.~026301, Aug 2019.

\bibitem{GHZstabilizers}
G.~T\'oth and O.~G\"uhne, ``Entanglement detection in the stabilizer
  formalism,'' {\em Phys. Rev. A}, vol.~72, p.~022340, Aug 2005.

\bibitem{QuantumAEP}
M.~{Tomamichel}, R.~{Colbeck}, and R.~{Renner}, ``A fully quantum asymptotic
  equipartition property,'' {\em IEEE Transactions on Information Theory},
  vol.~55, pp.~5840--5847, Dec 2009.

\bibitem{Masanes06}
L.~Masanes, ``Asymptotic violation of bell inequalities and distillability,''
  {\em Phys. Rev. Lett.}, vol.~97, p.~050503, Aug 2006.

\bibitem{Karush}
W.~Karush, ``Minima of functions of several variables with inequalities as side
  conditions,'' 1939.

\bibitem{KuhnTucker}
H.~W. Kuhn and A.~W. Tucker, ``Nonlinear programming,'' in {\em Proceedings of
  the Second Berkeley Symposium on Mathematical Statistics and Probability},
  (Berkeley, Calif.), pp.~481--492, University of California Press, 1951.

\end{thebibliography}

\end{document}